\documentclass[a4paper,12pt]{amsart}

\usepackage{kantlipsum} 
\calclayout{}

\usepackage[ruled,vlined]{algorithm2e}
\usepackage{amsfonts}
\usepackage{amsmath}    
\usepackage{amssymb}    
\usepackage{bm}         
\usepackage{bbm}
\usepackage{color}
\usepackage{enumitem}
\usepackage{graphicx}   
\usepackage{natbib}    
\usepackage{epsf}
\usepackage{lscape}
\usepackage{enumitem}
\usepackage{xr}
\usepackage{longtable}
\bibpunct{(}{)}{;}{a}{,}{,}

\usepackage[normalem]{ulem}

\usepackage{subcaption}

\newtheorem{thm}{Theorem}
\newtheorem{lemma}{Lemma}

\newtheorem{prop}{Proposition}

\newtheorem{defn}{Definition}

\begin{document}

  \title{\bf Semi-parametric local variable selection under misspecification}
  \author{David Rossell, 
    Pompeu Fabra University\\
    and \\
    Arnold Kisuk Seung, 
    University of California at Irvine \\
    and \\
    Ignacio Saez, 
    Mount Sinai \\
    and \\
    Michele Guindani \\
    University of California at Los Angeles}
  \maketitle

\bigskip
\begin{abstract}
Local variable selection aims to  test for the effect of covariates on an outcome within specific regions. We outline a challenge that arises in the presence of non-linear effects and model misspecification. Specifically, for common semi-parametric methods even slight model misspecification can result in a high false positive rate, in a manner that is highly sensitive to the chosen basis functions. We propose a methodology based on orthogonal cut splines that avoids false positive inflation for any choice of knots,  and achieves consistent local variable selection. Our approach offers simplicity, handles both continuous and categorical covariates, and provides theory for high-dimensional covariates and model misspecification. We discuss settings with either independent or dependent data. Our proposal allows including adjustment covariates that do not undergo selection, enhancing the model's flexibility. 
Our examples describe salary gaps associated with various discrimination factors at different ages, and the effects of covariates on functional data measuring brain activation at different times.
\end{abstract}

\noindent%
{\it Keywords:}  local null testing, semi-parametric model, Bayesian model selection, Bayesian model averaging, functional data
\vfill

Local variable selection, or local null hypothesis testing,  is an important problem in many fields. Given  a set of covariates $x$, the task is testing whether they have an effect on an outcome $y$, at specific values indicated by additional covariates $z$.  For example, we assess whether
disparities in salary ($y$) are associated to covariates such as race and gender ($x$) at different ages ($z$). Although there may be no evidence for disparities in individuals who recently entered the workforce (e.g., due to newly implemented policies), said disparities may exist at other career stages. 
Importantly, the primary objective in these examples is testing (multiple) scientific hypotheses indexed by $z$, not merely estimating covariate effects.  Our examples illustrate how estimation-focused methods may be prone to false positive inflation.  

Additive regression is a natural semi-parametric framework for analyzing local covariate effects, particularly when the sample size $n$ or computational power are not sufficient for fully non-parametric methods, or one wants simpler models to facilitate interpretation. While additive regression models are standard, their application to local variable selection has not been well-studied. We contribute primarily in two ways.
First, we highlight an important drawback: even under slight model misspecification, most additive regression methods present an inflated type I error for local effects.
 This issue is sensitive to the chosen basis functions,  and raises a significant concern for local variable selection. 
Second, we address this pitfall through a simple modification that preserves the interpretability and computational features of additive regression. We give theoretical and empirical results that our approach achieves consistent local variable selection when the model is misspecified. To our knowledge, ours are the first results for high-dimensional local variable selection under misspecification.

Figure \ref{fig:splinefit} offers an illustration with
a binary covariate $x$, defining two groups (see Section \ref{sec:simulation_iid}). The true group means are equal for $z \leq 0$ and different for $z>0$. We seek to identify the $z$'s at which the group means differ. We fit cubic B-splines to the combined data from both groups ($n=1,000$). 
The solid black lines in the left panel are the group means projected on the B-spline basis, i.e.  the approximation recovered as $n \rightarrow \infty$. Critically, said projections are no longer equal for $z \leq 0$. 
Their differences are small, but as $n$ grows significant type I error issues arise. The right panel shows that a near-one posterior probability is assigned to group differences for $z \in (-1,0)$ (similar issues occur when using posterior intervals or P-values). 
Here we employed moderately many knots, 20 for the baseline mean and 10 for the group differences. 
 Figure \ref{fig:splinefit_varyingknots} shows that type I errors also occur for other knots and sample sizes (bottom left panel).
Our proposal is more robust in that it prevents type I errors, regardless of the knot placement. 
It can hence operate with relatively few knots, 
which simplifies computations (the model space grows exponentially with the knots),
and one may use Bayesian model averaging to learn how many knots are needed.

\begin{figure}
\begin{center}
\begin{tabular}{cc}
\includegraphics[width=0.48\textwidth, height=0.42\textwidth]{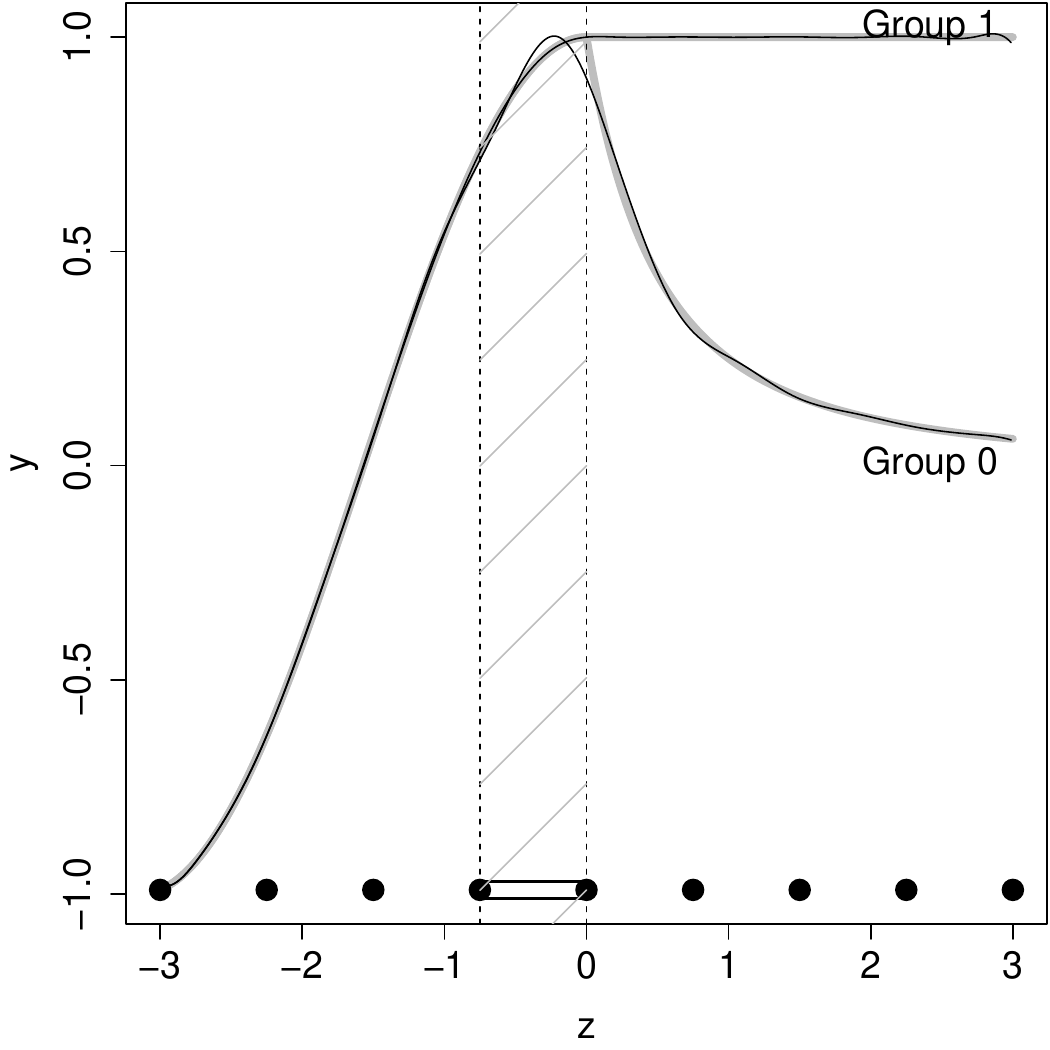} &
\includegraphics[width=0.48\textwidth, height=0.42\textwidth]{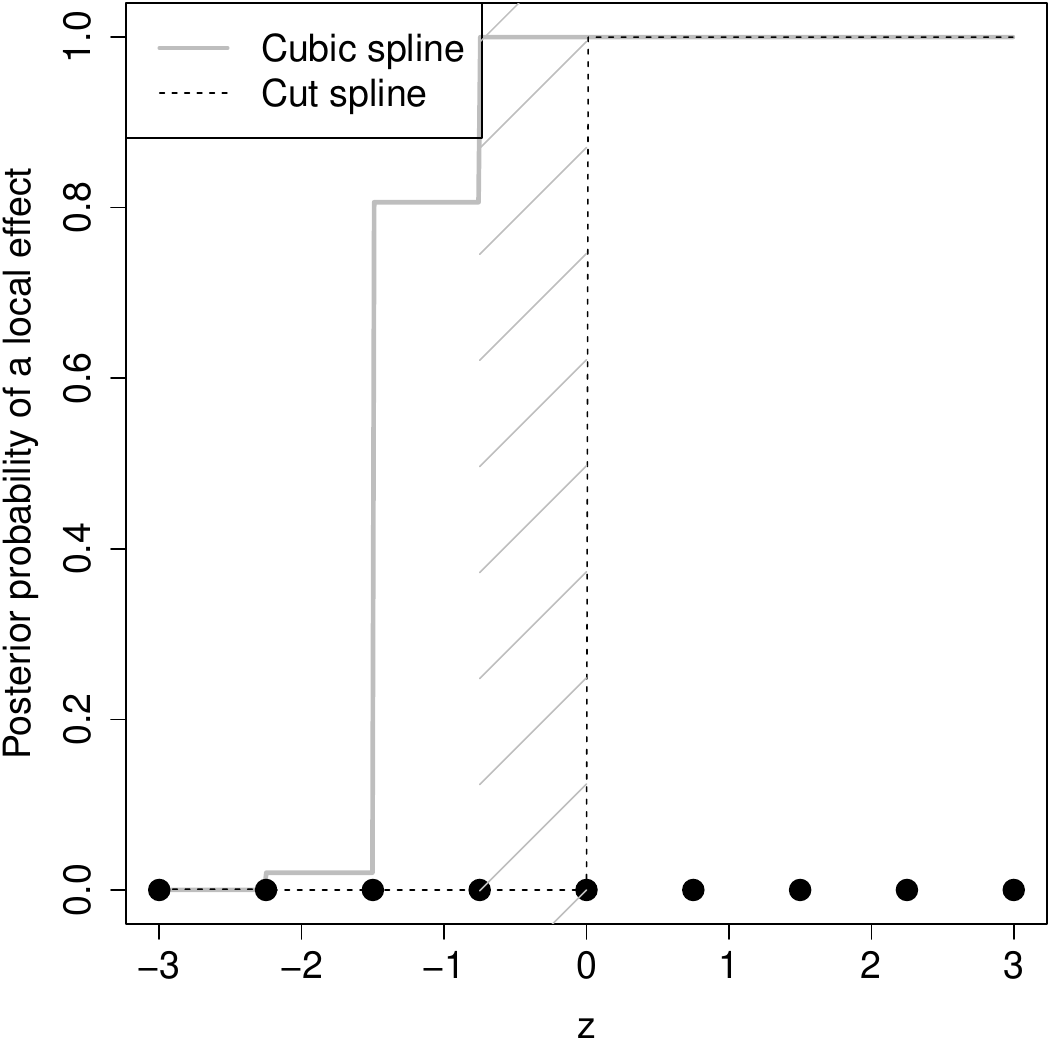} \\
\end{tabular}
\end{center}
\caption{Simulation. Left: true group means and their cubic B-spline projections (left). Right: posterior probability of local group differences, using 9 knots (black dots) for cubic splines and for cut cubic B-splines based on $n=1000$.}
\label{fig:splinefit}
\end{figure}

We consider settings with independent data 
and some extensions to dependent data. The latter setting received considerable attention.
\cite{boehm:2015} addressed local variable selection by combining spike-and-slab priors with latent Gaussian processes, \cite{jhuang:2019} by using horseshoe priors, and \cite{kang:2018} by applying soft-thresholding to a Gaussian process. For functional data on a common grid (e.g., mass spectrometry, images), \cite{morris:2011} 
employed basis functions such as wavelets, Fourier transforms, and splines. 
These methods test whether covariate effects are near 0, rather than exactly 0 as is our goal here. There is also abundant research in settings where $z$ is categorical. For example, \cite{smith_michael:2007} proposed a Bayesian framework for regressions on a lattice, whereas \cite{scheel:2013} 
focused on variable selection over discrete spatial locations.
 Similar problems have been addressed (\cite{warren2012,wilson2017}) to estimate pollutant effects during pregnancy. These models focus on estimation, thus lacking a formal testing framework, as we do here. 

In work related to ours, \cite{deshpande:2020} employed Bayesian additive regression trees for varying-coefficient models, 
and established near-minimax estimation rates when the number of covariates $p$ grows slower than $n$.
 However, their work did not provide theory for local null tests. Furthermore,  \cite{paulon:2023}  considered a non-parametric framework for time-varying variable selection with categorical covariates, to detect high-order covariate interactions. They proved parameter estimation and variable selection consistency for fixed $p$. Here we consider a simpler additive model that is applicable beyond longitudinal data, allows for continuous and  categorical  covariates and multivariate $z$, and provides high-dimensional theory under misspecification.

The key in our approach lies in using basis functions that capture covariate effects, locally for each neighborhood in $z$,
and are orthogonal to basis used outside that neighborhood. We refer to this construction as an {\it orthogonal cut basis}. It can be viewed as an extension of a piece-wise constant parameterization that maintains certain orthogonality properties and allows including higher-degree terms (e.g., cubic, as in Figure \ref{fig:splinefit}). In dependent data settings, an additional element is required: a block-diagonal approximation to the covariance that captures local dependence. 

We set notation. Let $y = (y_1, \ldots, y_n) \in \mathbb{R}^n$ be the observed outcomes,
$x_i= (x_{i1},\ldots,x_{ip})$ the $p$ covariates for individual $i$, and $z_i \in \mathbb{R}^d$ the ``coordinates" of interest. 
For example $y_i$ is the salary of individual $i$, $z_i$ the age, and $x_i$ records race and gender.
The goal is to evaluate the effect of $x_i$ on $y_i$ at each possible $z_i$. 
Our theory and methods also hold when \(z_i\) is multivariate,  see Section \ref{ssec:simulation_iid_bivar} and Figures \ref{fig:bspline_fit_cubic_bivar}-\ref{fig:simiid_bivar} for bivariate examples.
We note however that multivariate functional data requires
specifying suitable dependence models, which falls beyond our scope.  
Additionally to $z_i$, there may be adjustment covariates included in the model without undergoing testing. For simplicity we omit said covariates from our exposition, but they are implemented in our software and used in our salary example. 
We denote the (unknown) data-generating density of $y$ as $F$, which generally lies outside the assumed model class. For purposes of local null testing, it is convenient to partition the support of $z$ into a set of regions $\mathcal{R}= \{R_1,\ldots,R_{|\mathcal{R}|}\}$ (e.g., the 9 intervals in Figure \ref{fig:splinefit}). 
The idea is that these are the smallest regions that are of practical interest. We later discuss that, to improve power  and reduce sensitivity to the chosen knots,  one may consider multiple sets of regions $\mathcal{R}_1,\ldots,\mathcal{R}_L$ and combine them in a multi-resolution analysis via Bayesian model averaging. 
However, for clarity we consider a single $\mathcal{R}$ for now.
 Our method does not require that knots are perfectly aligned with the $z$ values where covariate effects transition from zero to non-zero. If a covariate truly has no effect for all $z$ in a region, then there is no effect, otherwise there is an effect for some $z$ in that region that will be asymptotically detected by our approach. Our multi-resolution analysis offers a way to increase the precision in certain regions, if so desired, as we illustrate in our simulations. Our framework applies to regions of arbitrary shape, e.g. defined by brain regions, but our software focuses on hyper-rectangles that lead to simpler interpretation (akin to tree-based methods). 
Our methods are implemented in the R package mombf.



\section{Local variable selection in group comparisons}
\label{sec:single_covar}

\subsection{Local selection in model-based tests}
\label{ssec:model}

We consider first the case of a single categorical covariate $x_{i1} \in \{1,\ldots,K\}$, defining $K$ distinct groups. 
 Continuous and discrete covariates taking infinitely many values, and multiple covariates, are discussed in Section \ref{sec:multiple_covar}. 
Suppose that the data-generating distribution $F$ has an expectation given by a baseline function $f_0(z_i)$ plus an interaction term $f_1(x_{i1}, z_i)$. 
Specifically,
\begin{align}
E_F(y_i \mid z_i, x_i) = f_0(z_i) + f_1(x_{i1}, z_i), \quad i=1,\ldots,n,
\label{eq:datagentruth_1covar}
\end{align}
where  $f_0$ and $f_1$ are  continuous functions in $z_i$.   An identifiability condition is required in \eqref{eq:datagentruth_1covar}, e.g. 
 a sum-to-zero constraint $\sum_{k=1}^K f_1(k,z)=0$ at each $z$,
so that 
$f_1(k,z)$ represents the deviation from the baseline mean for group $k$ at $z$.
The goal of local variable selection is to assess whether the groups have a zero effect at a specific value of $z$. This is expressed by the local null hypothesis
\begin{align}
f_1(1,z) = f_1(2,z)= \ldots = f_1(K,z)= 0.
\label{eq:localnull_1covar}
\end{align}
In practice, the data-generating $F$, and the functions $(f_0,f_1)$, are unknown and need to be approximated. 
A common strategy is to model $f_0$ and $f_1$ using basis functions (such as splines) and to assume a specific error model for the outcome $y=(y_1,\ldots,y_n)$, say Gaussian.
Specifically, consider the approximation  of $f_0(z_i)$ with $\beta_0(z_i) = w_{i0}^T \eta_0$ and $f_1(k,z_i)$ with $\beta_{1k}(z_i) = w_{ik}^T \eta_{1k}$, where $w_{i0}=w_{i0}(z_i) \in \mathbb{R}^{l_0}$ and $w_{ik}=w_{ik}(z_i) \in \mathbb{R}^{l_1}$ evaluate the chosen basis at $z_i$, $(\eta_0, \eta_{1k})$ denote the corresponding parameters, and $(l_0,l_1)$ their dimension.
Then,  \eqref{eq:datagentruth_1covar} is approximated by
\begin{align}
E_\eta(y_i \mid z_i, x_{i1}=k)= \beta_0(z_i) + \beta_{1 k}(z_i)= w_{i0}^T \eta_0 + w_{i k}^T \eta_{1 k}.
\label{eq:vcmodel_1covar}
\end{align}
Equation \eqref{eq:datagentruth_1covar} is a varying coefficient model \citep{hastie:1993}, where the regression coefficients $\beta_0(z_i)$ and $\beta_{1k}(z_i)$ depend on $z_i$.
Assuming Gaussian errors, we can express \eqref{eq:vcmodel_1covar} as
\begin{align}
 y \mid Z, X \sim N \left( W_0 \eta_0 + W_1 \eta_1, \Sigma \right),
\label{eq:vcmodel_1covar_matrix}
\end{align}
where $W_0$ is an $n \times l_0$ matrix with $i^{th}$ row given by $w_{i0}$, $W_1$ is an $n \times l_1 K$ matrix with $i^{th}$ row given by $(w_{i1}^T,\ldots,w_{iK}^T)$, and $\Sigma$ is the error covariance. For instance, one may assume independence and take $\Sigma = \sigma^2 I$ for some $\sigma^2 > 0$, or consider an appropriate dependence model.
For identifiability, we impose an orthogonality constraint $W_1^T W_0 = 0$ such that $W_0\eta_0$ represents the  baseline mean and $W_1\eta_1$ represents  group deviations from the baseline. This is achieved by defining a standard basis (e.g. splines) and letting $W_1$ be the residuals from regressing that basis onto $W_0$, see Section \ref{ssec:misspec}.
For simplicity, we assume that $W$ is non-random, meaning that the covariates $X$ and coordinates $Z$ are fixed, and that any subset of $W$ with $\leq n$ columns has full column-rank.

Under the model specified in \eqref{eq:vcmodel_1covar}, the local null hypothesis at $z$ can be expressed as 
\begin{align}
 \beta_{11}(z)= \ldots= \beta_{1K}(z)=0 \Longleftrightarrow 
w_1^T \eta_{11}= \ldots = w_K^T \eta_{1K}=0,
\label{eq:localnull_1covar_vcmodel}
\end{align}
where $w_k=w_k(z)$ is the basis function of $\beta_{1k}$ evaluated at $z$.
That is, the local null test at a given $z$ assesses whether particular linear combinations of the parameters $\eta_1$ are zero.

To facilitate interpretation, we use a local basis such that conducting the test for any $z$ within a given region $R_b$
is defined by finding zeroes in $\eta_{1k}$ (rather than its linear combinations). 
By local, we mean that the value of $\beta_{1k}(z)$ in each region $R_b$ is defined by a small subset of parameters in $\eta_{1k}$.  For instance, B-spline bases have minimal support: in an $l$-degree B-spline, each interval between consecutive knots is represented by $l+1$ coefficients.
 For multivariate $z$, we use tensor products of B-splines. 
Figure \ref{fig:splinefit_suppl} (left) illustrates a cubic B-spline ($l=3$) in which a local hypothesis is tested for $z \in [-0.75, 0]$ (black square) with $K=2$ groups. 
Only $l+1=4$ bases (marked in black) have nonzero values in this interval, specifically bases 2-5. That is, if $\eta_{12}=\eta_{13}=\eta_{14}=\eta_{15}=0$ it then follows that $\beta_{11}(z)= 0$ for all $z \in [-0.75, 0]$, where we recall that the deviation of group $K=2$ from the baseline is $\beta_{12}(z) = -\beta_{11}(z)$, from the identifiability constraint $W_1^T W_0=0$ discussed above.

\subsection{Model Misspecification and orthogonal cut basis}
\label{ssec:misspec}

One of our main messages is that model-based tests, such as \eqref{eq:localnull_1covar_vcmodel}, can lead to incorrect conclusions when the model is misspecified, i.e. \eqref{eq:vcmodel_1covar_matrix} does not perfectly capture the true mean \eqref{eq:datagentruth_1covar}.
Under mild conditions, likelihood-based estimators $\hat{\eta}=(\hat{\eta}_0,\hat{\eta}_1)$ for the parameters $\eta=(\eta_0,\eta_1)$ in \eqref{eq:vcmodel_1covar_matrix} converge to the optimal values that minimize the Kullback-Leibler divergence. 
For the independent errors case  ($\Sigma= \sigma^2 I$  in \eqref{eq:vcmodel_1covar_matrix}), the asymptotic value minimizes mean squared prediction error under the data-generating $F$, and is given by a least-squares regression of $E_F(y \mid Z,X)$ on $W$,
\begin{align}
\eta^*= \arg\min_\eta E_F \left( (y - W\eta)^T (y - W\eta) \mid Z,X \right)= (W^T W)^{-1} W^T E_F (y \mid Z,X).
\label{eq:asympsolution_indep}
\end{align}
The issue is that these optimal values in $\eta^*$ associated with a specific region may be non-zero even when the group means are equal in that region. In Figure \ref{fig:splinefit},
 although the group means are equal for all $z \in [-0.75,0]$, the optimal coefficients $\eta_{12}^*,\ldots,\eta_{15}^*$ associated to this region are non-zero.
 As $n \rightarrow \infty$, standard frequentist and Bayesian tests  provide overwhelming evidence for a group effect at such $z$ values, resulting in false positives. 
This occurs because $\eta_{12}^*,\ldots,\eta_{15}^*$ also play a role in approximating the group means for $z \not\in [-0.75,0]$, e.g. by setting non-zero entries in $\eta_1^*$ one obtains a better approximation to the true group means for $z>0$.

To address this issue, we introduce the concept of an \emph{orthogonal cut basis}. This is a basis for $\beta_{1k}(z)$ that has support only in a region of $z$ values and is orthogonal to bases outside that region.
Then, the optimal coefficients $\eta_{1}^*$ in \eqref{eq:asympsolution_indep}  contain zeroes when there are truly no local group differences, e.g. such as $z < 0$ in Figure \ref{fig:splinefit}. 
Figure \ref{fig:splinefit_suppl} (right) shows a cut cubic B-spline basis. 
The cut B-spline basis for region $R_b$ is equal to the B-spline basis in $R_b$ and to 0 outside $R_b$, hence it does not contribute to the group means outside $R_b$. 
Cut bases, similar to trees, set discontinuities in the estimated group differences. The baseline mean $\beta_0(z)$ can still be assumed to be continuous. Using B-splines for the baseline mean and cut B-splines for covariate effects combines desirable features of continuous semi-parametric and discontinuous non-parametric methods.

Specifically, we partition the support of $z$ into $|\mathcal{R}|$ regions and select a basis $W=(W_0,W_1)$ that satisfies two conditions.
Let $(y_b, W_{0b}, W_{1b})$ indicate the rows in $(y,\, W_0,\, W_1)$ for the observations in region $b$ (i.e. $z_i \in R_b$).
The first condition is to employ a cut basis $W_1$ for group effects $\beta_{1k}(z)$ within each region $b$.
Then, $(W_0,W_1)$ can be written as
\begin{align} 
W_0= \begin{pmatrix} W_{01} \\ W_{02} \\ \ldots \\ W_{0|\mathcal{R}|} \end{pmatrix}
; \hspace{5mm}
W_1= \begin{pmatrix} 
W_{11} & 0 & \ldots &  0 \\
0 & W_{12} & \ldots & 0 \\
\ldots & & & \\
0 & 0 & \ldots & W_{1|\mathcal{R}|}
\end{pmatrix}.
\label{eq:blockdiag_basis}
\end{align}
The second condition is that $W_{1b}^T W_{0b}=0$ (orthogonality), and is satisfied as follows. Let $\widetilde{W}_{1b}$ be a cut basis 
such that $\widetilde{W}_{1b}^T W_{0b} \neq 0$, then we set
$
W_{1b}= (I - W_{0b} (W_{0b}^T W_{0b})^{-1} W_{0b}) \widetilde{W}_{1b},
$
which are the residuals obtained from regressing $\widetilde{W}_{1b}$ onto $W_{1b}$. Consequently,  $W_{1b}^T W_{0b}=0$.
We refer to  any $W_1$ satisfying \eqref{eq:blockdiag_basis} and $W_{1b}^T W_{0b}=0$ as an {\it orthogonal cut basis}.

Lemma \ref{lem:zero_indepmodel} below shows that the asymptotic $\eta^*$ contains zeroes. Specifically $\eta_{1b}^*$, which quantifies the group differences in region $b$,  is obtained by regressing the true mean $E_F(y_b \mid X,Z)$ onto $W_{1b}$ (see Section \ref{supplsec:proof_zero_indepmodel} for the proof). 
Therefore, if $E_F(y_b \mid X,Z)$ is not linearly associated with $W_1$ in region $b$ (the $K$ data-generating group means are equal in that region), it follows that $\eta_{1b}^*=0$.

\begin{lemma}
Consider $\eta^*=(\eta_0^*,\eta_1^*)$ in \eqref{eq:asympsolution_indep}, where $W=(W_0,W_1)$, $W_0$ is an $n \times l_0$ matrix and $W_1$ an $n \times l_1$ an orthogonal cut basis as in \eqref{eq:blockdiag_basis} satisfying $W_0^T W_1=0$.
Let $\eta_1^*= (\eta_{11},\ldots,\eta_{1|\mathcal{R}|}^*)$ where $\eta_{1b} \in \mathbb{R}^{l_b}$ are the parameters associated to $W_{1b}$.
Then, $ \eta_{1b}^*= (W_{1b}^T W_{1b})^{-1} W_{1b}^T E_F(y_b \mid X,Z)$.
\label{lem:zero_indepmodel}
\end{lemma}

\subsection{Functional data}
\label{ssec:misspec_fda}

The presence of dependence, as in functional or dense longitudinal data, introduces additional challenges. 
Specifically, when assuming a given $\Sigma \neq \sigma^2 I$, the asymptotic solution becomes
\begin{align}
 \widetilde{\eta}^*&= \arg\min_\eta E_F \left( (y - W\eta)^T \Sigma^{-1} (y - W\eta) \mid Z,X \right)=
(W^T \Sigma^{-1} W)^{-1} W^T \Sigma^{-1} E_F(y \mid Z,X)
\nonumber \\
&=\eta^* + (W^T \Sigma^{-1} W)^{-1} W^T \Sigma^{-1} \left( E_F(y \mid Z,X) - W\eta^* \right),
\label{eq:asympsolution_dep}
\end{align}
where $\eta^*$ is defined as in \eqref{eq:asympsolution_indep}. Unless the second term in \eqref{eq:asympsolution_dep} is zero, $\tilde{\eta}^* \neq \eta^*$,
meaning that the optimal estimator for independent and dependent data are not equal.
Hence, even if $\eta^*$ contains zeroes -- as guaranteed by Lemma \ref{lem:zero_indepmodel}  when using orthogonal cut bases --  $\tilde{\eta}^*$ may not contain zeroes.
Intuitively, under correlated errors,  the optimal parameters for one region depend on those from other regions. To address this issue, we employ a block-diagonal approximation to $\Sigma$ that accounts only for within-region dependence.

 We consider a setting where one observes $M$ functions, each of which is evaluated at a common number of $z$ values.
Our model assumes independence across functions and across regions ($\Sigma_{ij}=0$ if observations $(i,j)$ are in different regions). 
Block-diagonal approximations are often reasonable for functional data, as long-range dependence tends to be weak,
see \cite{varin:2011}  (Section 3.1) for related composite likelihood ideas.

In the case of functional data, 
Lemma \ref{lem:zero_depmodel} shows that using such a block-diagonal $\Sigma$ leads to an asymptotic $\tilde{\eta}_1^*$ in \eqref{eq:asympsolution_dep} that contains zeroes when a covariate has no local effects. 
Specifically, if $E_F(y_b \mid X,Z)$ does not differ across groups in region $b$ then it is linearly independent of $W_{1b}$ and hence also of $\Sigma_b^{-1} W_{1b}$,
and then  $\tilde{\eta}_{1b}^*=0$ by Lemma \ref{lem:zero_depmodel}.

In our examples with univariate $z \in \mathbb{R}$, we use a simple parametric covariance, 
using either a first-order autoregressive or moving average model, selecting the preferred model with the Bayesian information criterion. 
Subsequently, we perform inference assuming that $\Sigma = \hat{\Sigma}$ is known. 
This is convenient in that one can take as a working model
\begin{align}
\tilde{y} \mid Z, X, \hat{\Sigma} \sim N \left( \widetilde{W}_0 \eta_0 + \widetilde{W}_1 \eta_1, \sigma^2 I \right),
\label{eq:vcmodel_1covar_fda}
\end{align}
where $\tilde{y}= \hat{\Sigma}^{-1/2} y$, $\widetilde{W}_0= \hat{\Sigma}^{-1/2} W_0$ and $\widetilde{W}_1= \hat{\Sigma}^{-1/2} W_1$.
Our strategy allows using standard Bayesian algorithms for independent errors on $\tilde{y}$, including model search.
 One could also use more flexible covariance models, particularly for multivariate $z$, and treat $\Sigma$ as unknown. Computing marginal likelihoods gets much costlier, however, see Section \ref{sec:discussion}.


\section{Multiple covariates}
\label{sec:multiple_covar}

We extend our approach to settings involving a covariate vector $x_i = (x_{i1},\ldots,x_{ip})$, where the data-generating truth in \eqref{eq:datagentruth_1covar} is expanded to
$E(y_i \mid z_i, x_i) = f_0(z_i) + f_1(x_i, z_i)$.
Similar to \eqref{eq:datagentruth_1covar}, we set an identifiability constraint so that $f_0(z_i)$ represents the baseline mean and $f_1(x_i,z_i)$ is a deviations from the mean. Then,  the local null test for covariate $j$ at $z$ corresponds to assessing whether $f_1(x_i,z)=0$ depends on $x_{ij}$ or not. To impose this constraint, we orthogonalize the basis used for $f_1$ with respect to that for $f_0$, as explained below.

As discussed in relation to \eqref{eq:vcmodel_1covar}, model-based approaches for local variable selection approximate $(f_0,f_1)$ using a suitable function class and then expressing the selection in terms of the parameters of that class. While one may model $f_1(x,z)$ non-parametrically, and assess the effect of $x$ at each given $z$, the computational burden increases sharply as $p$ or the number of possible $z$ values increase (the model space grows exponentially with the number of parameters). Instead, we consider a semi-parametric model that assumes additive covariate effects,
\begin{align}
 E_{\eta}(y_i \mid z_i, x_i)= \beta_0(z_i) + \sum_{j=1}^p \, \beta_{1j}(z_i) x_{ij}= w_{i0}^T \eta_0 + \sum_{j=1}^p \, w_{ij}^T \eta_{1j}.
\label{eq:vcmodel_additive}
\end{align}
Analogously to \eqref{eq:vcmodel_1covar}, $\beta_0(z_i)= w_{i0}^T \eta_0$ approximates $f_0(z_i)$,
$x_{ij} \beta_{1j}(z_i)= w_{ij}^T \eta_{1j}$ provides an additive approximation to $f_1(x_i,z_i)$,
$w_{i0}=w_{i0}(z_i) \in \mathbb{R}^{l_0}$ and $w_{ij}=w_{ij}(z_i) \in \mathbb{R}^{l_1}$ are basis expansions computed from $(z_i,x_i)$,
 and $(\eta_0,\eta_{11},\ldots,\eta_{1p}) \in \mathbb{R}^{l_0 + l_1 p}$ are the corresponding parameters.
We denote the total number of parameters as $q= l_0 + l_1p$.

The data-generating expectation $E_F(y_i \mid z_i,x_i)$ may not be additive, e.g. the effect of covariate $j$ at $z$ may interact with that of covariate $j'$. However, the parameters in the assumed model \eqref{eq:vcmodel_additive} remain interpretable, e.g. $\beta_{1j}(z)$ is the mean effect of covariate $j$ at coordinate $z$.
 If covariate $j$ truly has a non-linear effect (conditional on $z$), then $\beta_{1j}(z)$ asymptotically recovers the slope of the best linear approximation, i.e. the average effect for a unit increase in covariate $j$. 
 Model \eqref{eq:vcmodel_additive} does not include higher-order interactions. If the effect of covariate $j$ truly depends on other covariates, then $\beta_{1j}(z)$ captures the mean effect averaged across these other covariates.  Thus, testing whether $\beta_{1j}(z)=0$ remains a sensible goal.
If desired, one may relax the additivity assumption, e.g. by incorporating interactions and tensor products,  and define $w_i$ to be the vector containing all such terms. 

Assuming a Gaussian distribution for $y$, \eqref{eq:vcmodel_additive} can be written as
in \eqref{eq:vcmodel_1covar_matrix}
where $(W_0,W_1)$ is an orthogonal cut basis as in \eqref{eq:blockdiag_basis}. The only difference relative to \eqref{eq:blockdiag_basis} is that we extend $W_{1b}$, the basis for local covariate effects, by multiplying it by the value of each covariate (akin to interaction terms in standard regression). 
For brevity we refer further details to Section \ref{supplsec:cutbasis_multiplecovar}.
There we also discuss that, since $W_1$ is an orthogonal cut basis, Lemmas \ref{lem:zero_indepmodel} and \ref{lem:zero_depmodel}  still apply:
if $E_F(y_b \mid X,Z)$ is linearly independent of covariate $j$ given the other covariates, then the asymptotic covariate effect in region $b$ is $\eta_{1jb}^*=0$ (for dependent data, $\tilde{\eta}_{1jb}^*=0$).


As a practical remark, although our discussion applies to flexible basis functions such as cubic splines, the computational burden can be substantial when one has many covariates. In our examples, using cubic splines for the baseline ($W_0$) and 0-degree splines for local covariate effects led to faster computations without compromising the quality of local variable selection. 

\section{A framework for local null testing}
\label{sec:framework}

\subsection{Bayesian model selection and averaging}
\label{ssec:bms}

Equation \eqref{eq:vcmodel_1covar_matrix} defines the likelihood associated to the observed $y \in \mathbb{R}^n$, given the design matrices $W_0,W_1$ constructed from the $p$ covariates in $X$, the coordinates in $Z$, and the corresponding parameters $\eta=(\eta_0,\eta_1) \in \mathbb{R}^{l_0+l_1 p}$.
Under model \eqref{eq:vcmodel_1covar_matrix}, one may test for the effect of covariate $j=1,\ldots,p$ in region $b \in \{1,\ldots,|\mathcal{R}|\}$ by testing $\eta_{1jb}=0$ versus $\eta_{1jb} \neq 0$, resulting in a total of $p |\mathcal{R}|$ tests (Section \ref{supplsec:cutbasis_multiplecovar}).
 One may use alternatives to our cut basis, such as a Haar wavelet basis, but then local tests are no longer given by single parameters being 0, see Section \ref{supplsec:haar_basis}. 

We now discuss how to incorporate local tests into a standard Bayesian model selection framework,
how to set priors, and how to use Bayesian model averaging to consider multiple choices of knots, which we refer to as a multi-resolution analysis.

To describe any arbitrary model within the local variable selection process, we introduce latent variables $\gamma= \{ \gamma_{jb} \}$, where
$\gamma_{0b}= \mbox{I}(\eta_{0b} \neq 0)$ and $\gamma_{jb}= \mbox{I}(\eta_{1jb} \neq 0)$ 
serve as indicators for the inclusion covariate $j=1,\ldots,p$ across regions $b=1,\ldots,|\mathcal{R}|$. The model size, i.e. the number of non-zero parameters, is denoted as $|\gamma|_0 = \sum_{jb} \gamma_{jb}$.
 We denote by $W_\gamma$ the subset of columns in the design matrix $W=(W_0,W_1)$ such that $\gamma_{0b}=1$, $\gamma_{jb}=1$.  

The goal is to select the optimal model $\gamma^*$, where $\gamma_{jb}^*= \mbox{I}(\eta_{1jb}^* \neq 0)$, with $\eta^*$ minimizing the  mean squared prediction error under the data-generating truth $F$ in \eqref{eq:asympsolution_indep} and \eqref{eq:asympsolution_dep}.
Under model misspecification, $\gamma^*$ is the optimal model among those under consideration: it has smallest dimension $|\gamma^*|_0$ among models that are Kullback-Leibler closest to $F$ \citep{rossell:2022}.

The posterior probability of model $\gamma$ is obtained as
\begin{align}
 p(\gamma \mid y)= \frac{p(y \mid \gamma) \, p(\gamma)}{p(y)}=
\frac{p(\gamma) \int p(y \mid \eta)\,  p(\eta \mid \gamma)\,  d\eta}{p(y)},
\label{eq:pp}
\end{align}
where $p(\gamma)$ is the model's prior probability and $p(\eta \mid \gamma)$ the prior on the parameters under model $\gamma$ (see Section \ref{ssec:priors}).
Given the posterior distribution $p(\gamma \mid y)$ for the entire vector $\gamma$, one can assess the local effect of covariate $j$ in region $b$ using the marginal posterior inclusion probabilities
\begin{align}
 \mbox{pr}(\gamma_{jb}=1 \mid y)= \sum_{\gamma: \gamma_{jb}=1} \, p(\gamma \mid y).
\label{eq:margpp}
\end{align}

As shown in Section \ref{sec:theory}, under mild regularity conditions, $p(\gamma^* \mid y)$ converges to 1 as $n$ grows, and thus \eqref{eq:margpp} concentrates on $\gamma^*_{jb}$ uniformly across the pairs $(j, b)$.
Hence, as $n$ grows, using either \eqref{eq:pp} or \eqref{eq:margpp} leads to consistently selecting $\gamma^*$ and vanishing family-wise type I-II errors.


We propose including local covariate effects that have high marginal posterior probability in \eqref{eq:margpp}. Specifically, we set $\hat{\gamma}_{jb}=\mbox{pr}(\gamma_{jb}=1 \mid y) \geq t$, for some threshold $t \in (0,1)$. In our illustrations, we used $t=0.95$. This choice ensures that the model-based posterior expected false discovery proportion 
remains below 0.05 \citep{mueller:2004}. 
This procedure also has a connection with frequentist type I errors. From Corollary 2 in \cite{rossell:2022}, for any true $\gamma^*_{jb}=0$, the frequentist probability of falsely selecting a covariate $j$ in region $b$ under any data-generating $F$ is upper bounded by
$ P_F(\hat{\gamma}_{jb}=1) \leq t^{-1} E_F \left( \mbox{pr}(\gamma_{jb} =1 \mid y) \right)$.

\subsection{Priors for the local effects and model selection}
\label{ssec:priors}

Posterior model probabilities in \eqref{eq:pp} require a prior $p(\gamma)$ on the models and a prior $p(\eta_\gamma \mid \gamma)$ on the coefficients under each model.
For the latter we take 
\begin{align}
 \eta_\gamma \mid \gamma &\sim N(0, g V_\gamma),
\label{eq:prior_paramters}
\end{align}
where $V_\gamma$ is a $|\gamma|_0 \times |\gamma|_0$ positive-definite matrix, and $g \in \mathbb{R}^+$.
We consider two default $V_\gamma$, a diagonal matrix (Normal shrinkage prior), and an extension of the intrinsic conditionally auto-regressive prior \citep{besag:1974} that adds a diagonal matrix into $V_\gamma^{-1}$ to obtain a proper prior (Section \ref{supplsec:prior_parameters}). The first prior imposes a penalty on the $L_2$ norm of $\eta$, in particular on local effects driving deviations from the baseline mean. The second prior also encourages spatial smoothness in $z$.
In both cases, 
we set a default $g=1$ so that prior precision's trace equals that of the unit information prior of \cite{schwarz:1978}.
In our examples, both priors gave very similar local testing results. This is supported by our Bayes factor rates (Section \ref{sec:theory}) and occurs because we use relatively few knots, hence encouraging smoothness is less critical than when one has many knots. By default we recommend the Normal shrinkage as it is faster computationally (its prior normalization constant and posterior precision matrix are simpler).
Our theory allows for other $g$ and $V_\gamma$, e.g. letting $g$ grow with $n$ to enforce sparsity, at the cost of decreased statistical power, see \cite{narisetty:2014,rossell:2022}.


In settings with very strong dependence and moderate $n$,  we found that replacing \eqref{eq:prior_paramters} by a group pMOM prior with default parameters (\cite{rossell:2021}) helped prevent type I error inflation. Therefore, for dependent data, we recommend the group pMOM prior as the default choice (see Sections \ref{sec:simulation_fda} and \ref{ssec:saez_data} for examples).

For the model prior, we consider a Complexity prior akin to \cite{castillo:2015}, 
\begin{align}
p(\gamma)=
\frac{C}{q^{c |\gamma|_0}} \times {q \choose |\gamma|_0}^{-1} \mbox{I}(|\gamma|_0 \in \{0,\ldots,\bar{q}\}),
\label{eq:prior_modelsize}
\end{align}
where recall that $q$ is the total number of parameters defined after \eqref{eq:vcmodel_additive}, $c \geq 0$ a prior parameter, $C=(1 - 1/q^c) / (1-1/q^{c*(\bar{q}+1)})$ the normalizing constant,
and we consider that, although possibly $q \gg n$, one restricts attention to models with $|\gamma|_0 \leq \bar{q}$ parameters, where $\bar{q}$ is a user-defined bound such as $\bar{q} \leq n$ (since models with $|\gamma|_0 > n$ result in data interpolation).
For $c>0$ the prior probabilities decay exponentially with the model size, $|\gamma|_0$, and then equal prior probabilities are set on  all models that have the same size.

Setting $c=0$ returns the Beta-Binomial(1,1) prior of \cite{scott:2006}, i.e. a uniform prior on the model size, 
which served as the basis for the extended BIC of \cite{chen:2008}. 
While our primary emphasis is on Bayesian methods, the theoretical results of Section \ref{sec:theory} apply directly to the extended BIC framework by setting $c=0$ and $g=1$.  In our theoretical framework we consider a general $c$, but in our examples we set \(c = 0\). In our experience this yields a better balance between sparsity and power to detect non-zero coefficients for finite $n$.

\subsection{Multi-resolution analysis}
\label{ssec:varying_resolutions}

To simplify exposition, so far we  assumed that one has a predefined set of regions $\mathcal{R}$ to perform the local null tests, e.g. the 8 intervals in Figure \ref{fig:splinefit}. 
In practice, it is often unclear what regions to use, or in other words, what resolution is appropriate for the local variable selection. 
Although our cut basis formulation prevents type I errors for any choice of regions, if one defines too many regions then the sample size per region becomes too small, leading to type II errors.
For example, suppose that one has a single binary covariate $x_i$ that defines two groups, and we want to compare $E(y_i \mid z_i, x_i)$ across these groups. If the group differences change little across large regions in $z_i$, it is advantageous to use fewer regions, to increase statistical power.  In contrast, if the group differences change rapidly with $z_i$, a finer resolution with more regions  may be preferred. 

The choice of the maximum number of regions  is typically guided by computational limitations (as the model space grows exponentially) or because higher resolutions offer limited practical interest. 
We hence propose a multi-resolution analysis to determine whether a coarser resolution may be more suitable, given the data at hand.
For example, in our salary application, we consider 2.5	year bins as the highest resolution, since measures targeting salary discrimination are unlikely to have noticeable effects in shorter time spans. Interestingly, our multi-resolution analysis places most posterior probability on (coarser) 10-year intervals. 

\begin{figure}[t!]
\begin{center}
\includegraphics[width=0.5\textwidth,height=0.5\textwidth]{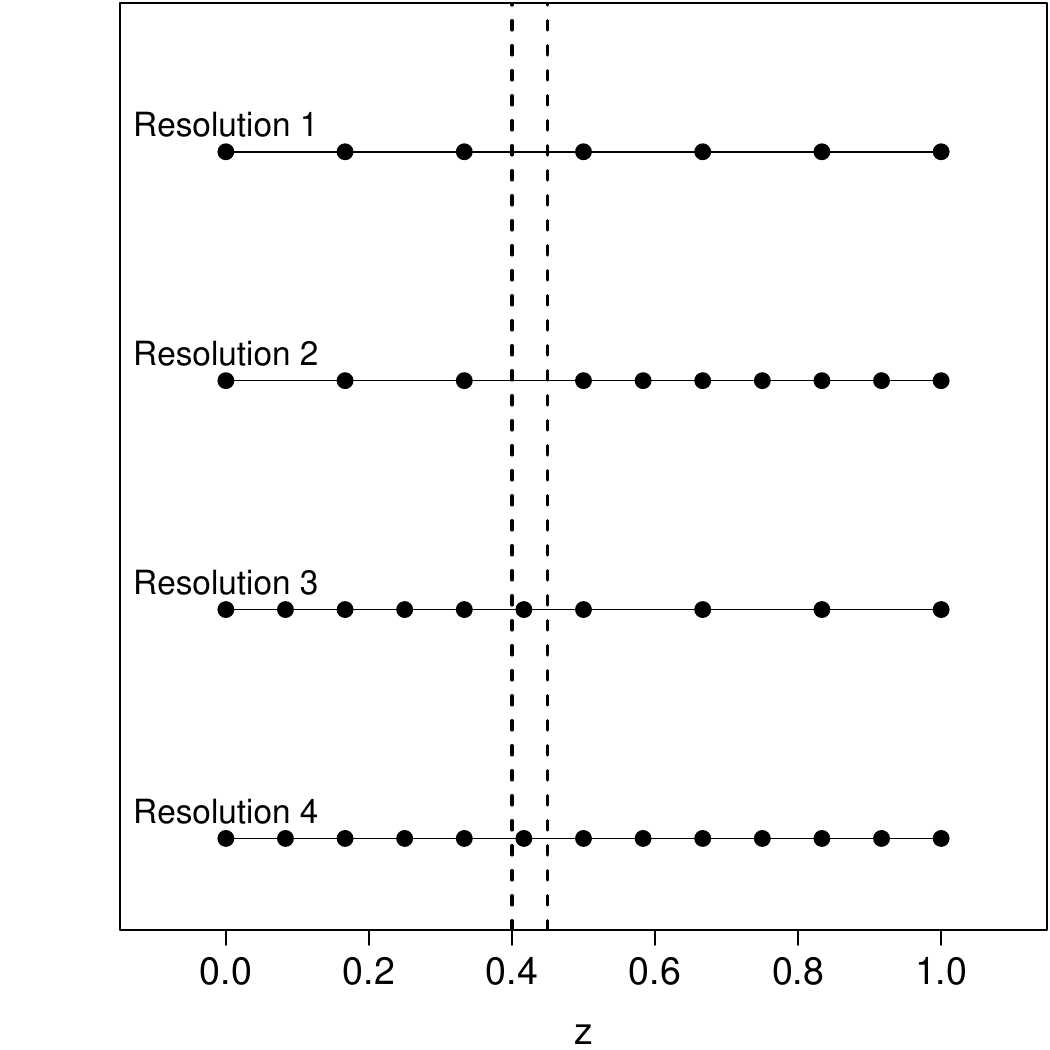}
\end{center}
\caption{Multi-resolution analysis with 4 resolutions and $z \in [0,1]$.
Resolution 1 has 6 regular intervals. Resolutions 2-3 split the [0,0.5] and [0.5,1] into 6 intervals.
Resolution 4 has 12 intervals.
In resolutions  1-2, $\beta_j(z=0.4) \neq 0$ if and only if $\beta_j(z=0.45) \neq 0$ (dashed lines), leading to prior dependence across $z$. 
}
\label{fig:multires}
\end{figure}

Bayesian model averaging provides a natural framework for handling multiple resolutions.  We consider $L$ resolutions, where each resolution $l=1,\ldots,L$ consists of a set of regions $\mathcal{R}_l={R}_{l1},\ldots,R_{l |\mathcal{R}_l|}$. 
 If so desired, the multi-resolution analysis can place more knots in certain parts of the support of $z$, see Figure \ref{fig:multires}. 
Similar to Section \ref{ssec:priors}, let $\eta_l=\{ \eta_{ljb} \}$ and $\gamma_l=\{ \gamma_{ljb} \}$ for $j=0,\ldots,p$ and $b=1,\ldots,|\mathcal{R}_l|$ represent the set of parameters and models associated with resolution $l$,
where the latent $\gamma_{ljb}= \mbox{I}(\eta_{lbj} \neq 0)$  indicates whether covariate $j$ has an effect in region $R_{lb}$. We consider a joint prior that factorizes as follows,
\begin{align}
 p(\eta_l, \gamma_l, \mathcal{R}_l)= p(\eta_l \mid \gamma_l, \mathcal{R}_l)\, p(\gamma_l \mid \mathcal{R}_l) \, p(\mathcal{R}_l).
\nonumber
\end{align}
where $p(\eta_l \mid \gamma_l, \mathcal{R}_l)$ is as in \eqref{eq:prior_paramters},
and the model selection prior in \eqref{eq:prior_modelsize} is modified to accommodate the multi-resolution setting as follows,
\begin{align}
p(\gamma_l \mid \mathcal{R}_l)=
\frac{C_l}{q_l^{c |\gamma|_0}} \times {q_l \choose |\gamma_l|_0}^{-1} \mbox{I}(|\gamma_l|_0 \in \{0,\ldots,\bar{q}\}),
\label{eq:prior_model_multires}
\end{align}
where $q_l= \mbox{dim}(\eta_l)$ is the number of parameters for resolution $l$,  $c \geq 0$ and $C_l=(1 - 1/q_l^c) / (1-1/q_l^{c*(\bar{q}+1)})$ the normalizing constant.
Finally, by default, uniform prior probabilities are assigned to the resolution levels, i.e., $p(\mathcal{R}_l)=1/L$.

The multi-resolution formulation induces dependence in the local variable selection across nearby $z$ values, as illustrated in Figure \ref{fig:multires}.
Specifically,  the probability of a local effect is 
$P \left( \beta_j(z) \neq 0 \right)= \sum_{l=1}^L  P \left( \beta_j(z) \neq 0 \mid \mathcal{R}_l  \right) p(\mathcal{R}_l)$,
where $P \left( \beta_j(z) \neq 0 \mid \mathcal{R}_l  \right)$ is  constant across all $z$ values within the same interval at resolution $l$. Hence, nearby $z$ and $z'$ receive similar prior probabilities $\mbox{pr}( \beta_j(z) \neq 0)$ and $\mbox{pr}(\beta_j(z') \neq 0)$.
The corresponding posterior probabilities are
\begin{align}
 P \left( \beta_j(z) \neq 0 \mid y \right)&=
 \sum_{l=1}^L  P \left( \beta_j(z) \neq 0 \mid y, \mathcal{R}_l  \right) \, p(\mathcal{R}_l \mid y),
 \nonumber
\end{align}
where $P \left( \beta_j(z) \neq 0 \mid \mathcal{R}_l, y  \right)$ given a resolution level $l$ is given in \eqref{eq:margpp}, and
$
 p(\mathcal{R}_l \mid y) \propto p(y \mid \mathcal{R}_l)\,  p(\mathcal{R}_l)= \sum_{\gamma_l} p(y \mid \gamma_l, \mathcal{R}_l) p(\gamma_l \mid \mathcal{R}_l) \, p(\mathcal{R}_l)$
is the marginal likelihood of resolution $l$.

Our multi-resolution strategy is designed to facilitate parallel computation across resolutions. We define a distinct model for each resolution  $\mathcal{R}_l$, obtain $\mbox{pr}(\beta_j(z) \neq 0 \mid y, \mathcal{R}_l)$ separately for each resolution, and then compute their weighted average.
A natural alternative is to induce prior dependence across resolutions. 
 This is because, if for example $\beta_j(z) \neq 0$ for some $z \in [0,1/2]$ and $\beta_j(z) = 0$ for all $z \in [1/2,1]$, then $\beta_j(z) \neq 0$ for some $z \in [0,1]$ at a coarser resolution. 
However, we chose not to adopt such more advanced priors here for computational reasons. 
First, computations would no longer be easily parallelizable.  Second, within a resolution level one can quickly obtain marginal likelihoods for a model by using rank 1 Cholesky decomposition updates, given that of another model that only differs by one zero parameter.

\section{Theoretical results}
\label{sec:theory}

\subsection{Bayes factors}
\label{ssec:bf}

We present two main results describing the asymptotic behavior of the proposed cut basis framework. First, we obtain rates at which the Bayes factor favors the optimal model $\gamma^*$ over some other model $\gamma$ (Theorem \ref{thm:bf}). Then, we establish that the  posterior model probabilities consistently select $\gamma^*$ (Theorem \ref{thm:pp}).  
These results imply that asymptotically one detects which covariates have an effect at each considered coordinate $z$, and hence attains consistent local variable selection. As discussed in Section \ref{sec:multiple_covar}, since the assumed additive structure in \eqref{eq:vcmodel_additive} may be misspecified, consistency refers to detecting non-zero marginal effects of each covariate in $x$, averaged across the values of other covariates in $x$.
We consider high-dimensional settings where the total number of parameters $q = l_0 + l_1p$ can grow with $n$.
The size of the optimal model $|\gamma^*|$ may also increase with $n$. 

We assume that the data-generating $F$ has sub-Gaussian tails, specifically  $F$ satisfies
$ y - W_{\gamma^*} \eta_{\gamma^*}^* \sim SG(0, \omega)$
for some $\omega > 0$, where $\eta^*{\gamma^*}$ is the optimal parameter value in \eqref{eq:asympsolution_dep}.
This assumption allows for $y$ to be dependent.  For example, if $y \sim N(\mu, \Omega)$ for some positive-definite $\Omega$, then $y \sim SG(\mu,\omega)$ where $\omega$ is the largest eigenvalue of $\Omega$.
We consider a misspecified setting where one specifies  a model that may not match $F$. 
To simplify the exposition and proofs,  we assume that the specified model has a fixed covariance, $\Sigma$.
Specifically, one specifies the model
\begin{align}
 y \mid \gamma, \eta_\gamma &\sim N(W_\gamma \eta_\gamma, \Sigma).
\label{eq:model_theorem}
\end{align}

Our results extend to unknown $\Sigma$ as follows. First, as long as $\hat{\Sigma} \stackrel{P}{\longrightarrow} \Sigma$ for a fixed positive-definite block-diagonal $\Sigma$, Theorems \ref{thm:bf}-\ref{thm:pp} continue to hold by the continuous mapping theorem.
If the tails of $\hat{\Sigma}^{-1/2} y$ are sub-Gaussian, then our rates remain exactly valid. If said tails are exponential or thicker, then the rates become slower.

The Bayes factor $B_{\gamma \gamma^*}= p(y \mid \gamma) / p(y \mid \gamma^*)$ compares a model $\gamma$ with the optimal $\gamma^*$, see Section \ref{supplsec:bf_derivation} for its expression. When $B_{\gamma \gamma^*}$ is close to zero then $\gamma^*$ is favored over $\gamma$. We give the rates at which $B_{\gamma \gamma^*}$ converges to 0 in probability as $n \rightarrow \infty$, assuming the following conditions.

\begin{enumerate}[leftmargin=*,label=(A\arabic*)]
\item The matrix $W_\gamma^T \Sigma^{-1} W_\gamma$ has full column rank $|\gamma|_0$.

\item Let $(\underline{l}_\gamma,\bar{l}_\gamma)$ be the smallest and largest eigenvalues of $V_{\gamma} W_{\gamma}^T \Sigma^{-1} W_\gamma/n$.
They satisfy $c_1 \leq \underline{l}_\gamma \leq \bar{l}_\gamma \leq c_2$ for all $n \geq n_0$ and some constants $c_1,c_2,n_0>0$.

\item $\Sigma^{-1}$ exists. Its largest eigenvalue $\tau$ satisfies $c_3 < \tau < c_4$ for constants $0< c_3,c_4 < \infty$.

\item $\lim_{n \rightarrow \infty} g n= \infty$.

\item Let $\lambda_\gamma= (\widetilde{W}_{\gamma^*} \eta_{\gamma^*}^*)^T (I - H_\gamma) \widetilde{W}_{\gamma^*} \eta_{\gamma^*}^*$,
where $\widetilde{W}_\gamma= \Sigma^{-1/2} W_\gamma$, 
$H_\gamma= \widetilde{W}_\gamma (\widetilde{W}_\gamma^T \widetilde{W}_\gamma)^{-1} \widetilde{W}_\gamma^T$.
For some sequence $d_n \geq 0$ such that $\lim_{n \rightarrow \infty} d_n= \infty$,
$$
\lim_{n \rightarrow \infty} \frac{\lambda_\gamma}{2 \log \lambda_\gamma} + \frac{|\gamma|_0 - |\gamma^*|_0}{2} \log(gn) - \omega \tau |\gamma|_0 \log d_n= \infty.
$$

\end{enumerate}

The conditions are fairly minimal, for brevity we refer their discussion to Section \ref{supplsec:conditions_thm_bf}.
A key quantity is $\lambda_\gamma$ in Assumption (A5), it is a non-centrality parameter measuring the sum of squares explained by the optimal $\gamma^*$ but not by model $\gamma$.
Under eigenvalue and beta-min conditions, $\lambda_\gamma$ is lower-bounded by $n$ times the smallest square entry in the optimal coefficients $|\eta_{\gamma^*}^*|$  \citep[see, e.g.,][Sections 2.2 and 5.4]{rossell:2022}.


Before stating Theorem \ref{thm:bf}, we interpret its main implications.
Part (i) indicates that overfitted models, which include all parameters in $\gamma^*$ and some extra, are essentially discarded at a polynomial rate $(gn)^{(|\gamma|_0 - |\gamma^*|_0)/2}$. 
Part (ii) states that non-overfitted models, which are missing parameters from $\gamma^*$,   are effectively discarded at an exponential rate in $\lambda_\gamma$, times a polynomial rate akin to that in Part (i). 
In the theorem statement the sequence $d_n$ should be considered as a lower-order term, e.g., $d_n= \log(gn)$, and $\delta$ as a constant close to 0.

\begin{thm}
Let $d_n\geq 0$ be any sequence such that $\lim_{n \rightarrow \infty} d_n= \infty$. Assume (A1)-(A5).
\begin{enumerate}[leftmargin=*,label=(\roman*)]
\item Overfitted models. If $\gamma^* \subset \gamma$, then
$\lim_{n \rightarrow \infty} P_F \left( B_{\gamma \gamma^*} \geq  \left( \frac{d_n}{gn} \right)^{\frac{|\gamma|_0 - |\gamma^*|_0}{2}} \right)= 0.$

\item Non-overfitted models. If $\gamma^* \not\subset \gamma$, then for all fixed $\delta>0$
$$
\lim_{n \rightarrow \infty} 
P_F \left( B_{\gamma \gamma^*} \geq \left( gnk_2 \right)^{-\frac{(|\gamma|_0-|\gamma^*|_0)}{2}} e^{-\frac{\lambda_\gamma (1-\delta)}{2 \log \lambda_\gamma} + \omega \tau |\gamma|_0 \log d_n} \right)= 0,
$$
where $\omega$ is the sub-Gaussian parameter, $\tau$ the largest eigenvalue of $\Sigma^{-1}$
and $k_2= \bar{l}_{\gamma^*} (1+\delta) / \underline{l}_\gamma$ is a constant under Assumption (A2).
\end{enumerate}
\label{thm:bf}
\end{thm}

\subsection{Model selection consistency}
\label{ssec:pp}

Theorem \ref{thm:pp} below shows that the posterior probability of the optimal model $p(\gamma^* \mid y) \stackrel{L_1}{\longrightarrow} 1$ as $n \rightarrow \infty$ and provides the associated rate.
The result bounds separately the total posterior probability assigned to the set of overfitted models ($S_0$),
and to non-overfitted models that are either smaller ($S_1$) or larger ($S_2$) than the optimal $\gamma^*$. This decomposition helps understand false positive versus power trade-offs.

Theorem \ref{thm:pp} requires additional conditions involving the sub-Gaussian parameter $\omega$ of the data-generating $F$, the largest eigenvalue $\tau$ of $\Sigma^{-1}$, the prior dispersion parameter $g$ in \eqref{eq:prior_paramters}, and the model prior parameter $c \geq 0$ in \eqref{eq:prior_modelsize}. More specifically, 

\begin{enumerate}[leftmargin=*,label=(B\arabic*)]
\item The optimal model has size $|\gamma^*|_0 \leq \bar{q}$, where $\bar{q}$ is the maximum model size in \eqref{eq:prior_modelsize}.

\item $\omega \, \tau >1$.

\item $p(\gamma)$ is non-increasing in model size $|\gamma|_0 \in \{0,\ldots,\bar{q}\}$ and, for any $|\gamma|_0 > |\gamma^*|_0$, 
$$\lim_{n \rightarrow \infty} \frac{1}{2}\log(g n) + \frac{1}{|\gamma|_0 - |\gamma^*|_0} \log \left( \frac{p(\gamma^*)}{p(\gamma)} \right) - (1+2^{\frac{1}{2}})^2 \omega \tau \bar{q} = \infty .$$

\item For some $a \in (0, 1/(2 \omega \tau))$,
$\lim_{n\rightarrow \infty} a \log(g n) + a (c + 1) \log(q) - \log(q) = \infty.$


\item For any not over-fitted model $\gamma \not\subset \gamma^*$ of size $|\gamma|_0 \leq |\gamma^*|_0$ and any fixed $k >0$,
$$
\lim_{n \rightarrow \infty} t + \frac{\lambda_\gamma}{\log \lambda_\gamma} - k |\gamma|_0 \log \left( t + \frac{\lambda_\gamma}{\log \lambda_\gamma} \right) = \infty
$$
where
$
 t= (|\gamma|_0 - |\gamma^*|_0) \log(g n \, k_2) + 2 \log \left( p(\gamma^*)/p(\gamma) \right).
$

\item Let $\underline{\lambda}= \frac{\min_{|\gamma|_0 \leq |\gamma^*|_0} \lambda_\gamma}{\max\{|\gamma^*|_0-|\gamma|_0,1\}}$. Then,
$\lim_{n \rightarrow \infty}
\frac{\underline{\lambda}}{2 \omega \tau} - c \log(q) - \frac{1}{2} \log(gn)= \infty.$

\item For some $r< 1/(\omega \, \tau)$,
$\lim_{n \rightarrow \infty} \frac{(|\gamma^*|_0+1) \bar{q}^r}{q^{cr-1-|\gamma^*|_0} (q - |\gamma^*|_0)^r (g n)^{r/2}}= 0$.
\end{enumerate}

Assumption (B1) says that the optimal model $\gamma^*$ has positive prior probability. Assumption (B2) is a worst-case scenario, one obtains faster rates in Theorem \ref{thm:pp}(i) and (iii) if $\omega \tau < 1$, see the proof for details. 
Assumption (B3) is a mild requirement that $g\, n$ and $p(\gamma^*)/p(\gamma)$ are not too small relative to the model size $|\gamma|_0$.
(B4) bounds the total number of parameters $q$ as a function of $n$, and is satisfied by setting a large $c$ (sparse model prior) or $g$ (dispersed coefficient prior).
Assumptions (B5) and (B6) are stronger versions of (A5) requiring that the non-centrality parameters $\lambda_\gamma \geq (|\gamma^*|_0 - |\gamma|_0) \underline{\lambda}$ are large enough.
Altogether, (B4)-(B6) limit the amount of prior sparsity induced by the model prior parameter $c$ and the prior dispersion $g$, relative to $n$ and $q$.
Finally, Assumption (B7) is similar to (B2)-(B3) and ensures that the rate in Theorem \ref{thm:pp}(iii) converges to 0. 
(B7) is stronger than needed but simplifies the exposition, 
please see Assumption (B7') and Theorem \ref{thm:pp_alternative} in Section \ref{ssec:alt_thm_pp} for further details.


\begin{thm}

\underline{Part (i).}
Assume (A1)-(A5) and (B2)-(B4). Let $S_0= \{\gamma: \gamma^* \subset \gamma\}$.
There exist constants $k>0$ and $n_0$ such that, for all $n \geq n_0$ and $r<1/(\omega \tau)$,
\begin{align}
E_F(\mbox{pr}(S_0 \mid y)) < \frac{k (|\gamma^*|_0 +1)}{q^{r(c + 1) -1} (g n)^{\frac{r}{2}}}.
\nonumber
\end{align}

\underline{Part (ii).}
Assume (A1)-(A5), (B1), (B2), (B5) and (B6). Let $\underline{\lambda}$ be the signal strength parameter in (B6),
and $S_1= \{\gamma: |\gamma|_0 \leq |\gamma^*|_0, \gamma \not\subset \gamma^*\}$.
For $n \geq n_0$ and fixed $n_0$,
\begin{align}
E_F(\mbox{pr}(S_1 \mid y))
< 9 \exp \left\{  -\frac{\underline{\lambda}(1-\epsilon)}{2 \tilde{\omega}} + [|\gamma^*|_0(1+\epsilon) + c] \log q + \frac{1}{2} \log(gn) \right\},
\nonumber
\end{align}
for a constant $\alpha >0$ that may be taken arbitrarily close to 0.

\underline{Part (iii).}
Assume (A1)-(A5), (B1), (B2), (B5) and (B7). 
Let $S_2= \{\gamma: |\gamma|_0 > |\gamma^*|_0, \gamma \not\subset \gamma^*\}$.
There exist constants $k>0$ and $n_0$ such that, for all $n \geq n_0$ and $r<1/(\omega \tau)$,
$$
E_F(\mbox{pr}(S_2 \mid y)) \leq \frac{k (|\gamma^*|_0+1) b^{1/2} \bar{q}^r}{q^{cr-1-|\gamma^*|_0} (q - |\gamma^*|_0)^r (g n)^{r/2}}.
$$
\label{thm:pp}
\end{thm}

In Theorem \ref{thm:pp},  $r$ should regarded as being close to $1/(\omega\, \tau)$. Theorem \ref{thm:pp} (i) and (iii) state that overfitted and large non-overfitted models (respectively) receive vanishing posterior probability as $n$ grows, at a rate that is faster when the prior complexity and the dispersion parameters $(c,g)$ are large, and slower when either the optimal model is not sparse (i.e., $|\gamma^*|_0$ is large) or $\omega \tau$ is large.
In the well-specified case where $F$ is Gaussian with independent observations, then $\omega\, \tau=1$. Hence, the condition $r < 1/(\omega\,  \tau) <1$ reflects that under strongly dependent data or model misspecification, convergence rates get slower.
Similarly, Part (ii) states that small non-overfitted models are discarded at an exponential rate in $\underline{\lambda}/(2\omega\, \tau)$, which gets slower when $ \omega\, \tau>1$ and when either  $q$ or $|\gamma^*|_0$ are large. If $(c,g)$ are large, i.e. chosen to favor smaller models, then the convergence rate is slower. This reflects the intuitive notion that, by inducing stronger sparsity, the statistical power to detect truly active coefficients is reduced.

\section{Simulation studies}
\label{sec:results}

\subsection{Simulation with independent errors}
\label{sec:simulation_iid}

We consider a simulation with independent errors to compare our cut orthogonal basis with our same Bayesian framework with standard (uncut) cubic splines.
 We used the Gaussian shrinkage and extended auto-regressive priors discussed after \eqref{eq:prior_paramters}. 
We also considered the VC-BART of  \cite{deshpande:2020}  fitting a varying coefficient model via Bayesian additive regression trees,
 applying the fused lasso to our cut orthogonal basis (\cite{tibshirani:2005}, regularization parameters set with the extended Bayesian information criterion), 
and a least-squares regression where one obtains Benjamini-Hochberg adjusted P-values for interactions between each covariate and discretized coordinates $z$.
 Finally, we used gam function in R package mgcv to fit an additive model with generalized cross-validation, and used 95\% simultaneous confidence bands for $\hat{\beta}_j(z)$ to test local null hypotheses. We used the default 12 knots, and also 24. 
In our approach, we used a cubic B-spline with 20 knots for the baseline and a multi-resolution analysis involving 7, 9 and 11 knots to evaluate  the covariate effects. 
We considered cut basis of degree 0 (piecewise-constant) and 3 (cubic) but we only report the results from the former, since the results were very similar.
We used Markov Chain Monte Carlo  sampling to search over models characterized by different $\gamma$, and to obtain posterior probability estimates for the covariate effects as well as posterior samples for all the model parameters. We used the default specifications in the R package mombf (5,000 Gibbs iterations, with 500 burn-in, see \cite{rossell:2017}). 

We illustrate the issues that arise when using cubic B-splines, and how these issues are addressed by using a cut orthogonal basis. 
We consider two scenarios, with $n=100$ and $n=1,000$.  In both cases, we consider $p=10$ covariates, and a regular grid of $n$ values for $z_i \in [-3,3]$. The first covariate is a binary group indicator $x_{i1} \in \{0,1\}$ with equal group sizes. The remaining covariates are normally distributed with mean $x_{i1}$, and variance 1,  $x_{ij} \sim N(x_{i1},1)$. This results in  a mild empirical correlation between $x_1$ and each remaining covariate of 0.43. 
The expected outcome is set to depend truly only on the first covariate, and that only when $z>0$, as depicted in Figure \ref{fig:splinefit} (grey lines).
More in detail, we set the following model 
\begin{align}
 E(y_i \mid x_i, z_i)= 
\begin{cases}
\cos(z_i), &\mbox{ if } z_i \leq 0, \\
0,   &\mbox{ if } z_i>0, x_{i1}=1, \\
1/(z_i+1)^2, &\mbox{ if } z_i>0, x_{i1}=0.
\end{cases}
\nonumber
\end{align}
 The error variance is $0.25^2$. 
Section \ref{ssec:simulation_iid_bivar} shows an extension assessing our approach with bivariate $z$ where, as summarized in Figure \ref{fig:simiid_bivar}, using orthogonal cut basis also prevents type I errors.
For each scenario, we report averaged results over 100 independent replicates

Table \ref{tab:simiid} shows the type I error and power for our cut B-spline basis, and the competing methods.
 The results for our two prior covariances are very similar. The additive model in mgcv exhibits good power, but the type I error is inflated for covariate 1 and $z \in (-1,0]$ (0.28 for $n=100$, 0.08 for $n=1000$). The results with 24 knots were very similar (Table \ref{tab:simiid_extra}). 
P-value adjustment exhibits a low type I error, but the power is very low for small $n$. For example, for $n=100$ and $z \in (1,2]$, the estimated power for covariate 1 is 0, whereas for the proposed cut basis, it is 0.91. This finding highlights the importance of conducting local variable selection:  by learning that covariates 2-10 are unnecessary, the power to detect local effects for covariate 1 increases significantly.
 Fused lasso resulted in high false positive rates. 
One striking result that strongly supports our theoretical arguments is that when employing cubic splines for the local tests, both the posterior probabilities of rejecting the null hypothesis and the type I error for covariate 1 are near 1 for values  $z \in (-1,0)$. The 0-degree orthogonal cut spline is not affected by these issues. In this example it yields near-perfect inference, except that for $n=100$, the power to detect $\beta_1(z) \neq 0$ for $z \in (0,1)$ is approximately 0.25.
 Figure \ref{fig:simiid}  further illustrates these results. The left panel displays the average posterior probabilities $\mbox{pr}(\beta_j(z) \neq 0 \mid y)$ across the simulation replicates while the right panel presents the power function. Both panels once again reveal that  mgcv and  standard B-splines run into false positive inflation for covariate 1 at $z \in (-1,0)$.


VC-BART reported local effects for (truly inactive) covariates 2-10 in all simulations, when using a threshold of 0.95 posterior probability inclusion in a tree. For this reason, we decided instead  to reject a local null hypothesis at $z$ when the 95\% posterior interval for $\beta_j(z)$ did not include 0.   Although this procedure resulted in a lower type I error, it remained above 0.15 for some $z$'s, even for $n=1,000$, and the power was generally lower than for our methodology (Figure \ref{fig:simiid}, Table \ref{tab:simiid}).
As discussed, a strategy to prevent false positives in VC-BART is to only report estimated covariate effects above some threshold; however,  it is not obvious how to set the threshold. Additionally, this approach cannot detect effects below the selected threshold. 

Table \ref{tab:simiid_mse} shows the root mean squared estimation error for the local effects $\beta_j(z)$, averaged across covariates and $z$. Said error  was very similar for degree 0 cut splines and standard cubic B-splines, and for VC-BART it was nearly twice for $n=100$ and four times larger for $n=1,000$. 
 R package mgcv also attained excellent estimation error, both for 12 and 24 knots. 

Finally, our proposed approach had substantially lower computational time compared to VC-BART, and scaled better as $n$ increased (e.g.,  5 seconds vs. 7 minutes  for $n=1,000$, Table \ref{tab:cputime_simiid}). 
This advantage is likely because our approach allows pre-computation of sufficient statistics. After this pre-computation, the computational cost does not depend on $n$.
 R package mgcv was fastest (e.g. 4 seconds for $n=1,000$), since it fits a single model. 


 Section \ref{ssec:simulation_iid_misalign} studies the case where the true change-point is at $z=0.15$, where no knot is placed in any of the 3 resolutions.
Our method found evidence for a local effect of covariate 1 in regions that include any $z \geq 0.15$, and no evidence elsewhere.
Section \ref{ssec:simulation_iid_nonequispaced} then considers 4 resolutions, two of which place more knots on $z \geq 0$ than on $z<0$.
The evidence for local covariate effects was similar than in our previous simulations.
Interestingly, for $n=100$ the posterior distribution favored the resolution with less knots, whereas for larger $n$ it favored the resolutions placing more knots in $z>0$, where the effect of covariate 1 truly varies (Figure \ref{fig:simiid_nonequispaced}).

\begin{table}
\begin{center}
\begin{tabular}{ccccccccc} 
\multicolumn{9}{c}{$n=100$} \\ 
& \multicolumn{4}{c}{Covariate 1} & \multicolumn{4}{c}{Covariates 2-10} \\ 
 Region          & Gaussian &  ICAR+    & VCBART & GAM         & Gaussian &  ICAR+      & VCBART     & GAM  \\
                 & prior    &  prior    &        &                         & prior    &  prior      &            & \\
$z \in $ (-3,-2] & 0        &  0        &0.05    &  0.004      & 0       &  0       & 0.14$^{**}$&  0.04 \\
$z \in $ (-2,-1] & 0        &  0        &0.05    &  0.03       & 0       &  0       & 0.14$^{**}$&  0.04  \\ 
$z \in $  (-1,0] & 0        &  0        & 0      &  0.28$^{**}$ & 0.001   &  0       & 0.002      &  0.04 \\ 
$z \in $   (0,1] & 0.25     &  0.21     &0.18    &  0.79       & 0.001   &  0.001   & 0          &  0.05 \\ 
$z \in $   (1,2] & 0.91     &  0.91     &0.95    &  0.99       & 0       &  0.001   & 0          &  0.05 \\ 
$z \in $   (2,3] & 0.96     &  0.96     &0.95    &  0.94       & 0       &  0       & 0          &  0.04 \\ 
\multicolumn{9}{|c|}{$n=1000$} \\ 
& \multicolumn{4}{c}{Covariate 1} & \multicolumn{4}{c}{Covariates 2-10} \\ 
 Region         & Gaussian &  ICAR+  & VCBART     & GAM         & Gaussian & ICAR+& VCBART     & GAM \\
                & prior    &  prior  &          &                       & prior  & prior&            & \\
$z \in $(-3,-2] &   0      & 0       & 0.08$^{**}$  & 0.001       & 0      & 0     & 0.17$^{**}$&   0.006 \\ 
$z \in $(-2,-1] &   0      & 0       & 0.08$^{**}$  & 0.001       & 0     & 0     & 0.17$^{**}$&   0.006 \\ 
$z \in $ (-1,0] &   0      & 0        & 0         & 0.08$^{**}$  & 0    & 0     & 0          &  0.01 \\ 
$z \in $  (0,1] &   1      & 1        & 0.68      & 0.95        & 0     & 0    & 0          &  0.01 \\ 
$z \in $  (1,2] &   1      & 1        &  1        & 1           & 0     & 0    & 0.002      &  0.01 \\ 
$z \in $  (2,3] &   1      & 1        &  1        & 1           & 0     & 0    & 0.002      &  0.01 \\ 
\end{tabular}
\end{center}
\caption{Independent errors simulation. Proportion of rejected null hypothesis under the Gaussian shrinkage and extended intrinsic conditionally auto-regressive (ICAR+) priors, VC-BART and gam in R package mgcv (GAM) with default 12 knots. 
For covariate 1 and $z>0$ this is the statistical power, else it is the type I error. ** indicates a type I error $>0.05$}
\label{tab:simiid}
\end{table}

\subsection{Functional data simulation}
\label{sec:simulation_fda}

\begin{table}
\begin{center}
\begin{tabular}{ccccc} 
 Region          & Cut0 & Tensor & IT       & VC-BART  \\  
$z \in $ (-3,-2] & 0.02 & 0.032&   0.01   &  0.51 \\ 
$z \in $ (-2,-1] & 0.035& 0.032&   0.01   &  0.44 \\ 
$z \in $  (-1,0] & 0.047& 0.032&  0.0131  &  0.46 \\
$z \in $   (0,1] & 0.96 & 0.657&     0.389&  0.92 \\ 
$z \in $   (1,2] & 0.981& 0.791&     0.964&  1    \\ 
$z \in $   (2,3] & 1    & 0.800&    0.992 &  1    \\ 
\end{tabular}
\end{center}
\caption{Functional data simulation. Proportion of rejected null hypothesis for 0-degree cut orthogonal basis, tensor model, interval testing (IT) procedure  and VC-BART. For $z<0$ this is the type I error, for $z>0$ the statistical power}
\label{tab:simfda}
\end{table}

We simulate functional data for $M=50$ individuals over a uniform grid of $100$ time points within the interval $z \in [-3,3]$, with strongly correlated errors. For each individual, we consider  $p=10$ covariates and generate data over said grid. We employ the same  model for the outcome and covariates as in Section \ref{sec:simulation_iid}, that is  only covariate 1 truly has an effect (at coordinates $z > 0$). However,  now the errors are drawn from a mean zero, unit-variance, Gaussian process with autocorrelation $\mbox{cov}(\epsilon_{it}, \epsilon_{it'})= 0.99^{|t-t'|}$. We simulate $100$ independent replicates and report averaged results.

As discussed in Section \ref{ssec:priors}, we replaced the Gaussian shrinkage prior in \eqref{eq:prior_paramters} by the group product moment prior of \citet{rossell:2021}, with default parameters.
We compare our approach to the tensor model method of \cite{paulon:2023}, the interval testing procedure of \cite{pini:2016},  and VC-BART.
The former is designed to detect high-order interactions between multiple discrete covariates, and has two versions: one designed for a single covariate and another for multiple covariates. The multiple covariate version performed poorly, with a type I error rate exceeding 0.5. 
Given also that the latter method is designed for a single covariate, we first focus on comparing the three methods when considering only the truly active covariate 1.

Table \ref{tab:simfda} shows that our method has higher power than LFMM (0.96 vs. 0.657 for $z \in (0,1]$), with a similar type I error rate.  Similarly, the interval testing procedure exhibits low power in detecting the covariate effect within the range of $z \in (0,1]$. While the procedure controls the family-wise error rate within any specified interval of the domain, this also diminishes its power.  VC-BART exhibited high type I error rates. 
Figure \ref{fig:pp_simfda} (left) shows the average posterior probability of a local effect for covariate 1 for both the LFMM and our method. At coordinates $z<0$, where there is truly no covariate effect, these probabilities are close to 0 for both methods. However, for $z>0$, where the covariate effect truly exists, our method exhibits higher posterior probabilities.

The right panel of Figure \ref{fig:pp_simfda} and Table \ref{tab:simfda_10covar} present the results for our method when using all 10 covariates. The results for covariate 1 are overall similar to the single covariate setting, although the posterior probabilities for $z<0$ and the corresponding type I errors are slightly higher. For  $z>0$, the power is also slightly lower (e.g. from $0.96$ to $0.73$ for $z \in (0,1]$). For the truly inactive covariates $2$-$10$, the average posterior probabilities were below $0.2$ at any $z$, and the type I errors ranged from $0.05$ to $0.08$. In Section \ref{ssec:simulation_fda} we show results for a larger sample size, $M=100$ individuals. As illustrated in Table \ref{tab:simfda_10covar},  all type I errors are below $0.05$, and the statistical power is $\geq 0.99$ for all $z$. 
 For VC-BART, the type I error remained above 0.5.  


\section{Salary gaps versus age}
\label{ssec:salary}

\begin{figure}
\begin{center}
\begin{tabular}{cc}
\includegraphics[width=0.5\textwidth]{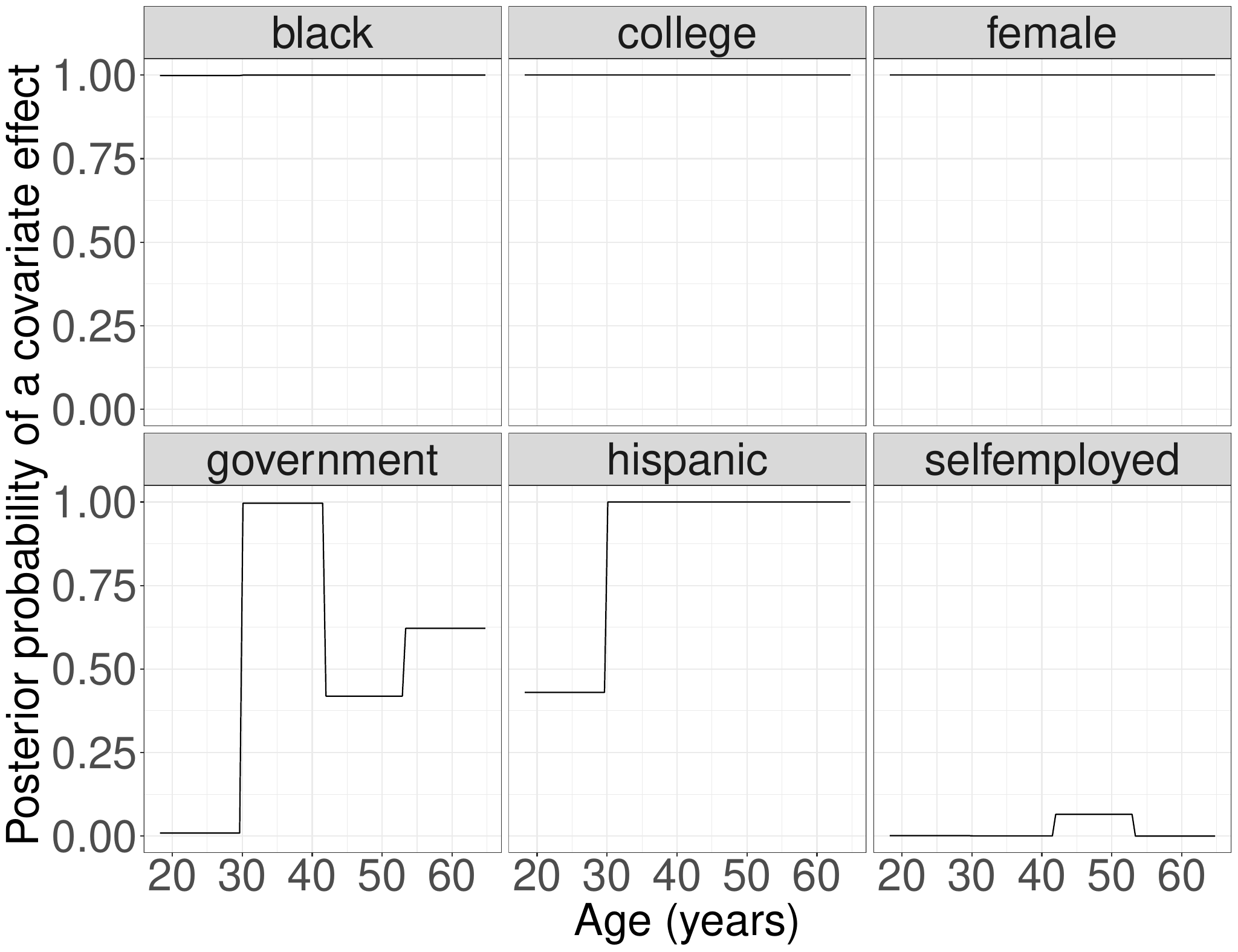} &
\includegraphics[width=0.5\textwidth]{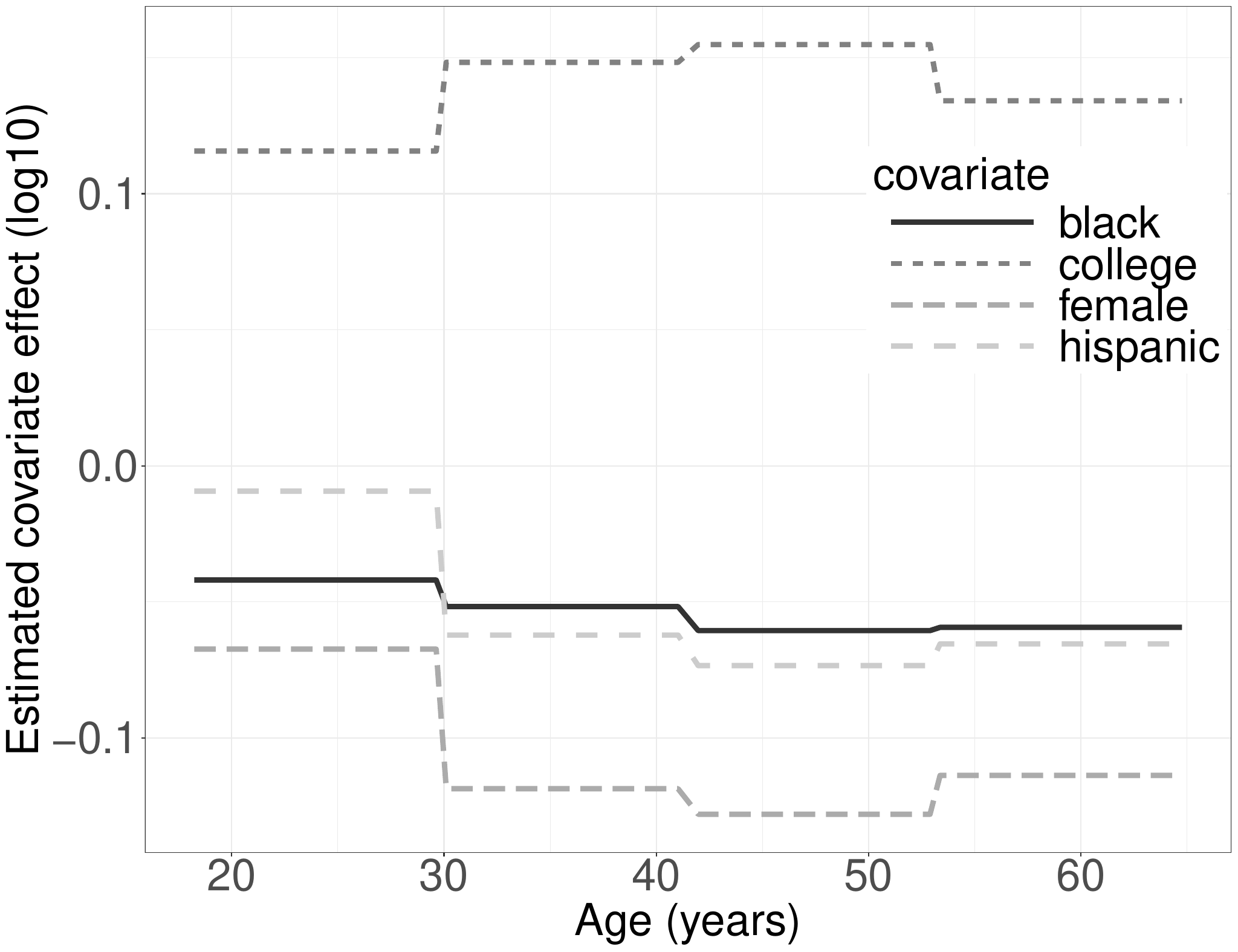}
\end{tabular}
\end{center}
\caption{Salary Data, discussed in Section \ref{ssec:salary}. Left: posterior probability of a salary gap associated with race, gender and college education versus age, adjusted by occupation and worker class. Right: corresponding estimated effect.}
\label{fig:salary}
\end{figure}

We used a dataset obtained from the 2019 Current Population Survey \citep{flood:2020}  to investigate the presence of a salary gap associated with potentially discriminatory factors, such as race (binary indicators for black race and Hispanic origin) and sex.
 Survey data require re-weighting for consistent estimation, but we do not pursue this here as the main goal was local testing. 
Specifically, the goal was to assess whether said gaps exist at all ages. The dataset includes $n=36,308$ single-race individuals aged 18-65 who were employed full-time (working 35-40 hours per week) and not serving in the military.
The response variable is the log-10 annual work income.  Besides assessing the effects of race and sex at various ages  ($z$ coordinates), we included the possession of a university degree as a benchmark that is expected to have a positive effect across all ages. We also examined the effects of being government-employed and self-employed. We also included the occupational sector as an adjustment covariate. The inclusion of the occupational sector is predetermined and no testing is performed. The data contain information on 24 occupational sectors, such as architect/engineer, computer/maths, construction, farming, food, maintenance, etc.


We performed a multi-resolution analysis using  four resolutions, defining 5, 10, 15 and 20 local testing regions respectively. The latter corresponds to 2.5 year bins, finer resolutions are unlikely of interest, as salary policies take years to have a noticeable effect.  
Figure \ref{fig:salary} summarizes the results. 
The existence of a black race, sex and college effect had a posterior probability near 1 at any age (left panel). 
College was associated with higher salaries, while black race and female sex were associated with lower salaries.
However, these gaps were estimated to be smaller for younger individuals, aged $\leq 30$ (right panel). 
Interestingly, we did not find strong evidence of salary gaps associated with Hispanic origins among individuals who had recently joined the workforce (aged $\leq 30$). However, strong evidence was found for such gaps at all other ages.
We emphasize  that the lack of evidence against a (local) null hypothesis does not prove its truth, as there may be limited power to detect such differences. However, given the large sample size, it appears that if there were truly a salary gap for individuals aged $\leq 30$, it should not be very large. 
We remark that these findings should be interpreted conditional on the assumption that two individuals are equal in terms of their occupational sector, college education, and other covariates. Our analysis is not designed to detect discrimination that may hinder individuals from obtaining a certain education or occupation.
 

Figure \ref{fig:salary_suppl} shows the sensitivity of our results to using a larger, and fixed, number of 20 and 30 knots.
With more knots statistical power should decrease, accordingly covariates black and hispanic show less evidence for local effects. 
For government and self-employed still no local effects are found, supporting that in our framework adding knots does not lead to false positives, whereas for college and female there continues to be evidence for local effects at most ages.

\section{Local brain activity over time}
\label{ssec:saez_data}

We consider a dataset from  \cite{saez:2018} measuring brain activity over time using multi-electrode electrocorticography.  
Patients played a game where they chose between a certain prize and a risky gamble. 
The goal is to assess whether certain game covariates are associated with brain activation and, if so, at what specific times. 
We consider eight covariates. Three are outcome-related: 1) whether the gamble resulted in a win or loss, 2) the reward prediction error (RPE) representing the difference between the obtained reward and the expected value of the gamble, and 3) the additional money that would have been won for the non-chosen option (regret). There was also a dummy covariate indicating whether participants chose to gamble or not and covariates about previous rounds (e.g. losing in the immediately preceding round).

Prior to time $z=0$ participants placed their bets. At exactly time $0$, the game's outcome (money gained) was revealed to the participants. The researchers recorded the brain activity before and after the outcome revelation, from $z=-450$ until $z=1500$ miliseconds.
Figure \ref{fig:ECoG} (left) shows the mean  high-frequency brain activity for an illustrative electrode averaged across 179 trials 
played by a single patient, split according to whether the game was won (top) or lost (bottom). 
See \cite{saez:2018} for further experimental and data pre-processing details.
 
First, we examine the local effect of each covariate separately. We  perform a multi-resolution analysis with 10, 20, 30 and 37 knots. The highest resolution corresponds to having a knot at every 50-ms interval, sufficient to capture changes in brain activity typically observed in this type of experiment (see \cite{saez:2018} for a discussion).
We find a significant increase in brain activity associated with game wins and regret. Specifically, the posterior probability of an effect for game wins is between 0.73 and 0.96 within the time range of $z \in (665,833)$. Similarly, we observe high probabilities for regret (between 0.89 and 1) in $z \in (665, 870)$. The estimated effects displayed temporal variability, gradually decreasing over time, see Figure \ref{fig:EcoG_effect_singles}. 

We next jointly include the eight covariates in our model, using multiple resolutions as above. 
This analysis identifies  regret 
as the only variable with a significant effect on brain activity, with a posterior probability of roughly 0.9 in the time range $z \in (720, 870)$ 
(Figure \ref{fig:ECoG}, right).
The differences with our previous analysis are likely because game wins (which before also had local effects) are highly correlated with regret.

Overall, our findings support that brain activity is associated to outcome-related covariates (gamble win, regret), rather than choice-related processes like gamble choice, and that it is important to account for the joint effect of multiple covariates.


\section{Discussion}
\label{sec:discussion}


Models assuming additive effects can provide more accurate inference than non-parametric methods when the sample size is moderate or there are many covariates, while facilitating interpretation and computation.  
If desired, one may also include covariate interactions, see the sex vs. race interaction in our salary example (Figure \ref{fig:salary_female_race}).
Although our orthogonal cut basis protect against type I error inflation in such models,
if one uses too many knots then power can suffer. However, our results support that the multi-resolution analysis can help improve power.

We focused on testing multiple scientific hypotheses, rather than estimating covariate effects. An alternative illustrated in Section \ref{sec:simulation_iid} is to use credible intervals for covariate selection. 
False positive inflation issues also apply to said intervals, unless one ensures the presence of zeroes in the asymptotic solution, e.g. our orthogonal cut basis.
Also, with many covariates, credible intervals can be inefficient compared to the exclusion of inactive covariates.

Our work enables the use of standard posterior inference algorithms.
These  algorithms were effective in our examples,
but computations become costlier when one considers many covariates, different resolutions, or as one adds local testing regions.
 Hence, with multivariate coordinates, restraint in setting many regions may be needed. 
When many covariate effects are zero (sparse truth) the posterior distribution concentrates on small models, which eases computation, but for non-sparse truths it would be interesting to develop tailored computations.
 For functional data we showed a proof-of-principle where one can use any covariance estimate $\hat{\Sigma}$ and then treats it as fixed. 
Fixing $\hat{\Sigma}$ underestimates uncertainty, but one still obtains consistent model selection and it enables
quickly computing marginal likelihoods via rank 1 Cholesky updates, and storing them whenever a model is revisited.
Developing statistically and computationally efficient covariance models that fully account for uncertainty would certainly be interesting future work.


Beyond our specific framework, we hope to shed light on subtle issues arising in local variable selection under misspecification, and help others put forward alternative solutions. 

\section*{Acknowledgments}

DR was funded by Ayudas Fundaci\'on BBVA Proyectos de Investigaci\'on Cient\'ifica en Matem\'aticas 2021, grant Consolidaci\'on investigadora CNS2022-135963 by the AEI, Europa Excelencia EUR2020-112096 from the AEI/10.13039/501100011033 and European Union NextGenerationEU/PRT.
ASK and MG were partially supported by NSF SES Award number 1659921.

\section*{Supplementary material}

The sections below contain all proofs, Lemma S1, and additional data analysis results. 
Additionally, folder 2023\_Rossell\_Kseung\_Guindani\_Saez\_localvarsel contains data and R scripts to reproduce our results, available at https://github.com/davidrusi/paper\_examples.




\section{Supplementary code}

R code to reproduce our examples is available at
{\small \texttt{https://github.com/davidrusi/paper\_examples/tree/main/2023\_Rossell\_Kseung\_Guindani\_Saez\_localvarsel}}

\section{Proof of Lemma \ref{lem:zero_indepmodel}}
\label{supplsec:proof_zero_indepmodel}

To ease notation let $\mu= (\mu_1^T,\ldots,\mu_{|\mathcal{R}|}^T)^T$ where $\mu_b= E_F(y_b \mid X,Z)$. We obtain
\begin{align}
\begin{pmatrix} \eta_0^* \\ \eta_1^* \end{pmatrix}= (W^T W)^{-1} W^T \mu=
\begin{pmatrix} (W_0^T W_0)^{-1} & 0 \\ 0 & (W_1^T W_1)^{-1} \end{pmatrix} \begin{pmatrix} W_0^T \\ W_1^T \end{pmatrix} \mu=
\begin{pmatrix} (W_0^T W_0)^{-1} W_0^T \\ (W_1^T W_1)^{-1} W_1^T \end{pmatrix} \mu,
\nonumber    
\end{align}
since $W_1^T W_0=0$.
Now, note that \eqref{eq:blockdiag_basis} implies that $W_1^TW_1$ is block-diagonal, i.e.  
\begin{align}
W_1^T W_1= \begin{pmatrix} 
             W_{11}^T W_{11} & 0 & \ldots & 0 \\
                          0 & W_{12}^T W_{12} & \ldots & 0 \\
                    \ldots &   &       & \ldots \\
                          0 & \ldots     &   &       W_{1|\mathcal{R}|}^T W_{1|\mathcal{R}|}
\end{pmatrix},
\nonumber
\end{align}
and hence
\begin{align}
&\begin{pmatrix} \eta_{11}^* \\ \ldots \\ \eta_{1|\mathcal{R}|}^* \end{pmatrix}= (W_1^T W_1)^{-1} W_1^T \mu= 
\nonumber \\
&\begin{pmatrix} 
             (W_{11}^T W_{11})^{-1} & 0 & \ldots & 0 \\
                          0 & (W_{12}^T W_{12})^{-1} & \ldots & 0 \\
                    \ldots &   &       & \ldots \\
                          0 & \ldots     &   &       (W_{1|\mathcal{R}|}^T W_{1|\mathcal{R}|})^{-1}
\end{pmatrix}
\begin{pmatrix} W_{11}^T & 0 & \ldots & 0 \\ 
0 & W_{12}^T & \ldots & 0 \\
\ldots & & &  \\ 
0 & 0 & \ldots & W_{1|\mathcal{R}|}^T \end{pmatrix} 
\begin{pmatrix} \mu_1 \\ \mu_2 \\ \ldots \\ \mu_{|\mathcal{R}|} \end{pmatrix}
\nonumber
\end{align}
and hence $\eta_{1b}^*= (W_{1b}^T W_{1b})^{-1} W_{1b}^T \mu_b$, as we wished to prove.

\section{Orthogonal cut basis with multiple covariates}
\label{supplsec:cutbasis_multiplecovar}

Recall that our assumed model is
\begin{align}
 y \mid Z, X \sim N(W_0 \eta_0 + W_1 \eta_1, \Sigma),
\nonumber
\end{align}
where $\Sigma$ is as in \eqref{eq:vcmodel_1covar_matrix} and $(W_0,W_1)$ is an orthogonal cut matrix as in \eqref{eq:blockdiag_basis}. Above $W_{1b}$ is obtained by multiplying a local cut basis by each covariate. 
Specifically,
\begin{align}
 W_{1b}= \begin{pmatrix} C_b \odot X_{b 1}, & \ldots, & C_b \odot X_{b p} \end{pmatrix},
\nonumber
\end{align}
where $C_b$ is the cut basis of region $b$, $X_{b j}$ is the column vector comprising the values of covariate $j$ for observations in region $b$, 
and $C_b \odot X_{b j}$ denotes the column-wise product obtained by multiplying each column in $C_b$ by $X_{b j}$, in an entry-wise fashion.

Let $\eta_{1j}=(\eta_{1j1}^T,\ldots,\eta_{1j|\mathcal{R}|}^T)^T$ where $\eta_{1jb} \in \mathbb{R}^p$ is the coefficient for the effect of covariate $j$ in region $b$. Then
the model-based local null hypothesis for covariate $j$ in region $R_b$ is 
\begin{align}
 \beta_{1j}(z)= 0 \mbox{ for all } z \in R_b  \Longleftrightarrow  \eta_{1jb} = 0.
\label{eq:localnull_additive}
\end{align}

Since $W_1$ is an orthogonal cut basis, Lemmas \ref{lem:zero_indepmodel} and \ref{lem:zero_depmodel}  directly apply to \eqref{eq:localnull_additive}. In fact, the statement and proof of Lemma \ref{lem:zero_depmodel} are given explicitly multiple covariate case.
This means that  if $E_F(y_b \mid X,Z)$ is linearly independent of covariate $j$ given the other covariates, then the asymptotic covariate effect in region $b$ is $\eta_{1jb}^*=0$ (for dependent data, $\tilde{\eta}_{1jb}^*=0$).

\section{Lemma \ref{lem:zero_depmodel}}
\label{supplsec:zero_depmodel}

Lemma \ref{lem:zero_depmodel} extends Lemma \ref{lem:zero_indepmodel} to dependent data settings where 
for each individual one observes data on a grid (e.g. longitudinal or functional data).
Specifically, let $y_i$ be observation $i=1,\ldots,n$, $m_i=1,\ldots,M$ denote the individual and $M$ the number of individuals, 
and $x_i=x_{m_i}$ be the covariates for individual $m_i$ (in particular, for group comparisons $x_{m_i}$ is a vector coding for group membership).
Suppose that for each individual we have the same number of measurements $n/M$.

We first recall the notation and setup. Recall that the assumed model is
\begin{align}
 y \mid Z, X \sim N \left( W_0 \eta_0 + W_1 \eta_1, \Sigma \right),
\nonumber
\end{align}
where $(W_0,W_1)$ are orthogonal cut basis as defined in Section \ref{ssec:misspec}. Recall that these can be written as
\begin{align} 
W_0= \begin{pmatrix} W_{01} \\ W_{02} \\ \ldots \\ W_{0|\mathcal{R}|} \end{pmatrix}
; \hspace{5mm}
W_1= \begin{pmatrix} 
W_{11} & 0 & \ldots &  0 \\
0 & W_{12} & \ldots & 0 \\
\ldots & & & \\
0 & 0 & \ldots & W_{1|\mathcal{R}|},
\end{pmatrix}.
\label{seq:blockdiag_basis}
\end{align}
where $W_{1b}$ is the basis evaluated at region $b=1,\ldots,|\mathcal{R}|$ interacted with the covariate effects.
Specifically,
\begin{align}
  W_{1b}= \begin{pmatrix} C_b \odot X_{b 1}, & \ldots, & C_b \odot X_{b p} \end{pmatrix},
\nonumber
\end{align}
where $C_b$ is the cut basis of region $b$, $X_{b j}$ the column vector with the values of covariate $j$ for observations in region $b$, 
and $C_b \odot X_{b j}$ the column-wise product multiplying each column in $C_b$ by $X_{b j}$, in an entry-wise fashion.

Finally, recall also that $\Sigma$ is assumed to have a block-diagonal structure across regions, i.e.
\begin{align}
 \Sigma= \begin{pmatrix} \Sigma_1 & 0 & \ldots & 0 \\
0 & \Sigma_2 & \ldots & 0 \\
\ldots & & & \\
0 & 0 & \ldots & \Sigma_{|\mathcal{R}|} \end{pmatrix}.
\nonumber
\end{align}

Denote by $\widetilde{\eta}^*=(\widetilde{\eta}_0^*,\widetilde{\eta}_1^*)$ the Kullback-Leibler optimal parameter value under the data-generating truth $F$. That is,
\begin{align}
 \widetilde{\eta}^*&= \arg\min_\eta E_F \left[ (y - W\eta)^T \Sigma^{-1} (y - W\eta) \mid Z,X \right]=
(W^T \Sigma^{-1} W)^{-1} W^T \Sigma^{-1} E_F(y \mid Z,X)
\nonumber
\end{align}

\begin{lemma}
Let $W=(W_0,W_1)$ where $W_1$ is orthogonal cut basis as in \eqref{eq:blockdiag_basis} satisfying $W_0^T W_1=0$.
Suppose that for each individual one has the same number of measurements $n/M$,
and that covariate values add to zero across individuals, i.e. $\sum_{m=1}^M x_{s_i} = 0$.

Let $\widetilde{\eta}_1^*= (\widetilde{\eta}_{11},\ldots,\widetilde{\eta}_{1|\mathcal{R}|}^*)$ where $\widetilde{\eta}_{1b}$ are the parameters associated to $W_{1b}$ (i.e. to region $b$).
Then, 
\begin{align}
 \widetilde{\eta}_{1b}^*= (W_{1b}^T \Sigma_b^{-1} W_{1b})^{-1} W_{1b}^T \Sigma_b^{-1} E_F(y_b \mid X,Z).
\nonumber
\end{align}

\label{lem:zero_depmodel}
\end{lemma}

\begin{proof}
Let $\widetilde{W}= \Sigma^{-1/2} W$ and $\tilde{y}= \Sigma^{-1/2} y$, the optimal solution can then be written as
\begin{align}
  \widetilde{\eta}^*&= (\widetilde{W}^T \widetilde{W})^{-1} \widetilde{W}^T E_F(\widetilde{y} \mid Z,X).
\nonumber
\end{align}

The proof strategy is to show that the conditions of Lemma \ref{lem:zero_indepmodel} hold for $\widetilde{W}=(\widetilde{W}_0,\widetilde{W}_1)=(\Sigma^{-1/2} W_0, \Sigma^{-1} W_1)$, which then immediately gives that
\begin{align}
  \widetilde{\eta}_{1b}^*= (\widetilde{W}_{1b}^T \widetilde{W}_{1b})^{-1} \widetilde{W}_{1b}^T E_F(\tilde{y}_b \mid X,Z)=
 (W_{1b}^T \Sigma_b^{-1} W_{1b})^{-1} W_{1b}^T \Sigma_b^{-1} E_F(y_b \mid X,Z),
\nonumber
\end{align}
where we used that $\tilde{y}_b= \Sigma^{-1/2} y_b$, proving Lemma \ref{lem:zero_depmodel}.

The first condition in Lemma \ref{lem:zero_indepmodel} is that $\widetilde{W}_1$ is a cut matrix. To see that this holds, note that
$ \widetilde{W}_1= \Sigma^{-1/2} W_1=$
\begin{align}
\begin{pmatrix} \Sigma_1 & 0 & \ldots & 0 \\
0 & \Sigma_2 & \ldots & 0 \\
\ldots & & & \\
0 & 0 & \ldots & \Sigma_{|\mathcal{R}|} \end{pmatrix}
\begin{pmatrix} 
W_{11} & 0 & \ldots &  0 \\
0 & W_{12} & \ldots & 0 \\
\ldots & & & \\
0 & 0 & \ldots & W_{1|\mathcal{R}|}
\end{pmatrix}
\nonumber
=\begin{pmatrix} 
\widetilde{W}_{11} & 0 & \ldots &  0 \\
0 & \widetilde{W}_{12} & \ldots & 0 \\
\ldots & & & \\
0 & 0 & \ldots & \widetilde{W}_{1|\mathcal{R}|},
\end{pmatrix}
\end{align}
where $\widetilde{W}_{1b}= \Sigma_b^{-1} W_{1b}$. 
This proves that $\widetilde{W}_1$ has the same structure as $W_1$ in \ref{seq:blockdiag_basis}, i.e. $\widetilde{W}_1$ is a cut basis.

The second condition in Lemma \ref{lem:zero_indepmodel} is that $\widetilde{W}_1$ is an orthogonal basis to $\widetilde{W}_0$, i.e. $\widetilde{W}_0^T \widetilde{W}_1= 0$. To show that this condition holds, first note that
\begin{align}
 \widetilde{W}_0^T \widetilde{W}_1
=
\widetilde{W}_0^T
\begin{pmatrix} 
\widetilde{W}_{11} & 0 & \ldots &  0 \\
0 & \widetilde{W}_{12} & \ldots & 0 \\
\ldots & & & \\
0 & 0 & \ldots & \widetilde{W}_{1|\mathcal{R}|},
\end{pmatrix}
=
\begin{pmatrix} 
\widetilde{W}_{01}^T \widetilde{W}_{11} \\
\widetilde{W}_{02}^T \widetilde{W}_{12} \\
\ldots \\
\widetilde{W}_{0|\mathcal{R}|}^T \widetilde{W}_{1|\mathcal{R}|},
\end{pmatrix},
\nonumber
\end{align}
where $\widetilde{W}_{0b}$ are the rows from $\widetilde{W}_0$ corresponding to region $b$.
Hence, it suffices to show that $\widetilde{W}_{0b}^T \widetilde{W}_{1b}=0$ for all $b$.
To see this, note that $\widetilde{W}_{0b}^T \widetilde{W}_{1b}$ can be computed by summing across individuals $s=1,\ldots,m$. That is,
\begin{align}
 \widetilde{W}_{0b}^T \widetilde{W}_{1b}
=\sum_{m=1}^M  \widetilde{W}_{0bm}^T \widetilde{W}_{1bm}.
\nonumber
\end{align}
Since all individuals are observed on the same grid we have that $W_{0bm}$ is constant across $m$, and similarly $\Sigma_{bm}$ is assumed constant across individuals, hence we may write $\widetilde{W}_{0bm}= \Sigma_{bm}^{-1/2} W_{0bm}= T_{0b}$, where $T_{0b}$ does not depend on the individual $m$.
Similarly, 
\begin{align}
\widetilde{W}_{1bm}= \Sigma_{bm}^{-1/2} W_{1bm}=
\Sigma_{bm}^{-1/2} \begin{pmatrix} C_{bm} \odot X_{bm 1}, & \ldots, & C_{bm} \odot X_{bm p} \end{pmatrix} 
= T_{1b} x_m
\nonumber
\end{align}
where $X_{b m j}$ is the column vector with the value of covariate $j$ for individual $m$, that is constant across observations in region $b$ and equal to $x_{mj}$,
and $T_{1b}= \Sigma_{bm}^{-1} C_{bm}$ does not depend on $m$ (since all individuals are observed on the same grid, $C_{bm}$ is constant across individuals indexed by $m$). Recall that $x_m= (x_{m1},\ldots,x_{mp})$ denotes the covariates for individual $m$.
We therefore obtain that
\begin{align}
 \widetilde{W}_{0b}^T \widetilde{W}_{1b}
=\sum_{m=1}^M  \widetilde{W}_{0bm}^T \widetilde{W}_{1bm}=
T_{0b}^T T_{1b} \sum_{m=1}^M x_m= 0,
\nonumber
\end{align}
since $\sum_{m=1}^M x_m=0$ (the covariates have zero sample mean) by assumption, completing the proof.

\end{proof}

\section{Handling partitions of group means defined by categorical covariates}

Our framework is primarily designed for binary and continuous covariates, and discrete covariates with infinite support that can be treated as continuous (possibly after a suitable transformation) or discretized into several categories.
If there is a discrete covariate defining $K$ groups, our approach effectively tests whether each group deviates from the overall mean.

There are however situations where one may be interested in learning whether a subset of the categories have the same expectation.
For simplicity, consider a covariate with 3 categories $x_i \in \{A,B,C\}$.
One may wish to assess whether the outcome has the same mean in groups A and B, and different than in group C, or perhaps whether the mean follows some other partition into groups.
Dealing with this situation when there are many groups requires separate developments beyond our scope, such as strategies to search over all possible sub-groups, and defining suitable priors on group partitions.

However, it is possible to accommodate situations with a few groups into our framework. 
The idea is to add new columns to the design matrix based on all the possible group combinations, and to then restrict the model space so that only certain combinations are allowed.
For concreteness, we focus on the 3 group case.
Let $b_{i1}= \mbox{I}(x_i = A)$, $b_{i2}= \mbox{I}(x_i= B)$, $b_{i3}= \mbox{I}(x_i= C)$ be indicators for the 3 groups. 
Let $b_{i4}= b_{i1} + b_{i2}= \mbox{I}(x_i \in \{A,B\})$, 
$b_{i5}= b_{i1} + b_{i3}= \mbox{I}(x_i \in \{A,C\})$,
and $b_{i6}= b_{i2} + b_{i3}= \mbox{I}(x_i \in \{B,C\})$.
For precision, as discussed in the paper, for identifiability in the actual design matrix $W_1$ we drop one category and orthogonalize the rest with respect to the baseline $W_0$ (e.g. if $W_0$ has 0 degree splines, then $W_1$ contains mean-centered versions of $b_{i1},b_{i2},\ldots$). 
Suppose that we drop $b_{i3}$ (the C category), and define the $i^{th}$ row in the design matrix $W_1$ as
$w_{i1}= b_{i1} - \hat{b}_{i1}$, $w_{i2}= b_{i2} - \hat{b}_{i2}$, $w_{i3}= b_{i4} - \hat{b}_{i4}$, $w_{i4}= b_{i5}-\hat{b}_{i5}$, $w_{i5}= b_{i6} - \hat{b}_{i6}$, where the $\hat{b}$'s are given by projecting the $b$'s onto $W_0$ (i.e. their least-squares prediction).
Let $\eta_1 \in \mathbb{R}^5$ be the corresponding regression coefficients, and $\gamma_j= \mbox{I}(\eta_{1j} \neq 0)$ the inclusion indicators defining the models.
The idea is to restrict the model space such that one considers separately the inclusion of $(b_{i1},b_{i2})$, which correspond to individual groups having different means,
 and that of $(b_{i4},b_{i5},b_{i6})$, which correspond to several groups having the same mean.
Specifically, in this example one would consider the following constrained model space, which features 5 models:
\begin{center}
\begin{tabular}{cccccc}
$\gamma_1$ & $\gamma_2$ & $\gamma_3$ & $\gamma_4$ & $\gamma_5$ & \\
0 & 0 & 0 & 0 & 0  & $\mu_A=\mu_B=\mu_C$  \\
1 & 1 & 0 & 0 & 0  & $\mu_A\neq \mu_B$, $\mu_A \neq \mu_C$, $\mu_B \neq \mu_C$  \\
0 & 0 & 1 & 0 & 0  & $\mu_A= \mu_B \neq \mu_C$  \\
0 & 0 & 0 & 1 & 0  & $\mu_A= \mu_C \neq \mu_B$  \\
0 & 0 & 0 & 0 & 1  & $\mu_A \neq \mu_B = \mu_C$  \\
\end{tabular}
\end{center}

One can then easily evaluate the marginal likelihood and prior probability for each of these 5 models, to obtain their posterior probabilities.

\section{Comparison to Haar basis}
\label{supplsec:haar_basis}

A possible alternative to our cut basis framework is to use Wavelets, and in particular using Haar basis presents analogies to our framework when using degree 0 splines. 
While using wavelets is a possible strategy, it is less convenient because then the selection of local effects would no longer correspond to selecting zero coefficients, as is the case with our cut orthogonal basis. 
As an example, consider the interval $z \in [0,1]$ and suppose that one uses a Haar basis with 3 resolutions to capture the differences between 2 groups (without loss of generality, we ignore the intercept in this discussion). Specifically, let the local group differences at $z$ be
\begin{align}
\sum_{r=1}^3 \sum_{j=1}^{2^r} \beta_{rj} w_{rj}(z)
\nonumber
\end{align}
where $\beta_{rj}$ is the coefficient for basis $j=1,\ldots,2^r$ in resolution $r=1,2,3$, and the basis is given by stepwise functions
\begin{align}
&w_{11}(z)= \mbox{I}(z \leq 1/2); w_{12}(z)= \mbox{I}(z > 1/2)
\nonumber \\
&w_{21}(z)= \mbox{I}(z \leq 1/4); w_{22}(z)= \mbox{I}(z \in (1/4,1/2]); \ldots; w_{28}(z)= \mbox{I}(z \in (3/4,1])
\nonumber \\
&w_{31}(z)= \mbox{I}(z \leq 1/8); w_{32}(z)= \mbox{I}(z \in (2/8,3/8]); \ldots; w_{38}(z)= \mbox{I}(z \in (7/8,1])
\nonumber
\end{align}

Then, testing for a local null effect at $z=0.1$ means testing the null hypothesis $H_0: \beta_{11} + \beta_{21} + \beta_{31}= 0$, 
whereas at $z=0.2$ it means testing $H_0: \beta_{11} + \beta_{21} + \beta_{32}= 0$.
It is of course possible to test for such linear parameter combinations to be zero, but our setting is more convenient in that one can simply test for zeroes in individual parameters (as you noted above), which allows one to directly use standard penalization or Bayesian model selection methods.

Also of interest in a Bayesian setting, our formulation facilitates specifying prior knowledge on the number of intervals where a local effect is expected, as this is directly given by the number of non-zero parameters, while this would be less straightforward using wavelets.

\section{Prior distribution on the parameters}
\label{supplsec:prior_parameters}

 As discussed, we consider two choices for the prior on the coefficients. 
First, we consider a Normal shrinkage prior
\begin{align}
  p(\eta_\gamma \mid \gamma)= N\left( \eta_\gamma; 0, g\,  \mbox{diag}(W_\gamma^TW_\gamma/n)^{-1} \right)
\label{eq:prior_normalshinkage}
\end{align}
where $g>0$ is the prior dispersion. By default, we set $g=1$  so that the trace of the prior precision equals that of the unit information prior of \cite{schwarz:1978}. 
Prior \eqref{eq:prior_paramters} can be viewed as imposing a penalty on the $L_2$ norm of $\eta$,
and further our multi-resolution analysis (Section \ref{ssec:varying_resolutions}) encourages smoothness across coordinates $Z$ in the fitted regression.

Second, we consider a prior that encourages smoothness in the estimated coefficients across the coordinates $z$.
The main idea is that intrinsic conditionally auto-regressive (ICAR) priors are improper and hence cannot be directly used for model selection (else, one incurs the so-called Jeffreys-Lindley-Bartlett paradox), hence we add a diagonal matrix to the prior precision.
Specifically, our novel prior, which we denominate ICAR+, considers
\begin{align}
p(\eta_\gamma \mid \gamma)=
N\left(\eta_{0\gamma}; 0, g S_{0\gamma}^{-1}   \right)
\prod_{j=1}^p N\left(\eta_{1j\gamma}; 0, g S_{1j\gamma}^{-1}   \right)
\nonumber
\end{align}
where $S_{0\gamma}$ and  $S_{1j\gamma}$ are the submatrices of
\begin{align}
&S_0= a P_0 + (1-a) n \mbox{diag}(W_0^T W_0)^{-1}
\nonumber \\
&S_{1j}= a P_1 + (1-a) n \mbox{diag}(W_{1j}^T W_{1j})^{-1}
\nonumber
\end{align}
obtained by selecting the row and columns with non-zero entries in $\gamma_0$ and $\gamma_{1j}$ respectively.
$P_0$ and $P_1$ are ICAR precision matrices (standardized to have trace equal to the unit information prior, i.e. $\mbox{dim}(\eta_{\gamma})$), and $a \in [0,1]$ with a default $a=0.5$.
Specifically, $P_0= D_0^T D_0 c_0$, where $D_0$ is a matrix that takes the difference between each entry in $\eta$ and the average of its neighbours according to the coordinates in $z$.
That is, $D_0$ is a $\mbox{dim}(\eta_0) \times \mbox{dim}(\eta_0)$ matrix with $(i,i)$ entry equal to 1 and $(i,j)$ entry $-1/M_i$, where $M_i$ is the number of neighbors of $i$, and $c_0= \{\mbox{dim}(\eta_0) / \mbox{tr}(D_0^TD_0)\}^{1/2}$ so that $\mbox{tr}(P_0)=\mbox{dim}(\eta_0)$. We define $P_1= D_1^T D_1 c_1$ analagously. Hence
\begin{align}
D_0 \eta_0= \begin{pmatrix} 
\eta_{01} - \frac{1}{N_1} \sum_{i \sim 1} \eta_{0i}   \\
\eta_{02} - \frac{1}{N_2} \sum_{i \sim 2} \eta_{0i}   \\
\ldots
\end{pmatrix}
\nonumber
\end{align}
where $i \sim j$ denotes that $i$ is a neighbor of $j$.
Assuming that the columns in $W_0$ have unit standard deviation, this gives that
$N\left(\eta_0; 0, g S_{0}^{-1}   \right)=$
\begin{align}
\frac{|S_0|^{1/2}}{(2\pi g)^{\mbox{dim}(\eta_0)/2}|}
\exp \left\{ \frac{1}{2g} \left[ 
a \sum_{j=1}^{\mbox{dim}(\eta_0)} \left( \eta_{0j} - \frac{1}{N_j} \sum_{i \sim j} \eta_{0i}  \right)^2
+ (1-a) \sum_{j=1}^{\mbox{dim}(\eta_0)} \eta_{0j}^2 \right] \right\}.
\nonumber
\end{align}
That is, the prior encourages each coefficient $\eta_{0j}$ and $\eta_{1j}$ to be close to the average of its neighbors, and it also penalizes its $L_2$ norm.
This is particularly interesting for $\eta_1$: since $\eta_1$ quantifies deviations from the mean (due to our orthogonalization step in defining the basis $W_1$), by penalizing the $L_2$ norm of $\eta_1$, the prior penalizes deviations from the baseline mean, i.e. encourages sparsity in the local covariate effects.

%

Our construction ensures that the prior precision for each entry in $\eta_\gamma$ is equal to $g$, which allows setting $g=1$ as a minimally informative default that mimics the unit information prior. Also, since the precision matrices $S_{0\gamma}$ and $S_{1\gamma}$ are obtained as a subset of a global $S_0$, $S_{1j}$, one can use fast rank 1 Cholesky updates that deliver significant savings when computing marginal likelihoods.

We remark that in the univariate case P-splines penalties for cubic splines \cite{eilers:1996} result in a prior precision matrix that is very similar to $P_0$ and $P_1$ above.

We also remark that a common strategy in the literature is to place priors on hyper-parameters such as $a$ or $g$. The issue is that this incurs a non-negligible computational cost, as one needs to update hyper-parameters within the model search, one cannot use the fast Cholesky updates and it is also not possible to save marginal likelihoods for each considered model in the MCMC search (since marginal likelihoods would be now conditional on hyper-parameters). Given that for purposes of (local) variable selection, the effect of such priors is mild (our model selection rates essentially remain unaltered), here we prefer to set $(a,g)$ to reasonable default values.

\section{Bayes factor derivation}
\label{supplsec:bf_derivation}

The Bayes factor comparing a model $\gamma$ with the optimal $\gamma^*$ is
\begin{align}
	\begin{split}
 B_{\gamma \gamma^*}&=
\frac{\left|g V_{\gamma^*} W_{\gamma^*}^T \Sigma^{-1} W_{\gamma^*}  + I \right|^{\frac{1}{2}}}{\left| g V_{\gamma} W_{\gamma}^T \Sigma^{-1} W_\gamma  + I \right|^{\frac{1}{2}}} \label{eq:bf_theorem}
\\
&\times \exp \left\{  \frac{1}{2} \left[ \hat{\eta}_\gamma^T ( W_\gamma^T \Sigma^{-1} W_\gamma + (g V_\gamma)^{-1}) \hat{\eta}_\gamma 
- \hat{\eta}_{\gamma^*}^T ( W_{\gamma^*}^T \Sigma^{-1} W_{\gamma^*} + (g V_{\gamma^*})^{-1}) \hat{\eta}_{\gamma^*}
\right] \right\},
\end{split}
\end{align}
where $\hat{\eta}_\gamma= E(\eta_\gamma \mid y,\gamma)= (W_\gamma^T \Sigma^{-1} W_\gamma + (g V_\gamma)^{-1})^{-1} W_\gamma^T \Sigma^{-1} y$. 

We derive Expression \eqref{eq:bf_theorem}. 
Recall that the assumed model is
\begin{align}
 y \mid \gamma &\sim N(W_\gamma \eta_\gamma, \Sigma)
\nonumber \\
 \beta_\gamma &\sim N(0, g V_\gamma),
\nonumber
\end{align}
where $\Sigma$ is an $n \times n$ and $V_\gamma$ a $|\gamma|_0 \times |\gamma|_0$ positive-definite matrix, and $g \in \mathbb{R}^+$.

The integrated likelihood under model $\gamma$ is hence
\begin{align}
 p(y \mid \gamma)= \int \frac{1}{(2\pi)^{\frac{n}{2}} |\Sigma|^{\frac{1}{2}} |g V_\gamma|^{\frac{1}{2}}} 
e^{  -\frac{1}{2} (y - W_\gamma \eta_\gamma)^T \Sigma^{-1} (y - W_\gamma \eta_\gamma)}
\frac{1}{(2\pi)^{\frac{|\gamma|_0}{2}}} 
e^{ -\frac{1}{2} \eta_\gamma^T (gV_\gamma)^{-1} \eta_\gamma}
d \eta_\gamma
\nonumber \\
=\frac{e^{-\frac{1}{2}y^T \Sigma^{-1} y}}{(2\pi)^{\frac{n}{2}} |\Sigma|^{\frac{1}{2}} |g V_\gamma|^{\frac{1}{2}}} 
\int \frac{1}{(2\pi)^{\frac{|\gamma|_0}{2}} } 
e^{-\frac{1}{2} \left[ \eta_\gamma^T (W_\gamma^T \Sigma^{-1} W_\gamma + (g V_\gamma)^{-1}) \eta_\gamma - 2 y^T \Sigma^{-1} W_\gamma \eta_\gamma  \right]}
d \eta_\gamma.
\nonumber
\end{align}
Denoting by $V_{post}= (W_\gamma^T \Sigma^{-1} W_\gamma + (g V_\gamma)^{-1})^{-1}$ and
$\hat{\eta}_\gamma= (W_\gamma^T \Sigma^{-1} W_\gamma + (g V_\gamma)^{-1})^{-1} W_\gamma^T \Sigma^{-1} y$, we obtain
\begin{align}
  p(y \mid \gamma)= 
\frac{e^{-\frac{1}{2}y^T \Sigma^{-1} y} e^{\frac{1}{2}\hat{\eta}_\gamma^T V_{post}^{-1} \hat{\eta}_\gamma}}{(2\pi)^{\frac{n}{2}} |\Sigma|^{\frac{1}{2}} |g V_\gamma|^{\frac{1}{2}}} 
\int \frac{1}{(2\pi)^{\frac{|\gamma|_0}{2}} } 
e^{-\frac{1}{2} \left[ \eta_\gamma^T V_{post}^{-1} \eta_\gamma - 2 y^T \Sigma^{-1} W_\gamma V_{post} V_{post}^{-1} \eta_\gamma + \hat{\eta}_\gamma^T V_{post}^{-1} \hat{\eta}_\gamma  \right]}
d \eta_\gamma=
\nonumber \\
\frac{e^{-\frac{1}{2}y^T \Sigma^{-1} y} e^{\frac{1}{2}\hat{\eta}_\gamma^T V_{post}^{-1} \hat{\eta}_\gamma}}{(2\pi)^{\frac{n}{2}} |\Sigma|^{\frac{1}{2}}  |g V_\gamma|^{\frac{1}{2}}} 
\int \frac{1}{(2\pi)^{\frac{|\gamma|_0}{2}}} 
e^{-\frac{1}{2} (\eta_\gamma - \hat{\eta}_\gamma)^T V_{post}^{-1} (\eta_\gamma - \hat{\eta}_\gamma)}
d \eta_\gamma=
\frac{e^{-\frac{1}{2}y^T \Sigma^{-1} y} e^{\frac{1}{2}\hat{\eta}_\gamma^T V_{post}^{-1} \hat{\eta}_\gamma} |V_{post}|^{\frac{1}{2}}}{(2\pi)^{\frac{n}{2}} |\Sigma|^{\frac{1}{2}}  |g V_\gamma|^{\frac{1}{2}}}.
\nonumber
\end{align}

Hence the Bayes factor is $ B_{\gamma \gamma^*}= p(y \mid \gamma) / p(y \mid \gamma^*)=$
\begin{align}
e^{\frac{1}{2}[\hat{\eta}_\gamma^T (W_\gamma^T \Sigma^{-1} W_\gamma + (g V_\gamma)^{-1}) \hat{\eta}_\gamma - \hat{\eta}_{\gamma^*}^T (W_{\gamma^*}^T \Sigma^{-1} W_{\gamma^*} + (g V_{\gamma^*})^{-1}) \hat{\eta}_{\gamma^*}]}
\frac{|(W_{\gamma^*}^T \Sigma^{-1} W_{\gamma^*} + (g V_{\gamma^*})^{-1})g V_{\gamma^*}|^{\frac{1}{2}}}{|(W_\gamma^T \Sigma^{-1} W_\gamma + (g V_\gamma)^{-1})g V_\gamma|^{\frac{1}{2}}},
\nonumber
\end{align}
as we wished to prove.

\newpage
\section{Discussion of the technical conditions of Theorem \ref{thm:bf}}
\label{supplsec:conditions_thm_bf}

Assumption (A1) implies that, although the number of parameters $q$ may be larger than $n$, we restrict our attention to full-rank models. This implies that the model size $|\gamma|_0 \leq n$, a common practice in model selection.  Assumption (A2) is a mild and can be relaxed, but simplifies our exposition. In our default setting, we assume $V_\gamma=I$. Thus (A2) is satisfied when the empirical covariance $W_\gamma^T \Sigma^{-1} W_\gamma/n$ converges to a fixed positive definite covariance. The assumption is also satisfied by Zellner's prior, where $V_\gamma^{-1}= W_\gamma^T \Sigma^{-1} W_\gamma/n$, so that $\underline{l}_\gamma=\bar{l}_\gamma$. Assumption (A3) holds in the homoskedastic independent errors setting where $\Sigma=\sigma^2 I$. For dependent data, (A3) is a mild condition that may be relaxed to accommodate cases where $\tau$ is not bounded as $n$ grows, see the proof of Theorem \ref{thm:bf}.  Assumption (A4) is minimal and ensures that  the prior dispersion $g$ does not vanish too fast with $n$ (it is satisfied by our default $g=1$ and by the discussed alternatives where $g$ may grow with $n$). In Assumption (A5) the parameter $\lambda_\gamma$ can be interpreted as a non-centrality parameter that  measures the sum of squares explained by model $\gamma^*$ but not by model $\gamma$. Assumption (A5) is minimal: otherwise,  Bayes factors are not consistent even in finite-dimensional settings with fixed $|\gamma|_0$. See also the discussion of Assumption (B6) in Section \ref{ssec:pp} on the relationship between $\lambda_\gamma$ and common beta-min and restricted eigenvalue conditions.


\newpage
\section{Auxiliary results for Theorem \ref{thm:bf}}
\label{supplsec:auxiliary_results}

We recall auxiliary results that are helpful in the proofs of our main results.
Lemmas \ref{lem:waldstat_nested}-\ref{lem:lefttail_quadform_subgaussian} are obtained from Sections S6 and S9 in \cite{rossell:2021}.
For brevity we do not reproduce their proofs, please see Sections S6 and S9 in \cite{rossell:2021}.
However, we do prove Lemma \ref{lem:lefttail_quadform_subgaussian} here since, although the result in \cite{rossell:2021} is correct, the proof offered there contains an error.

Lemma \ref{lem:subgaussian_difquad_tail} is a new result that we prove here bounding tail probabilities for certain differences between quadratic forms of sub-Gaussian vectors.

Lemma \ref{lem:waldstat_nested} is a well-known result. It expresses the difference between the explained sum of squares by two nested models $\gamma' \subset \gamma$ in terms of the regression parameters obtained under the larger model, after suitably orthogonalizing its design matrix.

Definition \ref{def:subgaussian} and Lemmas \ref{lem:subgaussian_lincomb}-\ref{lem:subgaussian_quadform} give the definition and two basic properties of 
sub-Gaussian random vectors: closed-ness under linear combinations and hat certain quadratic forms of $n$-dimensional sub-Gaussian vectors can be re-expressed as a quadratic form of $d$-dimensional sub-Gaussian vectors.
Lemmas \ref{lem:tail_quadform_subgaussian}-\ref{lem:lefttail_quadform_subgaussian} give bounds for right and left sub-Gaussian tail probabilities.

\begin{lemma}
  Let $y \in \mathbb{R}^n$ and $W_\gamma= (W_{\gamma'}, W_{\gamma \setminus \gamma'})$ be an $n$ times $|\gamma|_0$ matrix.
  Let $\hat{\eta}_\gamma= (W_\gamma^T W_\gamma)^{-1} W_\gamma^T y$
  and $\hat{\eta}_{\gamma'}= (W_{\gamma'}^T W_{\gamma'})^{-1} W_{\gamma'}^T y$
  the least-squares estimates associated to $\gamma$ and $\gamma'$, respectively.
Then
$$  \hat{\eta}_{\gamma}^T W_{\gamma}^TW_{\gamma} \hat{\eta}_{\gamma} - \hat{\eta}_{\gamma'}^T W_{\gamma'}^TW_{\gamma'} \hat{\eta}_{\gamma'}
=\tilde{\eta}_{\gamma \setminus \gamma'}^T \tilde{W}_{\gamma \setminus \gamma'}^T\tilde{W}_{\gamma \setminus \gamma'} \tilde{\eta}_{\gamma \setminus \gamma'}=
y^T \tilde{W}_{\gamma \setminus \gamma'} (\tilde{W}_{\gamma \setminus \gamma'}^T \tilde{W}_{\gamma \setminus \gamma'})^{-1} \tilde{W}_{\gamma \setminus \gamma'}^T y
$$
where $\tilde{\eta}_{\gamma \setminus \gamma'}= (\tilde{W}_{\gamma \setminus \gamma'}^T\tilde{W}_{\gamma \setminus \gamma'})^{-1} \tilde{W}_{\gamma \setminus \gamma'}^T y$
and $\tilde{W}_{\gamma \setminus \gamma'}= W_{\gamma \setminus \gamma'} - W_{\gamma'}(W_{\gamma'}^TW_{\gamma'})^{-1}W_{\gamma'}^TW_{\gamma \setminus \gamma'}$.
\label{lem:waldstat_nested}
\end{lemma}

\begin{defn}
A $d$-dimensional random vector $s=(s_1,\ldots,s_d)$ follows a sub-Gaussian distribution with parameters $\mu \in \mathbb{R}^d$ and $\sigma^2 > 0$, which we denote by $s \sim \mbox{SG}(\mu,\sigma^2)$ if and only if
$$
E\left[ \exp\{ \alpha^T(s - \mu) \} \right] \leq \exp\{ \alpha^T\alpha \sigma^2/2 \}
$$
for all $\alpha \in \mathbb{R}^d$.
\label{def:subgaussian}
\end{defn}

\begin{lemma}
Let $s \sim \mbox{SG}(\mu, \sigma^2)$ be a sub-Gaussian $d$-dimensional random vector, and $A$ be a $q \times d$ matrix.
Then $As \sim \mbox{SG}(A\mu, \lambda \sigma^2)$, where $\lambda$ is the largest eigenvalue of $A^TA$, or equivalently the largest eigenvalue of $A A^T$.
\label{lem:subgaussian_lincomb}
\end{lemma}

\begin{lemma}
Let $y \sim \mbox{SG}(\mu,\sigma^2)$ be an $n$-dimensional sub-Gaussian random vector.
Let $W= (y-a)^T X (X^T X)^{-1} X^T (y-a)$, where $X$ is an $n \times d$ matrix such that $X^TX$ is invertible.
Then $W= s^T s$, where $s \sim \mbox{SG}((X^TX)^{-1/2} X^T (\mu - a), \sigma^2)$ is $d$-dimensional.
\label{lem:subgaussian_quadform}
\end{lemma}

\begin{lemma} {\bf Central sub-Gaussian quadratic forms. Right-tail probabilities}
Let $s=(s_1,\ldots,s_d) \sim \mbox{SG}(0,\sigma^2)$. Then
\begin{enumerate}[leftmargin=*,label=(\roman*)]
\item For any $t>0$,
\begin{align}
 P \left( \frac{s^Ts}{\sigma^2} > d t [1 + (2/t)^{\frac{1}{2}} + 1/t] \right) \leq \exp\left\{- \frac{dt}{2} \right\}.
\nonumber
\end{align}

\item For any $q>0$ and any $k_0$ such that $k_0 \geq \frac{(1+k_0)}{q} + \left[\frac{2(1+k_0)}{q}\right]^{\frac{1}{2}}$,
\begin{align}
P \left( \frac{ s^T s}{\sigma^2} > dq \right) \leq \exp \left\{  - \frac{d q}{2(1+k_0)} \right\}
\nonumber
\end{align}

\item For any $q \geq 2 (1+2^{1/2})^2$,
\begin{align}
P \left( \frac{ s^T s}{\sigma^2} > dq \right) \leq \exp \left\{  - \frac{d q}{2(1+k_0)} \right\}
\nonumber
\end{align}
where $k_0\geq 2^{1/2}(1+2^{1/2})/q^{1/2}$.
\end{enumerate}
\label{lem:tail_quadform_subgaussian}
\end{lemma}

\begin{lemma} {\bf Non-central sub-Gaussian quadratic forms. Left-tail probabilities}
Let $s=(s_1,\ldots,s_d) \sim \mbox{SG}(\mu,\sigma^2)$ and $a \in (0,\mu^T\mu)$. Then
\begin{align}
 P(s^Ts < a) \leq \exp \left\{ - \frac{\mu^T\mu}{8\sigma^2} \left( 1 - \frac{a}{\mu^T\mu} \right)^2 \right\}.
\nonumber
\end{align}

\label{lem:lefttail_quadform_subgaussian}
\end{lemma}

\begin{proof} {\bf of Lemma \ref{lem:lefttail_quadform_subgaussian}.}
For any $t>0$ and $a>0$, it holds that
\begin{align}
 P( s^Ts < a)= P \left( e^{-t s^Ts} > e^{-t a} \right) \leq e^{ta} E(e^{-t s^Ts}),
\nonumber
\end{align}
where the right-hand side follows from Markov's inequality.
Since $s= z + \mu$ where $z \sim SG(0,\sigma^2)$, we obtain
\begin{align}
& P( s^Ts < a) \leq e^{ta} E(e^{-t [(z+\mu)^T (z+\mu)] })=
e^{t(a - \mu^T\mu)} E \left( \frac{e^{-2 t \mu^T z}}{e^{tz^T z}} \right)
\leq e^{t(a - \mu^T\mu)} E \left( e^{-2 t \mu^T z} \right),
\nonumber
\end{align}
where we used that $e^{-t z^Tz} \leq 1$.
Using the definition of sub-Gaussianity to bound the expectation on the right-hand side gives
\begin{align}
 P(s^Ts < a) \leq e^{t(a - \mu^T\mu) + 2 t^2 \mu^T\mu \sigma^2}.
\nonumber
\end{align}

The bound above holds for any $t>0$. The optimal $t$ is found by setting the first derivative to zero, which gives
\begin{align}
 t= \frac{\mu^T\mu - a}{4 \mu^T\mu \sigma^2}.
\nonumber
\end{align}
Note that to have $t>0$, we need that $a < \mu^T\mu$, which holds by assumption. Plugging in the optimal $t$ into the upper-bound gives
\begin{align}
 P(s^T s < a) &\leq 
\exp\left\{-\frac{(a-\mu^T\mu)^2}{4 \mu^T\mu \sigma^2} + \frac{\mu^T\mu \sigma^2 (\mu^T\mu - a)^2}{8 (\mu^T\mu)^2 \sigma^4}   \right\}
\nonumber \\
&=\exp\left\{-\frac{(\mu^T\mu - a)^2}{8 \mu^T\mu \sigma^2} \right\}
= \exp\left\{-\frac{\mu^T \mu}{8 \sigma^2} \left( 1 - \frac{a}{\mu^T\mu} \right)^2 \right\},
\nonumber
\end{align}
as we wished to prove.
\end{proof}

\begin{lemma} {\bf Tail probabilities for differences of sub-Gaussian quadratic forms}
Let $u_1 \sim SG(0,\sigma^2)$ be a $d_1$-dimensional sub-Gaussian vector and $u_2 \sim SG(\mu, \sigma^2)$ a $d_2$-dimensional sub-Gaussian vector, where $\sigma^2 > 0$ is finite and $\mu \in \mathbb{R}^{d_2}$.

\begin{enumerate}[leftmargin=*,label=(\roman*)]
\item  Let $a>0$, then
\begin{align}
 P\left( \frac{a u_1^Tu_1 - u_2^Tu_2}{\sigma^2} > t \right) \leq
 \exp \left\{ - \frac{d_1 q }{2(1+k_0)}\right\} + \exp \left\{ - \frac{\mu^T\mu}{8\sigma^2} \left( 1 - \frac{1}{\log \mu^T\mu} \right) \right\},
\nonumber
\end{align}
for any $t$ such that
\begin{align}
 q:= \frac{t}{a d_1} + \frac{\mu^T\mu}{a d_1 \sigma^2 \log \mu^T\mu} \geq 2(1+2^{1/2})^2
\nonumber
\end{align}
and $k_0= 2^{1/2}(1+2^{1/2})/q^{1/2}$.

In particular, for $t= d \log h$ where $d,h>0$,
\begin{align}
 P\left( \frac{u_1^Tu_1 - u_2^Tu_2}{\sigma^2} > t \right) \leq
h^{-\frac{d}{2a(1+k_0)}} e^{-\frac{\mu^T\mu}{2a(1+k_0) \sigma^2 \log \mu^T\mu}}
+ \exp \left\{ - \frac{\mu^T\mu}{8\sigma^2} \left( 1 - \frac{1}{\log \mu^T\mu} \right) \right\}.
\nonumber
\end{align}

\item Let $t \geq 2(1+2^{1/2})^2 a d_1$. Then
\begin{align}
  P\left( \frac{a u_1^Tu_1 - u_2^Tu_2}{\sigma^2} > t \right) \leq \exp \left\{ -\frac{t}{2 (1+k_0) a} \right\},
\nonumber
\end{align}
\end{enumerate}
for any $k_0 \geq 2^{1/2}(1+2^{1/2}) (a d_1/t)^{1/2}$.

\label{lem:subgaussian_difquad_tail}
\end{lemma}

\begin{proof} {\bf of Lemma \ref{lem:subgaussian_difquad_tail}, Part (i).}
Let $\lambda= \mu^T\mu$.
The union bound gives that, for any $t'>0$,
\begin{align}
  P\left( \frac{a u_1^Tu_1 - u_2^Tu_2}{\sigma^2} > t \right) =
 P\left( \frac{a u_1^Tu_1 - u_2^Tu_2}{\sigma^2} > \frac{t}{2} + t' - \left[t' - \frac{t}{2}\right] \right) 
\nonumber \\
\leq
P \left( \frac{u_1^T u_1}{\sigma^2} > \frac{t}{2} + t' \right) + P \left( \frac{u_2^Tu_2}{\sigma^2} < t' - \frac{t}{2} \right).
\label{seq:bound_sgdif}
\end{align}
In particular take $t'=t/2 + \lambda/[\sigma^2 \log \lambda]$, so that $t'-t/2= \lambda/[\sigma^2 \log \lambda]$ and $t'+t/2= t + \lambda/[\sigma^2 \log \lambda]$.
Then, using Lemma \ref{lem:lefttail_quadform_subgaussian} we have that the second term in \eqref{seq:bound_sgdif} is
\begin{align}
 P \left( \frac{u_2^T u_2}{\sigma^2} < \frac{\lambda}{\sigma^2 \log \lambda}  \right) \leq \exp \left\{ - \frac{\lambda}{8\sigma^2} \left( 1 - \frac{1}{\log \lambda} \right) \right\}.
\label{seq:bound_sgdif2}
\end{align}

The first term in \eqref{seq:bound_sgdif} is
\begin{align}
 P \left( \frac{u_1^Tu_1}{\sigma^2} > \frac{t}{a} + \frac{\lambda}{a \sigma^2 \log \lambda}  \right)=
P \left( \frac{u_1^Tu_1}{\sigma^2} > d_1 \left[ \frac{t}{d_1 a} + \frac{\lambda}{d_1 a \sigma^2 \log \lambda} \right]  \right).
\label{seq:bound_sgdif1}
\end{align}
Since
\begin{align}
 q= \frac{t}{d_1 a} + \frac{\lambda}{d_1 a \sigma^2 \log \lambda} \geq 2(1+2^{1/2})^2,
\label{seq:qcondition_sgdif}
\end{align}
by assumption, Lemma \ref{lem:tail_quadform_subgaussian}(iii) gives that \eqref{seq:bound_sgdif1} is
\begin{align}
 \leq \exp \left\{ - \frac{d_1 q }{2(1+k_0)}\right\}.
\nonumber
\end{align}
Combining \eqref{seq:bound_sgdif1} and \eqref{seq:bound_sgdif2} gives
\begin{align}
  P\left( \frac{a u_1^Tu_1 - u_2^Tu_2}{\sigma^2} > t \right) \leq 
 \exp \left\{ - \frac{d_1 q }{2(1+k_0)}\right\} + \exp \left\{ - \frac{\lambda}{8\sigma^2} \left( 1 - \frac{1}{\log \lambda} \right) \right\},
\nonumber
\end{align}
as we wished to prove.
Finally, for the particular case where one plugs in $t= d \log h$ where $d,h>0$ satisfy the condition \eqref{seq:qcondition_sgdif} above, gives
\begin{align}
 q= \frac{d \log h}{d_1 a} + \frac{\lambda}{d_1 a \sigma^2 \log \lambda}
\nonumber
\end{align}
and hence
\begin{align}
 \exp \left\{ - \frac{d_1 q }{2(1+k_0)}\right\}=
 \exp \left\{ - \frac{1}{2a(1+k_0)} \left[ d \log h + \frac{\lambda}{\sigma^2 \log \lambda} \right]\right\}=
 h^{-\frac{d}{2a(1+k_0)}} e^{-\frac{\lambda}{2a (1+k_0) \sigma^2 \log \lambda}}.
\nonumber
\end{align}

\end{proof}

\begin{proof} {\bf of Lemma \ref{lem:subgaussian_difquad_tail}, Part (ii).}
Since $a u_1^Tu_1 - u_2^Tu_2 \leq a u_1^T u_1$, it follows that
\begin{align}
  P\left( \frac{a u_1^Tu_1 - u_2^Tu_2}{\sigma^2} > t \right) \leq  
P\left( \frac{a u_1^Tu_1}{\sigma^2} > t \right)
=P\left( \frac{u_1^Tu_1}{\sigma^2} > d_1 \frac{t}{a d_1} \right).
\nonumber
\end{align}
Using Lemma \ref{lem:tail_quadform_subgaussian}(iii) gives that the right-hand side is
\begin{align}
\leq \exp \left\{ -\frac{t}{2 (1+k_0) a} \right\} 
\nonumber
\end{align}
for any $t/(ad_1) \geq 2(1+2^{1/2})^2$ and $k_0 \geq 2^{1/2}(1+2^{1/2}) (a d_1/t)^{1/2}$, as we wished to prove.
\end{proof}


\newpage
\section{Proof of Theorem \ref{thm:bf}}
\label{supplsec:proof_bf}

Recall that the Bayes factor is
\begin{align}
 B_{\gamma \gamma^*}&=
\frac{\left|g V_{\gamma^*} W_{\gamma^*}^T \Sigma^{-1} W_{\gamma^*}  + I \right|^{\frac{1}{2}}}{\left| g V_{\gamma} W_{\gamma}^T \Sigma^{-1} W_\gamma  + I \right|^{\frac{1}{2}}}
\nonumber \\
&\times \exp \left\{  \frac{1}{2} \left[ \hat{\eta}_\gamma^T ( W_\gamma^T \Sigma^{-1} W_\gamma + (g V_\gamma)^{-1}) \hat{\eta}_\gamma 
- \hat{\eta}_{\gamma^*}^T ( W_{\gamma^*}^T \Sigma^{-1} W_{\gamma^*} + (g V_{\gamma^*})^{-1}) \hat{\eta}_{\gamma^*}
\right] \right\},
\label{seq:bf_theorem}
\end{align}
where $\hat{\eta}_\gamma= E(\eta_\gamma \mid y,\gamma)= (W_\gamma^T \Sigma^{-1} W_\gamma + (g V_\gamma)^{-1})^{-1} W_\gamma^T \Sigma^{-1} y$.

The proof strategy is as follows. First we show that the first term in \eqref{seq:bf_theorem} is essentially given by $(g n)^{(|\gamma^*|_0 - |\gamma|_0)/2}$ (up to lower-order terms). To characterize the second term in \eqref{seq:bf_theorem} we note that the terms in the exponent are 
a Bayesian version of the sum of explained squares under model $\gamma$ minus that for $\gamma^*$, show that these are essentially equivalent to the least-squares counterparts. The latter sum-of-squares can be re-written as a quadratic form of sub-Gaussian vectors using Lemma \ref{lem:waldstat_nested}, which can be bounded using the inequalities developed in Section \ref{supplsec:auxiliary_results}.

Consider the first term in \eqref{seq:bf_theorem}. Under Assumption (A2) the eigenvalues of $g V_{\gamma} W_{\gamma}^T \Sigma^{-1} W_{\gamma} + I$
lie in $(gn \underline{l}_\gamma + 1, g n \bar{l}_\gamma + 1)$, hence
\begin{align}
\frac{(g n \underline{l}_{\gamma^*})^{\frac{|\gamma^*|_0}{2}}}{(g n \bar{l}_{\gamma} + 1)^{\frac{|\gamma|_0}{2}}}
\leq
\frac{(g n \underline{l}_{\gamma^*} + 1)^{\frac{|\gamma^*|_0}{2}}}{(g n \bar{l}_{\gamma} + 1)^{\frac{|\gamma|_0}{2}}}
\leq
\frac{\left|g V_{\gamma^*} W_{\gamma^*}^T \Sigma^{-1} W_{\gamma^*}  + I \right|^{\frac{1}{2}}}{\left| g V_{\gamma} W_{\gamma}^T \Sigma^{-1} W_\gamma  + I \right|^{\frac{1}{2}}}
\leq
\frac{(g n \bar{l}_{\gamma^*} + 1)^{\frac{|\gamma^*|_0}{2}}}{(g n \underline{l}_{\gamma} + 1)^{\frac{|\gamma|_0}{2}}}
\leq
\frac{(g n \bar{l}_{\gamma^*} + 1)^{\frac{|\gamma^*|_0}{2}}}{(g n \underline{l}_{\gamma})^{\frac{|\gamma|_0}{2}}}.
\nonumber
\end{align}

Using that $\lim_{n \rightarrow \infty} g n= \infty$ by Assumption (A4), and that $(\underline{l}_\gamma,\bar{l}_\gamma)$ are bounded by constants by Assumption (A2),
it is simple to show that 
\begin{align}
(g n k_1)^{\frac{|\gamma^*|_0-|\gamma|_0}{2}}
\leq
\frac{(g n \underline{l}_{\gamma^*})^{\frac{|\gamma^*|_0}{2}}}{(g n \bar{l}_{\gamma} k')^{\frac{|\gamma|_0}{2}}}
\leq
\frac{\left|g V_{\gamma^*} W_{\gamma^*}^T \Sigma^{-1} W_{\gamma^*}  + I \right|^{\frac{1}{2}}}{\left| g V_{\gamma} W_{\gamma}^T \Sigma^{-1} W_\gamma  + I \right|^{\frac{1}{2}}}
\leq
\frac{(g n \bar{l}_{\gamma^*} k')^{\frac{|\gamma^*|_0}{2}}}{(g n \underline{l}_{\gamma})^{\frac{|\gamma|_0}{2}}}
\leq (g n k_2)^{\frac{|\gamma^*|_0-|\gamma|_0}{2}}
\label{seq:bf_term1}
\end{align}
for large enough $n$, a fixed $k'$ that can be taken arbitrarily close to 1, and 
$k_1= \underline{l}_{\gamma^*} k' / \bar{l}_\gamma$ and $k_2= \bar{l}_{\gamma^*} k' / \underline{l}_\gamma$
where $0 < k_1,k_2 < \infty$ are finite non-zero constants by assumption.
This concludes the first part of the proof.

Regarding the second term in (\ref{seq:bf_theorem}), to ease notation let $\tilde{y}= \Sigma^{-1/2} y$, $\widetilde{W}_\gamma= \Sigma^{-1/2} W_\gamma$,
and $\tilde{\eta}_\gamma= (\widetilde{W}_\gamma^T \widetilde{W}_\gamma)^{-1} \widetilde{W}_\gamma^T \tilde{y}$ the least-squares estimate when regressing $\tilde{y}$ on $\widetilde{W}_\gamma$, which is guaranteed to exist by Assumption (A1).
Then the exponent in \eqref{seq:bf_theorem} features $\tilde{s}_\gamma - \tilde{s}_{\gamma^*}$, where
\begin{align}
\tilde{s}_\gamma= \hat{\eta}_\gamma^T (W_\gamma^T \Sigma^{-1} W_\gamma + (g V_\gamma)^{-1}) \hat{\eta}_\gamma=
\tilde{y}^T \widetilde{W}_\gamma [\widetilde{W}_\gamma^T \widetilde{W}_\gamma + (g V_\gamma)^{-1}]^{-1} \widetilde{W}_\gamma^T \tilde{y}
\label{seq:bayesian_sse}
\end{align}
can be interpreted as the Bayesian sum of explained squares by model $\gamma$.
Under Assumptions (A2) and (A4), $\tilde{s}_\gamma - \tilde{s}_{\gamma^*}$ is essentially equivalent to the difference between classical least-squares sum of explained squares $s_\gamma - s_{\gamma^*}$, where
\begin{align}
 s_\gamma= \tilde{y}^T \widetilde{W}_\gamma (\widetilde{W}_\gamma^T \widetilde{W}_\gamma)^{-1} \widetilde{W}_\gamma^T \tilde{y}=
\tilde{\eta}_\gamma^T \widetilde{W}_\gamma^T \widetilde{W}_\gamma \tilde{\eta}_\gamma.
\nonumber
\end{align}
Briefly, let
\begin{align}
\tilde{H}_\gamma&= \widetilde{W}_\gamma [\widetilde{W}_\gamma^T \widetilde{W}_\gamma + (g V_\gamma)^{-1}]^{-1} \widetilde{W}_\gamma^T
\nonumber \\
H_\gamma&= \widetilde{W}_\gamma [\widetilde{W}_\gamma^T \widetilde{W}_\gamma]^{-1} \widetilde{W}_\gamma^T
\nonumber
\end{align}
so that
\begin{align}
 \tilde{s}_\gamma - \tilde{s}_{\gamma^*}=
\tilde{y}^T( \tilde{H}_\gamma - \tilde{H}_{\gamma^*}) (H_\gamma - H_{\gamma^*})^{-1} (H_\gamma - H_{\gamma^*}) \tilde{y},
\nonumber
\end{align}
which lies in the interval $(s_\gamma - s_{\gamma^*}) (l_1, l_2)$,
where $(l_1,l_2)$ are the smallest and largest eigenvalues of $( \tilde{H}_\gamma - \tilde{H}_{\gamma^*}) (H_\gamma - H_{\gamma^*})^{-1}$.
Using Assumption (A2) it is possible to show that both $l_1$ and $l_2$ converge to 1 as $n \rightarrow \infty$, implying that for large enough $n$
\begin{align}
\tilde{s}_\gamma - \tilde{s}_{\gamma^*} \leq
\begin{cases}
 (s_\gamma - s_{\gamma^*}) (1 + \delta) \mbox{, if } s_\gamma - s_{\gamma^*} \geq 0 \\
 (s_\gamma - s_{\gamma^*}) (1 - \delta) \mbox{, if } s_\gamma - s_{\gamma^*} < 0
\end{cases}
\label{eq:ineq_bayesssdif}
\end{align}
where $\delta>0$ is a constant that may be taken arbitrarily close to 0 as $n$ grows.


The remainder of the proof characterizes the behavior of $s_\gamma - s_{\gamma^*}$, separately for the over-fitted case where $\gamma^* \subset \gamma$ and the non over-fitted case where $\gamma^* \not\subset \gamma$.

\subsection{Part (i). Case $\gamma^* \subset \gamma$}

Let $\widetilde{W}_\gamma= (\widetilde{W}_{\gamma^*}, \widetilde{W}_{\gamma \setminus \gamma^*})$ where $\widetilde{W}_{\gamma^*}$ are the columns corresponding to $\gamma^*$ and $\widetilde{W}_{\gamma \setminus \gamma^*}$ the remaining columns.
Lemma \ref{lem:waldstat_nested} gives that
\begin{align}
 s_\gamma - s_{\gamma^*}= \tilde{y}^T Z_{\gamma \setminus \gamma^*} (Z_{\gamma \setminus \gamma^*}^T Z_{\gamma \setminus \gamma^*})^{-1} Z_{\gamma \setminus \gamma^*}^T \tilde{y} \geq 0
\nonumber
\end{align}
where $Z_{\gamma \setminus \gamma^*}= (I - H_{\gamma^*}) \widetilde{W}_{\gamma \setminus \gamma^*}$ are the residuals from regressing $\widetilde{W}_{\gamma \setminus \gamma^*}$ on $\widetilde{W}_{\gamma^*}$, and $H_{\gamma^*}= \widetilde{W}_{\gamma^*} (\widetilde{W}_{\gamma^*}^T \widetilde{W}_{\gamma^*})^{-1} \widetilde{W}_{\gamma^*}^T$ the projection matrix onto the column span of $\widetilde{W}_{\gamma^*}$.

Recall that $y - W_{\gamma^*} \eta_{\gamma^*}^* \sim SG(0, \omega)$ by assumption, which by Lemma \ref{lem:subgaussian_lincomb} implies that
$$
\tilde{y} - \widetilde{W}_{\gamma^*} \eta_{\gamma^*}= \Sigma^{-1/2} (y - W_{\gamma^*} \eta_{\gamma^*}) \sim SG(0, \tilde{\omega}),
$$
where $\tilde{\omega}= \omega \tau$ and $\tau$ is the largest eigenvalue of $\Sigma^{-1}$.
Now, using that $(\widetilde{W}_{\gamma^*} \eta_{\gamma^*}^*)^T Z_{\gamma \setminus \gamma^*}=0$ and Lemma \ref{lem:subgaussian_quadform} gives that
\begin{align}
 s_\gamma - s_{\gamma^*}= 
(\tilde{y} - \widetilde{W}_{\gamma^*} \eta_{\gamma^*}^* + \widetilde{W}_{\gamma^*} \eta_{\gamma^*}^*)^T Z_{\gamma \setminus \gamma^*} (Z_{\gamma \setminus \gamma^*}^T Z_{\gamma \setminus \gamma^*})^{-1} Z_{\gamma \setminus \gamma^*}^T (\tilde{y} - \widetilde{W}_{\gamma^*} \eta_{\gamma^*}^* + \widetilde{W}_{\gamma^*} \eta_{\gamma^*}^*)=
\nonumber \\
(\tilde{y} - \widetilde{W}_{\gamma^*} \eta_{\gamma^*}^*)^T Z_{\gamma \setminus \gamma^*} (Z_{\gamma \setminus \gamma^*}^T Z_{\gamma \setminus \gamma^*})^{-1} Z_{\gamma \setminus \gamma^*}^T (\tilde{y} - \widetilde{W}_{\gamma^*} \eta_{\gamma^*}^*)
=u^T u,
\nonumber
\end{align}
where $u \sim SG(0, \tilde{\omega})$ is a sub-Gaussian vector of dimension $|\gamma|_0 - |\gamma^*|_0$.
Combining this result with \eqref{seq:bf_term1} and \eqref{eq:ineq_bayesssdif} gives
\begin{align}
B_{\gamma \gamma^*} \leq
\exp \left\{  \frac{1}{2} \left[ (1 + \delta) u^Tu + (|\gamma^*|_0-|\gamma|_0) \log(g n k_2) \right] \right\}.
\nonumber
\end{align}
Hence, consider any sequence $a_n \geq 0$, it follows that
\begin{align}
 P_F \left( B_{\gamma \gamma^*} \geq a_n \right) \leq 
P_F \left( \frac{1}{2} \left[ (1 + \delta) u^Tu + (|\gamma^*|_0-|\gamma|_0) \log(g n k_2) \right] \geq \log a_n \right)=
\nonumber \\
=P_F \left(  \frac{u^Tu}{\tilde{\omega}}  \geq \frac{|\gamma|_0 - |\gamma^*|_0}{\tilde{\omega} (1+\delta)}  \log\left(g n k_2a_n^{\frac{2}{|\gamma|_0 - |\gamma^*|_0}} \right) \right).
\label{seq:bfrighttail_overfitted}
\end{align}
Since $u \sim SG(0,\tilde{\omega})$, this is a right-tail probability of a sub-Gaussian quadratic form. Using Lemma \ref{lem:tail_quadform_subgaussian} with $d= |\gamma|_0 - |\gamma^*|_0$ and that $(\delta,k_2)$ are constants by assumption and $\tilde{\omega}$ is bounded by constants by Assumption (A3), said tail probability vanishes for any $a_n$ such that
$\lim_{n \rightarrow \infty} g n a_n^{2/(|\gamma|_0 - |\gamma^*|_0)}$.
In particular, take $a_n= b_n/(gn)^{(|\gamma|_0 - |\gamma^*|_0)/2}$, then the condition is that
\begin{align}
 \lim_{n \rightarrow \infty} g n a_n^{\frac{2}{|\gamma|_0 - |\gamma^*|_0}}=
b_n^{\frac{2}{|\gamma|_0 - |\gamma^*|_0}}= \infty
\Longleftrightarrow 
\lim_{n \rightarrow \infty} \frac{\log b_n}{(|\gamma|_0 - |\gamma^*|_0)/2}= \infty.
\nonumber
\end{align}
The latter condition holds for any $b_n$ such that $\log b_n= c_n (|\gamma|_0 - |\gamma^*|_0)/2$, for any $c_n$ such that $\lim_{n \rightarrow \infty} c_n= \infty$.
For these choices of $a_n$ and $b_n$ we obtain
\begin{align}
 a_n= \frac{b_n}{(gn)^{\frac{|\gamma|_0 - |\gamma^*|_0}{2}}}=
 \left( \frac{e^{c_n}}{gn} \right)^{\frac{|\gamma|_0 - |\gamma^*|_0}{2}}
\nonumber
\end{align}
In conclusion,
\begin{align}
\lim_{n \rightarrow \infty} P_F \left( B_{\gamma \gamma^*} \geq  \left( \frac{d_n}{gn} \right)^{\frac{|\gamma|_0 - |\gamma^*|_0}{2}} \right)= 0
\nonumber
\end{align}
for any $d_n= e^{c_n}$ such that $\lim_{n \rightarrow} d_n= \infty$, as we wished to prove.

\subsection{Part (ii). Case $\gamma^* \not\subset \gamma$}

Let $\gamma'$ be the union of models $\gamma^*$ and $\gamma$, i.e with design matrix $W_{\gamma'}$ containing all columns in $W_{\gamma^*}$ and $W_{\gamma}$, so that both $\gamma^*$ and $\gamma$ are contained in $\gamma'$.
Proceeding similarly to Part (i), Lemmas \ref{lem:waldstat_nested} and \ref{lem:subgaussian_quadform} give that
\begin{align}
 s_\gamma - s_{\gamma^*}=
s_\gamma - s_{\gamma'} + s_{\gamma'} - s_{\gamma^*}=
u_1^T u_1 - u_2^T u_2
\nonumber
\end{align}
where
\begin{align}
u_1^T u_1&=  s_{\gamma'} - s_{\gamma^*}= \tilde{y}^T Z_{\gamma' \setminus \gamma^*} (Z_{\gamma' \setminus \gamma^*}^T Z_{\gamma' \setminus \gamma^*})^{-1} Z_{\gamma' \setminus \gamma^*}^T \tilde{y}
\nonumber \\
u_2^T u_2&= s_{\gamma'} - s_{\gamma}= \tilde{y}^T Z_{\gamma' \setminus \gamma} (Z_{\gamma' \setminus \gamma}^T Z_{\gamma' \setminus \gamma})^{-1} Z_{\gamma' \setminus \gamma}^T \tilde{y}
\nonumber
\end{align}
and $Z_{\gamma' \setminus \gamma^*}= (I - H_{\gamma^*}) \widetilde{W}_{\gamma' \setminus \gamma^*}$, analogously for $Z_{\gamma' \setminus \gamma}$, and $H_{\gamma}$ and $H_{\gamma^*}$ are the projection matrix defined above.

The key is to bound the two terms $u_1^Tu_1$ and $u_2^Tu_2$, by noting that they are quadratic forms of sub-Gaussian vectors, which will allow us to use Lemma \ref{lem:subgaussian_difquad_tail}.
Specifically, recall from Part(i) that $\tilde{y} - \widetilde{W}_{\gamma^*} \eta_{\gamma^*}^* \sim SG(0, \tilde{\omega})$ where $\tilde{\omega}= \omega \tau$ and $\tau$ is the largest eigenvalue of $\Sigma^{-1}$.
Using that $(\widetilde{W}_{\gamma^*} \eta_{\gamma^*}^*)^T Z_{\gamma' \setminus \gamma^*}=0$ and Lemma \ref{lem:subgaussian_quadform} gives that
$u_1 \sim SG(0, \tilde{\omega})$ with dimension $p_{\gamma'}-|\gamma^*|_0$.
Regarding the term $u_2^T u_2$, since $\tilde{y} - \widetilde{W}_\gamma \eta_\gamma^* \sim SG(\widetilde{W}_{\gamma^*}^* \eta_{\gamma^*}^* - \widetilde{W}_\gamma \eta_\gamma^*, \tilde{\omega})$ by assumption, by Lemma \ref{lem:subgaussian_quadform} we have that $u_2 \sim SG(\mu,\tilde{\omega})$ is a $p_{\gamma'}-|\gamma|_0$ dimensional sub-Gaussian vector with 
$\mu= (Z_{\gamma' \setminus \gamma}^T Z_{\gamma' \setminus \gamma})^{-1/2} Z_{\gamma' \setminus \gamma}^T (\widetilde{W}_{\gamma^*} \eta_{\gamma^*}^* - \widetilde{W}_\gamma \eta_\gamma^*) \neq 0$.

Combining these results with \eqref{seq:bf_term1} and \eqref{eq:ineq_bayesssdif} gives
\begin{align}
B_{\gamma \gamma^*} \leq
\exp \left\{  \frac{1}{2} \left[  ((1 + \delta) u_1^Tu_1 - (1 - \delta) u_2^Tu_2) + (|\gamma^*|_0-|\gamma|_0) \log(g n k_2) \right] \right\}.
\nonumber
\end{align}
and hence, for any sequence $a_n \geq 0$,
\begin{align}
 P_F \left( B_{\gamma \gamma^*} \geq a_n  \right) \leq
P_F \left( \frac{1}{2} \left[ (1 - \delta) \left( \frac{1+\delta}{1-\delta} u_1^Tu_1 - u_2^Tu_2 \right) + (|\gamma^*|_0-|\gamma|_0) \log(g n k_2) \right] \geq \log a_n \right)
\nonumber \\
= 
P_F \left( \frac{\frac{1+\delta}{1-\delta} u_1^Tu_1 - u_2^Tu_2}{\tilde{\omega}}  \geq \frac{|\gamma|_0 - |\gamma^*|_0}{\tilde{\omega} (1 - \delta)} \left[\log(g n k_2 a_n^{\frac{2}{|\gamma|_0 - |\gamma^*|_0}}) \right] \right).
\label{seq:bftailprob_notsubsetcase}
\end{align}

\eqref{seq:bftailprob_notsubsetcase} is of the form given in Lemma \ref{lem:subgaussian_difquad_tail}, taking 
$t= d \log h$ with $d=(|\gamma|_0-|\gamma^*|_0)/[\tilde{\omega}(1-\delta)]$ and $h=gnk_2 a_n^{2/(|\gamma|_0-|\gamma^*|_0}$,
$a=(1+\delta)/(1-\delta)$, $d_1=p_{\gamma'}-|\gamma^*|_0$, $d_2=p_{\gamma'} - |\gamma|_0$ and
\begin{align}
 \mu^T\mu&= 
(\widetilde{W}_{\gamma^*} \eta_{\gamma^*}^* - \widetilde{W}_\gamma \eta_\gamma^*)^T Z_{\gamma' \setminus \gamma}
(Z_{\gamma' \setminus \gamma}^T Z_{\gamma' \setminus \gamma})^{-1} 
Z_{\gamma' \setminus \gamma}^T (\widetilde{W}_{\gamma^*}^* \eta_{\gamma^*}^* - \widetilde{W}_\gamma \eta_\gamma^*)
\nonumber \\
&=(\widetilde{W}_{\gamma^*} \eta_{\gamma^*}^*)^T (I-H_\gamma) \widetilde{W}_{\gamma^* \setminus \gamma}
(\widetilde{W}_{\gamma^* \setminus \gamma}^T (I- H_\gamma) \widetilde{W}_{\gamma^* \setminus \gamma})^{-1} 
\widetilde{W}_{\gamma^* \setminus \gamma}^T (I - H_\gamma) \widetilde{W}_{\gamma^*} \eta_{\gamma^*}^*
\nonumber
\end{align}
where to obtain the right-hand side we used that $I-H_\gamma$ is idempotent, $\widetilde{W}_{\gamma' \setminus \gamma}= \widetilde{W}_{\gamma^* \setminus \gamma}$ and that $\widetilde{W}_\gamma^T (I - H_\gamma)= \widetilde{W}_\gamma^T - \widetilde{W}_\gamma^T= 0$.

The expression for $\mu^T\mu$ can be simplified. First, note that the linear predictor $\widetilde{W}_{\gamma^*} \eta_{\gamma^*}^*= \widetilde{W}_{\gamma^* \setminus \gamma} \eta_{\gamma^* \setminus \gamma}^* + \widetilde{W}_{\gamma^* \cap \gamma} b_{\gamma^* \cap \gamma}$ can be decomposed into the contribution from $\gamma^*$ and $\gamma^* \cap \gamma$, where $b_{\gamma^* \cap \gamma}$ is the subset of $\eta_{\gamma^*}$ corresponding to elements in $\gamma^* \cap \gamma$.
Second, $(I - H_\gamma) \widetilde{W}_{\gamma^* \cap \gamma}=0$, since $\widetilde{W}_{\gamma^* \cap \gamma}$ is in the column span of $\widetilde{W}_\gamma$ and $H_\gamma$ is the corresponding linear projection operator. Hence,
\begin{align}
 \mu^T\mu&= (\eta_{\gamma^* \setminus \gamma}^*)^T (\widetilde{W}_{\gamma^*\setminus \gamma}^T (I-H_\gamma) \widetilde{W}_{\gamma^* \setminus \gamma}
(\widetilde{W}_{\gamma^* \setminus \gamma}^T (I- H_\gamma) \widetilde{W}_{\gamma^* \setminus \gamma})^{-1} 
\widetilde{W}_{\gamma^* \setminus \gamma}^T (I - H_\gamma) \widetilde{W}_{\gamma^*} \eta_{\gamma^*}^*
\nonumber \\
&= (\eta_{\gamma^* \setminus \gamma}^*)^T \widetilde{W}_{\gamma^* \setminus \gamma}^T (I - H_\gamma) \widetilde{W}_{\gamma^*} \eta_{\gamma^*}^*
= (\widetilde{W}_{\gamma^*} \eta_{\gamma^*}^*)^T (I - H_\gamma) \widetilde{W}_{\gamma^*} \eta_{\gamma^*}^*.
\label{seq:noncentrality_param}
\end{align}

To conclude the proof we apply Lemma \ref{lem:subgaussian_difquad_tail} and deduce the smallest $a_n$ one can set so that the tail probability \eqref{seq:bftailprob_notsubsetcase} vanishes as $n \rightarrow \infty$.
From Lemma \ref{lem:subgaussian_difquad_tail}, it suffices that $\mu^T \mu$ diverges to $\infty$, which holds by Assumption (A5), and that
\begin{align}
q:= \frac{d \log h}{d_1 a} + \frac{\mu^T\mu}{d_1 a \tilde{\omega} \log \mu^T\mu}
= \frac{\frac{1-\delta}{1+\delta}(|\gamma|_0-|\gamma^*|_0) \log \left( gnk_2 a_n^{2/(|\gamma|_0-|\gamma^*|_0)} \right)}{(|\gamma'|_0-|\gamma^*|_0) \tilde{\omega} (1-\delta)} + \frac{\mu^T\mu (1-\delta)/(1+\delta)}{(|\gamma'|_0-|\gamma^*|_0) \tilde{\omega} \log \mu^T\mu}
\nonumber
\end{align}
diverges to infinity as $n \rightarrow \infty$, which happens if and only if its exponential
\begin{align}
\left[ \left( gnk_2 \right)^{\frac{|\gamma|_0 - |\gamma^*|_0}{1+\delta}} a_n^{\frac{2}{1+\delta}} 
 e^{\frac{\mu^T\mu}{\log \mu^T\mu} \frac{1-\delta}{1+\delta}} \right]^{\frac{1}{\tilde{\omega}(|\gamma'|_0-|\gamma^*|_0)}}.
\nonumber
\end{align}
diverges to infinity. This can be achieved by taking
\begin{align}
 a_n= \left[ \left( gnk_2 \right)^{\frac{-(|\gamma|_0 - |\gamma^*|_0)}{1+\delta}} e^{-\frac{\mu^T\mu}{\log \mu^T\mu} \frac{1-\delta}{1+\delta}} b_n  \right]^{\frac{1+\delta}{2}}
\nonumber
\end{align}
where $b_n$ satisfies $\lim_{n \rightarrow \infty} b_n^{1/[\tilde{\omega}(|\gamma'|_0-|\gamma^*|_0)]}= \infty$.
The latter condition is equivalent to $\lim_{n \rightarrow \infty} [\log b_n]/[\tilde{\omega}(|\gamma'|_0-|\gamma^*|_0)] = \infty$, and is satisfied by taking
\begin{align}
 \log b_n= \tilde{\omega} (|\gamma'|_0 - |\gamma^*|_0) c_n \Longrightarrow b_n= e^{\tilde{\omega} (|\gamma'|_0 - |\gamma^*|_0) c_n}
\nonumber
\end{align}
for any $c_n$ that diverges to infinity. Then the expression for $a_n$ becomes
\begin{align}
 a_n= \left[ \left( gnk_2 \right)^{\frac{-(|\gamma|_0 - |\gamma^*|_0)}{1+\delta}} e^{-\frac{\mu^T\mu}{\log \mu^T\mu} \frac{1-\delta}{1+\delta} + (|\gamma'|_0 - |\gamma^*|_0) \tilde{\omega} c_n}  \right]^{\frac{1+\delta}{2}} 
= \left( gnk_2 \right)^{-\frac{(|\gamma|_0-|\gamma^*|_0)}{2}}
e^{-\frac{\mu^T\mu (1-\delta)}{2 \log \mu^T\mu}} d_n^{\tilde{\omega} (|\gamma'|_0-|\gamma^*|_0)}
\nonumber
\end{align}
where $d_n= e^{c_n (1+\delta)/2}$ is any sequence diverging to infinity.

Finally, we argue that $|\gamma'|_0 - |\gamma^*|_0$ may be replaced by $|\gamma|_0$ in the expression of $a_n$.
Since $|\gamma'|_0 - |\gamma^*|_0 \leq |\gamma|_0$ and for any $\tilde{a}_n \geq a_n$, 
\begin{align}
 P_F \left( B_{\gamma \gamma^*} \geq \tilde{a}_n \right) \leq P_F \left( B_{\gamma \gamma^*} \geq a_n \right) 
\nonumber
\end{align}
and the right-hand side vanishes as $n \rightarrow \infty$, we may replace $a_n$ by
\begin{align}
\tilde{a}_n= 
\left( gnk_2 \right)^{-\frac{(|\gamma|_0-|\gamma^*|_0)}{2}} e^{-\frac{\mu^T\mu (1-\delta)}{2 \log \mu^T\mu}} d_n^{\tilde{\omega} |\gamma|_0}
\nonumber
\end{align}
and still obtain that
$\lim_{n \rightarrow \infty} P_F \left( B_{\gamma \gamma^*} \geq \tilde{a}_n \right) \leq \lim_{n \rightarrow \infty} P_F \left( B_{\gamma \gamma^*} \geq a_n \right) = 0$.

\newpage
\section{Auxiliary results for Theorem \ref{thm:pp}}
\label{supplsec:aux_pp}

This section derives several auxiliary results used in the proof of Theorem \ref{thm:pp}.
Lemma \ref{slem:binomial_ogf} is an elementary statement that is useful in carrying out sums of posterior model probabilities over certain model sets.
The remaining results in this section derive bounds on integrals that involve certain tail probabilities for central and non-central sub-Gaussian quadratic forms, which are useful to bound the expected posterior probability assigned to an arbitrary model $\gamma$.
Proposition \ref{prop:intbound_sg_central} considers central sub-Gaussians, and is used in the proof of Theorem \ref{thm:pp} when considering overfitted models, i.e. models $\gamma \supset \gamma^*$ that contain all parameters in the optimal $\gamma^*$ plus some extra spurious parameters.
Theorem \ref{thm:pp} uses only Part (ii) of Proposition \ref{prop:intbound_sg_central}, but for completeness we also provide Part (i), which considers a situation where one obtains faster rates.

Proposition \ref{prop:intbound_sg_noncentral} considers the difference between a central and non-central sub-Gaussian quadratic forms. It is used in Theorem \ref{thm:pp} to bound the posterior probability assigned to non-overfitted models $\gamma \not\supset \gamma^*$.
Parts (i)-(ii) are analogous results, the latter corresponding to a more conservative setting where $a/b \leq 1$, for $(a,b)$ defined below.
The proof of Theorem \ref{thm:pp} uses Part (ii) to bound the posterior probability for non-overfitted models of size less than $\gamma^*$.
For these models, the argument $h$ as defined below is typically decreasing in $n$: roughly speaking, $h$ is a complexity penalty for larger models, which favors $\gamma$ over $\gamma^*$ if the latter has larger size. Hence, for posterior probability of $\gamma$ to vanish one must rely on the non-centrality parameter $\mu^T\mu$ to be large enough.
In contrast, when one considers non-overfitted models of size larger than $\gamma^*$ then $h$ is typically increasing. For those models one may then either use Proposition \ref{prop:intbound_sg_noncentral}(ii), or alternatively use \ref{prop:intbound_sg_noncentral}(iii) which relies solely on $h$ being large enough. The latter option leads to slightly simpler arguments when proving Theorem \ref{thm:pp}.

\begin{lemma}
For any two natural numbers $(k,\bar{l})$ it holds that
\begin{align}
 \sum_{l=k+1}^{\bar{l}} a^{l-k} {l \choose k} \leq \frac{1}{(1-a)^{k+1}} - 1.
\nonumber
\end{align}

Further, suppose that $k$ is either fixed or a non-decreasing function of $a$. Then
\begin{align}
 \lim_{a \rightarrow 0  } \frac{\frac{1}{(1 - a)^{k+1}} -1 }{e^{(k+1)a} -1}= 1.
\nonumber
\end{align}
\label{slem:binomial_ogf}
\end{lemma}

\begin{prop} {\bf Bound on integrated central sub-Gaussian tails.}
Let $u \sim \mbox{SG}(0,\sigma^2)$ be a $d$-dimensional sub-Gaussian vector, and
\begin{align}
 U(a,h)= \int_0^1 P \left( \frac{u^Tu}{\sigma^2} > d a \log \left( \frac{h}{(1/v-1)^{2/d}} \right) \right) dv
\nonumber
\end{align}
where $a>0$ is a constant and $h > e$.
\begin{enumerate}[leftmargin=*,label=(\roman*)]
\item Suppose that $a > 1 + \frac{q_0^{1/2}}{(q_0 + a \log\log h)^{1/2}}$, where $q_0=2(1+2^{1/2})^2$. Then 
\begin{align}
 U(a,h) \leq \frac{2 \max \left\{ (e^{2(1+2^{1/2})^2/a} \log h)^{d/2}, \log(h^{d/2})\right\} }{h^{d/2}}.
\nonumber
\end{align}

\item Suppose that $a \leq 1$ and that $\log h > 2 (1+2^{1/2})^2/a$. Then
\begin{align}
 U(a,h) \leq \frac{2\max \left\{ [e^{\frac{q_0}{a'}} \log (h^{\frac{a}{a'}})]^{\frac{d}{2}}, \log (h^{\frac{ad}{2a'}}) \right\} }{h^{\frac{ad}{2a'}}}
+ \frac{1}{2} \left( \frac{1}{h} \right)^{\frac{ad}{2(1+k_0)}},
\nonumber
\end{align}
for any $a'>1 + \frac{q_0^{1/2}}{(q_0 + a \log\log h)^{1/2}}$, where $k_0= q_0^{1/2}/(a \log h)^{1/2}$ and $q_0=2(1+2^{1/2})^2$.


Further, if $\log \left( h/\log h \right) > q_0/a$, then
\begin{align}
 U(a,h) \leq \frac{3\max \left\{ [e^{\frac{q_0}{a'}} \log (h^{\frac{a}{a'}})]^{\frac{d}{2}}, \log (h^{\frac{ad}{2a'}}) \right\} }{h^{\frac{ad}{2a'}}}.
\nonumber
\end{align}

\end{enumerate}
\label{prop:intbound_sg_central}
\end{prop}

\begin{prop} {\bf Bound on integrated non-central sub-Gaussian tails.}
Let $u_1 \sim SG(0,\sigma^2)$ be a sub-Gaussian vector of dimension $d_1$ and $u_2 \sim SG(\mu,\sigma^2)$ be of dimension $d_2$, where $\mu \in \mathbb{R}^{d_2}$ and $\sigma^2 \in \mathbb{R}^+$. Define
\begin{align}
 U(a,b,h)= \int_0^1 P \left( \frac{bu_1^Tu_1 - u_2^Tu_2}{\sigma^2} > a \log \left( \frac{h}{(1/v-1)^2} \right) \right) dv,
\nonumber
\end{align}
where $a>0$, $b>0$ and $h>0$.

\begin{enumerate}[leftmargin=*,label=(\roman*)]
\item Let $r=h^{\frac{1}{2}} e^{ \frac{\mu^T\mu}{2 a \sigma^2 \log \mu^T\mu}}$ and suppose that
\begin{align}
 \frac{a}{b} > 1 + \frac{q_0^{1/2}}{\left( q_0 + \frac{a}{b d_1} \log\log (h e^{ \frac{\mu^T\mu}{a \sigma^2 \log \mu^T\mu}}) \right)^{\frac{1}{2}}},
\nonumber
\end{align}
where $q_0= 2(1+2^{1/2})^2$ and that
$
 \log(h) + \mu^T\mu/[a \sigma^2 \log \mu^T\mu]= 2 \log(r) > d_1.
$
Then
\begin{align}
 U(a,b,h) \leq \exp \left\{ -\frac{\mu^T\mu}{8\sigma^2} \left( 1 - \frac{1}{\log \mu^T\mu} \right)\right\} +
\frac{1}{r} +
\frac{2 \max \left\{ \frac{e^{\frac{q_0b}{a}}}{d_1/2} \log \left(r \right) , \log \left( r \right)  \right\} }{r},
\nonumber
\end{align}

\item Let $r= h^{1/2} e^{\mu^T\mu/[2 a \sigma^2 \log \mu^T\mu]}$ and suppose that $a/b \leq 1$ and that
\begin{align}
  \log( h ) + \frac{\mu^T\mu}{a \sigma^2 \log \mu^T\mu} > \frac{q_0 b d_1}{a},
\nonumber
\end{align}
where $q_0=2 (1+2^{1/2})^2$. Then
\begin{align}
 \exp \left\{ -\frac{\mu^T\mu}{8\sigma^2} \left( 1 - \frac{1}{\log \mu^T\mu} \right)\right\} +
\frac{2 \max \left\{ \left[ \frac{a e^{\frac{q_0}{a'}}}{a'b d_1/2} \log r \right]^{\frac{d_1}{2}} , \log \left( r^{\frac{a}{b a'}} \right)  \right\}}{r^{\frac{a}{b a'}}}
+ \frac{3}{2 r^{\frac{a}{b (1+k_0)}}},
\nonumber
\end{align}
where $a'= 1 + q_0^{1/2}/[q_0 + a \log \log(r / [b d_1/2])]^{1/2}$ and
$k_0= q_0^{1/2}/[(2 a/b) \log r]^{1/2}$.

\item Suppose that $a \leq b$ and $\log(h) > q_0 b d_1/a$, where $q_0= 2(1+2^{1/2})^2$. Then
\begin{align}
 U(a,b,h) \leq \frac{2 \max \left\{([e^{q_0}/d_1] \log(h^{\frac{a}{b a'}}))^{d_1/2}  , \log(h^{\frac{a}{2b a'}}) \right\}}{h^{\frac{a}{2b a'}}}
+ \frac{1}{2 h^{\frac{a}{2b(1+k_0)}}}
\nonumber
\end{align}
for any 
$$
a'> 1 + \frac{q_0^{1/2}}{(q_0 + (a/b) \log \log (h^{1/d_1}))^{1/2}},
$$
where $k_0= [q_0b d_1/(a \log h)]^{1/2}$.

\end{enumerate}

\label{prop:intbound_sg_noncentral}
\end{prop}

\subsection{Proof of Lemma \ref{slem:binomial_ogf}.} The Binomial coefficient's ordinary generating function states that
\begin{align}
 \sum_{l=0}^\infty {l \choose k} a^{l-k} = \frac{1}{(1 - a)^{k+1}},
\nonumber
\end{align}
which implies that 
\begin{align}
  \sum_{l=k+1}^{\bar{l}} a^{l-k} {l \choose k}
\leq  \sum_{l=0}^{\infty} a^{l-k} {l \choose k}   -  \sum_{l=k} a^{l-k} {l \choose k}
= \frac{1}{(1 - a)^{k+1}} - 1,
\nonumber
\end{align}
proving the first part of the result.

Now, to derive the limit as $a \rightarrow 0$, note that 
\begin{align}
 \frac{1}{(1 - a)^{k+1}}=
\left([1 - a]^{\frac{1}{a}} \right)^{-(k+1)a}
= \left(\frac{e^{-1}}{[1 - a]^{\frac{1}{a}}} \right)^{(k+1)a} e^{(k+1)a}
\nonumber
\end{align}
where $\lim_{a \rightarrow 0} [1-a]^{1/a}= e^{-1}$ from the definition of the exponential function.
Hence, since $k$ is either fixed or bounded above a constant for all $a$, we have $\lim_{a \rightarrow 0} (k+1)a=0$ and therefore that
\begin{align}
\lim_{a \rightarrow 0  } \frac{\frac{1}{(1 - a)^{k+1}}}{e^{(k+1)a}}
=\lim_{a \rightarrow 0} \frac{\left(\frac{e^{-1}}{[1 - a]^{\frac{1}{a}}} \right)^{(k+1)a} e^{(k+1)a}}{e^{(k+1)a}}= 1^0 = 1.
\nonumber
\end{align}

This implies that
\begin{align}
\lim_{a \rightarrow 0  } \frac{\frac{1}{(1 - a)^{k+1}} -1 }{e^{(k+1)a} -1}= \frac{0}{0}.
\nonumber
\end{align}
To solve the limit we apply l'Hopital's rule and take the derivative of both numerator and denominator,
\begin{align}
 \lim_{a \rightarrow 0} \frac{\frac{(k+1)}{(1-a)^{k+2}}}{(k+1)e^{(k+1)a}}=
\lim_{a \rightarrow 0} \frac{1}{1-a} \frac{\frac{1}{(1-a)^{k+1}}}{e^{(k+1)a}}=1,
\nonumber
\end{align}
and therefore
\begin{align}
 \lim_{a \rightarrow 0  } \frac{\frac{1}{(1 - a)^{k+1}} -1 }{e^{(k+1)a} -1}= 1,
\nonumber
\end{align}
as we wished to prove.

\subsection{Proof of Proposition \ref{prop:intbound_sg_central}, Part (i).}

The proof strategy is to apply Lemma \ref{lem:tail_quadform_subgaussian}(iii) to bound the probability in the term inside the integral defining $U(a,h)$, and then carrying out the integration.

Let $q= a \log(h/(1/v-1)^{2/d})$ in Lemma \ref{lem:tail_quadform_subgaussian}(iii), and note that $q$ is an increasing function in $v$.
To apply Lemma \ref{lem:tail_quadform_subgaussian}(iii) we need that $q \geq 2(1+2^{1/2})^2$, that is
\begin{align}
 \frac{h^{d/2}}{e^{\frac{q_0d}{2a}}} \geq 1/v -1 \Longleftrightarrow
v \geq \left( 1 + \frac{h^{d/2}}{e^{\frac{q_0d}{2a}}} \right)^{-1},
\label{seq:integral_lowlimit}
\end{align}
where to ease the upcoming expressions we defined $q_0=2(1+2^{1/2})^2$.
Note that in particular \eqref{seq:integral_lowlimit} holds for $v=v_0$, where
\begin{align}
 v_0= \left( 1 + \frac{h^{d/2}}{[\log h]^{\frac{d}{2}} e^{\frac{q_0d}{2a}}} \right)^{-1}.
\nonumber
\end{align}
Hence, by Lemma \ref{lem:tail_quadform_subgaussian}(iii),
\begin{align}
& U(a,h) \leq v_0 + \int_{v_0}^1 P \left( \frac{u^Tu}{\sigma^2} > d a \log \left( \frac{h}{(1/v-1)^{2/d}} \right) \right) dv
\nonumber \\
&\leq v_0 + \int_{v_0}^1 \exp \left\{  -\frac{da}{2(1+k(v))} \log \left( \frac{h}{[1/v-1]^{\frac{2}{d}}} \right)  \right\} dv
= v_0 + \int_{v_0}^1 \left[ \frac{(1/v-1)^{\frac{2}{d}}}{h} \right]^{\frac{da}{2(1+k(v))}} dv
\nonumber
\end{align}
for any $k(v) \geq 2^{1/2}(1+2^{1/2})/q^{1/2}$. 
It is easy to show the term $(1/v-1)^{2/d}/h \leq 1$ for all $v \in (v_0,1)$, hence the integral can be upper-bounded by replacing the power $da/[2(1+k(v))]$ by a smaller quantity.
Now, recall that $q$ is an increasing function of $v$ and note that $k(v)$ is largest when $q$ is smallest, i.e. when $v$ is smallest, so that $k(v) \leq k(v_0)$ 
and hence
\begin{align}
 \frac{da}{2(1+k(v))} \geq \frac{da}{2(1+k_0)}
\nonumber
\end{align}
where to ease notation we defined $k_0=k(v_0)$.
Therefore,
\begin{align}
 U(a,h) \leq v_0 + \int_{v_0}^1 \left[ \frac{(1/v-1)^{\frac{2}{d}}}{h} \right]^{\frac{da}{2(1+k_0)}} dv
\label{seq:int_bound1}
\end{align}
Note also that 
\begin{align}
k_0= k(v_0)= \frac{2^{1/2}(1+2^{1/2})}{\left[ a \log \left( \frac{h}{(1/v_0-1)^{\frac{2}{d}}} \right) \right]^{\frac{1}{2}}}
= \frac{q_0^{1/2}}{(q_0 + a \log\log h)^{1/2}}
\nonumber
\end{align}
where the right-hand side follows from simple algebra. Hence as $h$ grows $k_0$ may be taken arbitrarily close to 0.

To bound \eqref{seq:int_bound1}, note that $(1/v-1)^{2/d} \leq (1/v_0-1)^{2/d}= h/\log h < h$, since $\log h>1$ by assumption. Hence the integrand in \eqref{seq:int_bound1} is $<1$ and the integral is upper-bounded by taking a smaller power. 
Recall that by assumption
\begin{align}
a > 1 + \frac{q_0^{1/2}}{(q_0 + a \log\log h)^{1/2}} \Longrightarrow \frac{a}{1+k_0}>1
\nonumber
\end{align}
hence \eqref{seq:int_bound1} is upper-bounded by replacing the power $da/[2(1+k_0)]$ by $d/2$, giving
\begin{align}
U(a,b) &\leq v_0 + \frac{1}{h^{\frac{d}{2}}} \int_{v_0}^1 \frac{1}{v} -1 dv
=v_0 + \frac{1}{h^{\frac{d}{2}}} [\log(\frac{1}{v_0}) - (1-v_0)]
< v_0 + \frac{1}{h^{\frac{d}{2}}} \log(\frac{1}{v_0})
\nonumber \\
&= \left( 1 + \frac{h^{d/2}}{[\log h]^{\frac{d}{2}} e^{\frac{q_0d}{2a}}} \right)^{-1} + \frac{1}{h^{\frac{d}{2}}} \log(\frac{1}{v_0})
<  \frac{[\log h]^{\frac{d}{2}} e^{\frac{q_0d}{2a}}}{h^{d/2}} +
\frac{\log (h^{d/2})}{h^{\frac{d}{2}}}
\nonumber \\
& \leq \frac{2\max \left\{ [e^{\frac{q_0}{a}} \log h]^{\frac{d}{2}}, \log (h^{d/2}) \right\} }{h^{\frac{d}{2}}}
\nonumber
\end{align}
as we wished to prove.

\subsection{Proof of Proposition \ref{prop:intbound_sg_central}, Part (ii).}

The proof strategy is to split $U(a,h)$ as the sum of the integral for $v \in (0,0.5)$ plus that for $v \in (0.5,1)$. The first integral is then bound using Part (i) of this proposition, and the second integral using Lemma \ref{lem:tail_quadform_subgaussian}(iii).

Take any $a'>1 + \frac{q_0^{1/2}}{(q_0 + a \log\log h)^{1/2}}$. Then the integral for $v \in (0,0.5)$ is
\begin{align}
& \int_0^{0.5} P \left( \frac{u^Tu}{\sigma^2} > d a' \log \left( \frac{h^{\frac{a}{a'}}}{[1/v-1]^{\frac{2a}{d a'}}} \right) \right) dv
\nonumber \\
&< \int_0^{0.5} P \left( \frac{u^T u}{\sigma^2} > d a' \log \left( \frac{h^{\frac{a}{a'}}}{[1/v-1]^{\frac{2}{d}}} \right) \right) dv,
\label{seq:tailint_bound1}
\end{align}
since 
\begin{align}
 \frac{1}{[1/v-1]^{\frac{a}{a'}}}= \left( \frac{v}{1-v} \right)^{\frac{a}{a'}} > \frac{v}{1-v},
\nonumber
\end{align}
given that $v/(1-v) < 1$ for $v \in (0,0.5)$ and that $a/a'<1$. Applying Part (i) of the current Proposition \ref{prop:intbound_sg_central} gives that \eqref{seq:tailint_bound1} is
\begin{align}
\leq \frac{2\max \left\{ [e^{\frac{q_0}{a'}} \log (h^{\frac{a}{a'}})]^{\frac{d}{2}}, \log (h^{\frac{ad}{2a'}}) \right\} }{h^{\frac{ad}{2a'}}}.
\label{seq:tailint_bound2}
\end{align}

Next consider the integral for $v \in (0.5,1)$,
\begin{align}
 \int_{0.5}^1 P \left( \frac{u^Tu}{\sigma^2} > d \log \left( h^a \left[ \frac{v}{1-v} \right]^{2a/d} \right) \right) dv
\nonumber
\leq \int_{0.5}^1 P \left( \frac{u^T u}{\sigma^2} > d \log h^a \right) dv=
0.5  P \left( \frac{u^T u}{\sigma^2} > d \log h^a \right),
\nonumber
\end{align}
where in the inequality above we used that $v/(1-v)>1$ for $v \in (0.5,1)$.

We may now use Lemma \ref{lem:tail_quadform_subgaussian}(iii) setting $q= \log h^a$, since $\log h^a > 2 (1+2^{1/2})^2$ by assumption, giving that
\begin{align}
< \frac{1}{2} \exp \left\{ -\frac{d \log(h^a)}{2(1+k_0)} \right\}
= \frac{1}{2} \left( \frac{1}{h} \right)^{\frac{ad}{2(1+k_0)}},
\nonumber
\end{align}
for $k_0= 2^{1/2}(1+2^{1/2})/(a \log h)^{1/2}$.
Combining this expression with \eqref{seq:tailint_bound2} gives that
\begin{align}
 U(a,h) \leq \frac{2\max \left\{ [e^{\frac{q_0}{a'}} \log (h^{\frac{a}{a'}})]^{\frac{d}{2}}, \log (h^{\frac{ad}{2a'}}) \right\} }{h^{\frac{ad}{2a'}}}
+ \frac{1}{2} \left( \frac{1}{h} \right)^{\frac{ad}{2(1+k_0)}},
\nonumber
\end{align}
where recall that $k_0= 2^{1/2}(1+2^{1/2})/(a \log h)^{1/2}$, as we wished to prove.

As a final remark, note that when $a' > 1 + k_0$ the second term is smaller than the first one.
Further note that
\begin{align}
\log \left( \frac{h}{\log h} \right) > \frac{q_0}{a} \Leftrightarrow 
q_0 + a \log\log h < a \log h  \Rightarrow a' > 1 + k_0.
\nonumber
\end{align}
Hence if $\log \left( h/\log h \right) > q_0/a$ we obtain
\begin{align}
 U(a,h) \leq \frac{3\max \left\{ [e^{\frac{q_0}{a'}} \log (h^{\frac{a}{a'}})]^{\frac{d}{2}}, \log (h^{\frac{ad}{2a'}}) \right\} }{h^{\frac{ad}{2a'}}}.
\nonumber
\end{align}

\subsection{Proof of Proposition \ref{prop:intbound_sg_noncentral}, Parts (i)-(ii).}

The proof strategy is to split the integral into two terms, where the first term is an integral involving a central sub-Gaussian that can be bound using Proposition \ref{prop:intbound_sg_central}, and the second term involves a inequality that can be bound with Lemma \ref{lem:lefttail_quadform_subgaussian}.

Denote by $w= a \log(h/[1/v-1]^2)$ and let $w'>0$ be an arbitrary number, then the union bound gives
\begin{align}
& P \left( \frac{bu_1^Tu_1 - u_2^Tu_2}{\sigma^2}  > w \right) =
P \left(  \frac{bu_1^Tu_1 - u_2^Tu_2}{\sigma^2} > \frac{w}{2} + w' - (w' - \frac{w}{2})  \right)
\nonumber \\
&\leq P \left( \frac{b u_1^T u_1}{\sigma^2} > \frac{w}{2} + w' \right) +
P \left( \frac{u_2^Tu_2}{\sigma^2} < w' - \frac{w}{2} \right).
\label{seq:ineq_sgnoncentral}
\end{align}

We shall take $w'= w/2 + \mu^T\mu/[\sigma^2 \log \mu^T \mu]$ so that
$w' - w/2=\mu^T\mu/[\sigma^2 \log \mu^T\mu]$ and 
$w' + w/2= w + \mu^T\mu/[\sigma^2 \log \mu^T\mu]$.
Applying Lemma \ref{lem:lefttail_quadform_subgaussian} immediately gives that the second term in \eqref{seq:ineq_sgnoncentral}.
\begin{align}
 P \left( \frac{u_2^Tu_2}{\sigma^2} < w' - \frac{w}{2} \right) 
= P \left( \frac{u_2^Tu_2}{\sigma^2} < \frac{\mu^T\mu}{\sigma^2 \log \mu^T\mu} \right)
\leq \exp \left\{ -\frac{\mu^T\mu}{8\sigma^2} \left( 1 - \frac{1}{\log \mu^T\mu} \right)\right\}.
\nonumber
\end{align}

The first term in \eqref{seq:ineq_sgnoncentral} is
\begin{align}
 P \left( \frac{b u_1^T u_1}{\sigma^2} > a \log\left(\frac{h}{[1/v-1]^2}\right) + \frac{\mu^T\mu}{\sigma^2 \log \mu^T\mu} \right)=
 P \left( \frac{u_1^T u_1}{\sigma^2} > \frac{a}{b} \log\left(\frac{h e^{ \frac{\mu^T\mu}{a \sigma^2 \log \mu^T\mu}}}{[1/v-1]^2}\right) \right),
\nonumber
\end{align}
giving that
\begin{align}
 U(a,b,h) \leq \exp \left\{ -\frac{\mu^T\mu}{8\sigma^2} \left( 1 - \frac{1}{\log \mu^T\mu} \right)\right\} +
\int_0^1  P \left( \frac{u_1^T u_1}{\sigma^2} > \frac{d_1 a}{b} \log\left(\frac{h^{\frac{1}{d_1}} e^{ \frac{\mu^T\mu}{d_1 a \sigma^2 \log \mu^T\mu}}}{[1/v-1]^{2/d_1} }\right) \right) dv.
\label{seq:ineqint_sgnoncentral}
\end{align}

To bound the second term in \eqref{seq:ineqint_sgnoncentral}, we first split that integral according to the range of $v$ values such that the term inside the log is negative and positive. Note that
\begin{align}
\frac{h e^{ \frac{\mu^T\mu}{a \sigma^2 \log \mu^T\mu}}}{[1/v-1]^{2}} > 1
\Leftrightarrow
h e^{ \frac{\mu^T\mu}{a \sigma^2 \log \mu^T\mu}} > [1/v-1]^{2}
\Leftrightarrow
v > \left( 1 + h^{\frac{1}{2}} e^{ \frac{\mu^T\mu}{2a \sigma^2 \log \mu^T\mu}} \right)^{-1}=v_0,
\nonumber
\end{align}
where we denoted the right-hand side $v_0$ for convenience, hence the second term in \eqref{seq:ineqint_sgnoncentral} is
\begin{align}
\leq v_0 + \int_{v_0}^1  P \left( \frac{u_1^T u_1}{\sigma^2} > \frac{d_1 a}{b} \log\left(\frac{h^{\frac{1}{d_1}} e^{ \frac{\mu^T\mu}{d_1 a \sigma^2 \log \mu^T\mu}}}{[1/v-1]^{2/d_1} }\right) \right) dv.
\label{seq:ineqint_sgnoncentral_term2}
\end{align}

To bound the second term in \eqref{seq:ineqint_sgnoncentral_term2} we use Proposition \ref{prop:intbound_sg_central}. We do this separately for Part (i) and (ii).

\vspace{3mm}
\underline{{\bf Part (i).}} By assumption we have that 
\begin{align}
 \frac{a}{b} > 1 + \frac{q_0^{1/2}}{\left(q_0 + \frac{a}{b d_1} \log\log (h e^{ \frac{\mu^T\mu}{a \sigma^2 \log \mu^T\mu}})\right)^{\frac{1}{2}}},
\nonumber
\end{align}
where $q_0= 2(1+2^{1/2})^2$. 
To apply Then Proposition \ref{prop:intbound_sg_central}(i) we also need that
\begin{align}
h^{\frac{1}{d_1}} e^{ \frac{\mu^T\mu}{d_1 a \sigma^2 \log \mu^T\mu}}  > e
\Leftrightarrow
\log(h) + \frac{\mu^T\mu}{a \sigma^2 \log \mu^T\mu} > d_1,
\nonumber
\end{align}
which also holds by assumption.

Then Proposition \ref{prop:intbound_sg_central}(i) gives that the second term in \eqref{seq:ineqint_sgnoncentral_term2} is
\begin{align}
& \leq \frac{2 \max \left\{ e^{\frac{q_0b}{a}} \log \left(h^{\frac{1}{d_1}} e^{ \frac{\mu^T\mu}{d_1 a \sigma^2 \log \mu^T\mu}}  \right) , \log \left( h^{\frac{1}{2}} e^{ \frac{\mu^T\mu}{2 a \sigma^2 \log \mu^T\mu}} \right)  \right\} }{h^{\frac{1}{2}} e^{ \frac{\mu^T\mu}{2 a \sigma^2 \log \mu^T\mu}}}
= \frac{2 \max \left\{ \frac{e^{\frac{q_0b}{a}}}{d_1/2} \log \left(r \right) , \log \left( r \right)  \right\} }{r}
\nonumber
\end{align}
where $r=h^{\frac{1}{2}} e^{ \frac{\mu^T\mu}{2 a \sigma^2 \log \mu^T\mu}}$.
Combining this expression with \eqref{seq:ineqint_sgnoncentral} and \eqref{seq:ineqint_sgnoncentral_term2} and noting that $v_0= (1+h_0)^{-1} \leq 1/h_0$ gives that
\begin{align}
 U(a,b,h) \leq \exp \left\{ -\frac{\mu^T\mu}{8\sigma^2} \left( 1 - \frac{1}{\log \mu^T\mu} \right)\right\} +
\frac{1}{r} +
\frac{2 \max \left\{ \frac{e^{\frac{q_0b}{a}}}{d_1/2} \log \left(r \right) , \log \left( r \right)  \right\} }{r},
\nonumber
\end{align}
as we wished to prove.

\vspace{3mm}
\underline{{\bf Part (ii).}}
By assumption we have that $a/b \leq 1$ and that
\begin{align}
 \log \left( h^{1/d_1} e^{\frac{\mu^T\mu}{d_1 a \sigma^2 \log \mu^T\mu}} \right) > \frac{2 (1+2^{1/2})^2 b}{a}
\Leftrightarrow
 \log( h ) + \frac{\mu^T\mu}{a \sigma^2 \log \mu^T\mu} > \frac{q_0 b d_1}{a},
\nonumber
\end{align}
where recall that $q_0=2 (1+2^{1/2})^2$, which are the two conditions to apply Proposition \ref{prop:intbound_sg_central}(i) to the second term in \eqref{seq:ineqint_sgnoncentral_term2}.
Hence Proposition \ref{prop:intbound_sg_central}(i) gives that the second term in \eqref{seq:ineqint_sgnoncentral_term2} is
\begin{align}
 \leq \frac{2 \max \left\{ \left[ \frac{a e^{\frac{q_0}{a'}}}{a'b d_1/2} \log r \right]^{\frac{d_1}{2}} , \log \left( r^{\frac{a}{b a'}} \right)  \right\}}{r^{\frac{a}{b a'}}}
+ \frac{1}{2 r^{\frac{a}{b (1+k_0)}}}
\nonumber
\end{align}
where $r= h^{1/2} e^{\mu^T\mu/[2 a \sigma^2 \log \mu^T\mu]}$,
$a'= 1 + q_0^{1/2}/[q_0 + a \log \log r / (b d_1/2)]^{1/2}$ and
$k_0= q_0^{1/2}/[(2 a/b) \log r]^{1/2}$.
Combining this expression with \eqref{seq:ineqint_sgnoncentral} and \eqref{seq:ineqint_sgnoncentral_term2} and noting that $v_0= (1+h_0)^{-1} \leq 1/h_0$ gives that $U(a,b,h) \leq$
\begin{align}
 &\exp \left\{ -\frac{\mu^T\mu}{8\sigma^2} \left( 1 - \frac{1}{\log \mu^T\mu} \right)\right\} +
\frac{1}{r} +
\frac{2 \max \left\{ \left[ \frac{a e^{\frac{q_0}{a'}}}{a'b d_1/2} \log r \right]^{\frac{d_1}{2}} , \log \left( r^{\frac{a}{b a'}} \right)  \right\}}{r^{\frac{a}{b a'}}}
+ \frac{1}{2 r^{\frac{a}{b (1+k_0)}}}
\nonumber \\
&<
\exp \left\{ -\frac{\mu^T\mu}{8\sigma^2} \left( 1 - \frac{1}{\log \mu^T\mu} \right)\right\} +
\frac{2 \max \left\{ \left[ \frac{a e^{\frac{q_0}{a'}}}{a'b d_1/2} \log r \right]^{\frac{d_1}{2}} , \log \left( r^{\frac{a}{b a'}} \right)  \right\}}{r^{\frac{a}{b a'}}}
+ \frac{3}{2 r^{\frac{a}{b (1+k_0)}}}
\nonumber
\end{align}
since $a/[b(1+k_0)] < 1$ and $r>1$, as we wished to prove.

\subsection{Proof of Proposition \ref{prop:intbound_sg_noncentral}, Part (iii).}

The proof strategy is to note that, since $b u_1^Tu_1 - u_2^Tu_2 \leq b u_1^Tu_1$, it holds that
\begin{align}
 U(a,b,h) \leq \int_0^1 P \left( \frac{u_1^Tu_1}{\sigma^2} > d_1 \frac{a}{b} \log \left( \frac{h^{2/d_1}}{(1/v-1)^{1/d_1}} \right) \right) dv,
\nonumber
\end{align}
where the right-hand side can be bound using the result for central sub-Gaussians in Proposition \ref{prop:intbound_sg_central}.
Specifically, since $a/b \leq 1$ and $\log(h^{1/d_1}) > 2(1+2^{1/2})^2 b/a$ by assumption, we may directly apply Proposition \ref{prop:intbound_sg_central}(ii) to obtain that
\begin{align}
 U(a,b,h) \leq \frac{2 \max \left\{([e^{q_0}/d_1] \log(h^{\frac{a}{b a'}}))^{d_1/2}  , \log(h^{\frac{a}{2b a'}}) \right\}}{h^{\frac{a}{2b a'}}}
+ \frac{1}{2 h^{\frac{a}{2b(1+k_0)}}}
\nonumber
\end{align}
for any 
$$
a'> 1 + \frac{q_0^{1/2}}{[q_0 + (a/b) \log \log (h^{1/d_1})]^{1/2}},
$$
where $q_0= 2(1+2^{1/2})^2$ and $k_0= [q_0b d_1/(a \log h)]^{1/2}$.

\newpage
\section{Proof of Theorem \ref{thm:pp}}
\label{supplsec:proof_pp}

Let $S= \left\{ \gamma : \gamma^* \subset \gamma \right\}$ be the set of overfitted models and $S^c= \left\{  \gamma: \gamma^* \not\in \gamma \right\}$ that of non-overfitted models.
Their expected posterior probabilities under the data-generating $F$ are
\begin{align}
 E_F\left[ p(S \mid y) \right]&=
\sum_{l=|\gamma^*|_0+1}^{\bar{q}} \sum_{\gamma \in S, |\gamma|_0=l} E_F \left( p(\gamma \mid y) \right)
\label{seq:pp_spurious} \\
 E_F\left[ p(S^c \mid y) \right]&= \sum_{l=0}^{\bar{q}} \sum_{\gamma \in S^c, |\gamma|_0=l} E_F \left( p(\gamma \mid y) \right)
\label{seq:pp_nonspurious}
\end{align}
where $\bar{q}$ is the maximum model size (as defined by the model space prior), and
\begin{align}
 E_F \left( p(\gamma \mid y) \right)=
E_F \left( \frac{1}{1 + \sum_{\gamma' \neq \gamma} B_{\gamma' \gamma} \frac{p(\gamma')}{p(\gamma)}} \right)
\leq E_F \left( \frac{1}{1 + B_{\gamma^* \gamma} \frac{p(\gamma^*)}{p(\gamma)}} \right).
\nonumber
\end{align}
Since the right-hand side is the expectation of a positive random variable taking values in $[0,1]$, it may be obtained by integrating its survival (or right-tail probability) function, that is
\begin{align}
 & E_F (p(\gamma \mid y)) \leq
\int_0^1 P_F \left( \left[ 1 + B_{\gamma^* \gamma} \frac{p(\gamma^*)}{p(\gamma)} \right]^{-1} > v \right) dv
=\int_0^1 P_F \left( B_{\gamma \gamma^*} > \frac{p(\gamma^*)/p(\gamma)}{1/v -1} \right) dv.
\label{seq:pp_singlemodel}
\end{align}

For later reference note that under our model space prior
\begin{align}
 \frac{p(\gamma)}{p(\gamma^*)}= \frac{p(|\gamma|_0)}{p(|\gamma^*|_0)} \frac{{q \choose |\gamma^*|_0}}{{q \choose |\gamma|_0}}
= q^{-c(|\gamma|_0 - |\gamma^*|_0)} \frac{{q \choose |\gamma^*|_0}}{{q \choose |\gamma|_0}}
= q^{-c(|\gamma|_0 - |\gamma^*|_0)} \frac{{|\gamma|_0 \choose |\gamma^*|_0 }}{{q - |\gamma^*|_0 \choose |\gamma|_0 - |\gamma^*|_0}},
\label{seq:ratio_priormodelprob}
\end{align}
provided that both model sizes $|\gamma|_0 \leq \bar{q}$ and $|\gamma^*|_0 \leq \bar{q}$,
where $c \geq 0$ is the Complexity prior's parameter, and recall that $c=0$ corresponds to a Beta-Binomial(1,1) prior.
Note also that using \eqref{seq:bf_term1} and \eqref{eq:ineq_bayesssdif} gives that for large enough $n$
\begin{align}
 B_{\gamma \gamma^*} \leq
\begin{cases}
 (g n k_2)^{\frac{|\gamma^*|_0 - |\gamma|_0}{2}} \exp \left\{ (s_\gamma - s_{\gamma^*})(1 + \delta)/2 \right\} \mbox{, if } s_\gamma - s_{\gamma^*} \geq 0
\\
 (g n k_2)^{\frac{|\gamma^*|_0 - |\gamma|_0}{2}} \exp \left\{ (s_\gamma - s_{\gamma^*})(1 - \delta)/2 \right\} \mbox{, if } s_\gamma - s_{\gamma^*} <0
\end{cases}
\nonumber
\end{align}
where $k_2= \bar{l}_{\gamma^*}(1 + \delta) / \underline{l}_\gamma$, $\delta$ is a constant that can be taken arbitrarily close to 0, 
$\bar{l}_{\gamma^*}$ and $\underline{l}_\gamma$ are the eigenvalues defined in Condition (A2) (i.e. $k_2$ is bounded by a constant under (A2))
and \begin{align}
 s_\gamma= \tilde{y}^T \widetilde{W}_\gamma (\widetilde{W}_\gamma^T \widetilde{W}_\gamma)^{-1} \widetilde{W}_\gamma^T \tilde{y}=
\tilde{\eta}_\gamma^T \widetilde{W}_\gamma^T \widetilde{W}_\gamma \tilde{\eta}_\gamma
\nonumber
\end{align}
is the sum of explained squares by the least-squares estimator under model $\gamma$.

The proof strategy is to bound $E_f ( p(\gamma \mid y) )$ separately for overfitted $\gamma \in S$ and non-overfitted $\gamma \in S^c$, and then carrying out the deterministic sums in \eqref{seq:pp_spurious} and \eqref{seq:pp_nonspurious}.

\subsection{Single overfitted model}

Using \eqref{seq:pp_singlemodel} and \eqref{seq:bfrighttail_overfitted}, we have that
\begin{align}
 E_F (p(\gamma \mid y)) &\leq 
\int_0^1 P_F \left( B_{\gamma \gamma^*} > \frac{p(\gamma^*)/p(\gamma)}{1/v -1} \right) dv
\nonumber \\
&\leq \int_0^1 P_F \left( \frac{u^T u}{\tilde{\omega}}  > \frac{|\gamma|_0 - |\gamma^*|_0}{\tilde{\omega} (1 + \delta)} \log \left( g n k_2 \left[ \frac{p(\gamma^*)/p(\gamma)}{1/v-1} \right]^{\frac{2}{|\gamma|_0 - |\gamma^*|_0}} \right) \right) dv,
\label{seq:modelppbound_overfitted}
\end{align}
where $u \sim SG(0, \tilde{\omega})$ is a $|\gamma|_0 - |\gamma^*|_0$ dimensional sub-Gaussian vector, $\tilde{\omega}=\omega \tau$, $\omega$ is the sub-Gaussian dispersion parameter associated to $F$ and $\tau$ the largest eigenvalue of $\Sigma^{-1}$.

Since \eqref{seq:modelppbound_overfitted} corresponds to the integral $U(a,h)$ in Proposition \ref{prop:intbound_sg_central}, setting
$a= 1/[\tilde{\omega} (1 + \delta)]$ and $h= gnk_2 [p(\gamma^*)/p(\gamma)]^{2/(|\gamma|_0 - |\gamma^*|_0)}$,
and sub-Gaussian parameters $\sigma^2= \tilde{\omega}$ and $d= |\gamma|_0 - |\gamma^*|_0$.
Condition (B2) implies that $a \leq 1$, so that we may apply Part (ii) of Proposition \ref{prop:intbound_sg_central}. Note that if Condition (B2) did not hold then we would apply Part (i), which corresponds to a faster rate, i.e. (B2) considers a worse-case scenario.
Note also that Part (ii) requires that
\begin{align}
 \log h > \frac{q_0}{a}= q_0 \omega \tau (1+\delta)
\nonumber
\end{align}
which holds since
\begin{align}
\log h= \log(gnk_2) + \frac{2}{|\gamma|_0 - |\gamma^*|_0} \log \frac{p(\gamma^*)}{p(\gamma)} \geq \log(gnk_2)
\label{seq:logh_grows}
\end{align}
and the right-hand side diverges to infinity under Assumption (A4), where we used that $p(\gamma^*) \geq p(\gamma)$ for $|\gamma|_0 > |\gamma^*|_0$, from Assumption (B3).

Hence we may apply Proposition \ref{prop:intbound_sg_central}(ii) to \eqref{seq:modelppbound_overfitted}, obtaining
\begin{align}
 E_F (p(\gamma \mid y)) \leq 
\frac{2 \max \left\{  [e^{q_0} \log(h^{a/a'})]^{d/2}, \log\left(h^{\frac{ad}{2a'}}\right)\right\}}{h^{\frac{ad}{2a'}}} + \frac{1}{2 h^{\frac{ad}{2(1+k_0)}}},
\label{seq:modelppbound_overfitted2}
\end{align}
where $k_0= q_0^{1/2}/(a \log h)^{1/2}$ and 
$a'= 1 + q_0^{1/2}/(q_0 + a \log\log h)^{1/2}$.
Using \eqref{seq:logh_grows}, Assumption (B3) gives that $\log h > q_0/a + \log\log h$ and hence that $1+ k_0 \leq a'$, thus two times the second term in the right-hand side of \eqref{seq:modelppbound_overfitted2} is smaller than the first term.
Noting that as $n$ grows one may take $a'$ arbitrarily close to 1, it is also simple to see that the first term in \eqref{seq:modelppbound_overfitted2} is upper-bounded by $2 e^{q_0}/h^{\frac{ad(1-\epsilon)}{2}}$, for any fixed $\epsilon >0$, and all $n > n_0$ for some fixed $n_0$. Combining these observations gives that
\begin{align}
& E_F (p(\gamma \mid y)) \leq 2.5 h^{-\frac{ad(1-\epsilon)}{2}}
=2.5 \left( \frac{e^{q_0}}{(g n k_2)^{a (1-\epsilon)}} \right)^{\frac{d}{2}} \left( \frac{p(\gamma)}{p(\gamma^*)} \right)^{a (1-\epsilon)}
\nonumber \\
&=2.5 \left( \frac{b}{(g n)^{\frac{1-\epsilon}{\omega \tau}}} \right)^{\frac{|\gamma|_0 - |\gamma^*|_0}{2}} \left( \frac{p(\gamma)}{p(\gamma^*)} \right)^{\frac{1-\epsilon}{\omega \tau}}
 \label{seq:modelppbound_overfitted3}
\end{align}
for all $n \geq n_0$, where $b= e^{2(1+2^{1/2})^2}/k_2^{(1-\epsilon)/[\omega \tau]}$ is a constant.

\subsection{Sum across overfitted models}

Plugging \eqref{seq:modelppbound_overfitted3} into \eqref{seq:pp_spurious} gives that
\begin{align}
 E_F[p(S \mid y)] \leq \sum_{l=|\gamma^*|_0+1}^{\bar{q}} \sum_{|\gamma|_0=l, \gamma \supset \gamma} 
2.5 \left( \frac{b}{(g n)^{\frac{1-\epsilon}{\omega \tau}}} \right)^{\frac{|\gamma|_0 - |\gamma^*|_0}{2}} 
\left( \frac{p(\gamma)}{p(\gamma^*)} \right)^{\frac{1-\epsilon}{\omega \tau}}
\label{seq:pp_spurious_bound}
\end{align}
for all $n \geq n_0$ and some fixed $n_0$ and $b>0$, where $\epsilon$ is a constant that may be taken arbitrarily close to 0 as $n$ grows.

To prove the desired result we plug in the expression for $p(\gamma)/p(\gamma^*)$, and then use algebraic manipulation and Lemma \ref{slem:binomial_ogf} to carry out the summation. 
To alleviate upcoming expressions let $r= (1-\epsilon)/[\omega \tau]$, and recall that $r<1$.
Plugging \eqref{seq:ratio_priormodelprob} into \eqref{seq:pp_spurious_bound} gives
\begin{align}
  E_F[p(S \mid y)] \leq 2.5 \sum_{l=|\gamma^*|_0+1}^{\bar{q}} 
\left( \frac{b}{(g n)^{r}} \right)^{\frac{l - |\gamma^*|_0}{2}} 
q^{-cr(l - |\gamma^*|_0)} {l \choose |\gamma^*|_0 }^{r} {q - |\gamma^*|_0 \choose l - |\gamma^*|_0}^{-r}
\sum_{|\gamma|_0=l, \gamma \supset \gamma^*}  1
\nonumber \\
= 2.5 \sum_{l=|\gamma^*|_0+1}^{\bar{q}} 
\left( \frac{b^{1/2}}{q^{cr} (g n)^{r/2}} \right)^{l - |\gamma^*|_0}
{l \choose |\gamma^*|_0 }^r {q - |\gamma^*|_0 \choose l - |\gamma^*|_0}^{1-r}
\nonumber \\
<
2.5 \sum_{l=|\gamma^*|_0+1}^{\bar{q}} 
\left( \frac{b^{1/2} q^{1-r}}{q^{cr} (g n)^{r/2}} \right)^{l - |\gamma^*|_0}
{l \choose |\gamma^*|_0 } 
\label{seq:pp_spurious_bound2}
\end{align}
where in the second line of \eqref{seq:pp_spurious_bound2} we used that $\sum_{\gamma \in S, |\gamma|_0=l}  1= {q - |\gamma^*|_0 \choose l - |\gamma^*|_0}$ is the number of spurious models adding $l - |\gamma^*|_0$ out of the $q - |\gamma^*|_0$ spurious parameters,
and in the third line of \eqref{seq:pp_spurious_bound2} we used that
${l \choose |\gamma^*|_0}^r < {l \choose |\gamma^*|_0}$ (since $r \in (0,1)$) and that
${q - |\gamma^*|_0 \choose l - |\gamma^*|_0}  < (q-|\gamma^*|_0)^{l - |\gamma^*|_0} \leq q^{l - |\gamma^*|_0}$.

Lemma \ref{slem:binomial_ogf} gives that the right-hand side of \eqref{seq:pp_spurious_bound2} is
\begin{align}
< 2.5
\left[ \left(1 - \frac{b^{1/2}}{q^{r(c + 1) -1} (g n)^{\frac{r}{2}}} \right)^{-(|\gamma^*|_0 +1)} - 1 \right].
\label{seq:pp_spurious_bound3}
\end{align}

Lemma \ref{slem:binomial_ogf} also gives a simpler asymptotic version of \eqref{seq:pp_spurious_bound3}, under the condition that
\begin{align}
\lim_{n \rightarrow \infty} q^{r(c + 1) -1} (g n)^{r/2} = \infty
\Leftrightarrow \lim_{n\rightarrow \infty} q^{\frac{1-\epsilon}{2 \omega \tau}(c + 1) -1} (g n)^{\frac{1-\epsilon}{2 \omega \tau}}= \infty
\nonumber
\end{align}
which holds under Condition (B4).
Hence, applying Lemma \ref{slem:binomial_ogf} gives that
\begin{align}
2.5
\frac{b^{1/2} (|\gamma^*|_0 +1)}{q^{r(c + 1) -1} (g n)^{\frac{r}{2}}},
\nonumber
\end{align}
which proves our stated result.

\subsection{Single non-overfitted model}

From \eqref{seq:bftailprob_notsubsetcase}, for any non-overfitted model $\gamma \not\supset \gamma^*$ we have that
\begin{align}
 E_F (p(\gamma \mid y)) &\leq 
\int_0^1 P_F \left( B_{\gamma \gamma^*} > \frac{p(\gamma^*)/p(\gamma)}{1/v -1} \right) dv
\nonumber \\
&\leq \int_0^1 P_F \left( \frac{\frac{1+\delta}{1-\delta}u_1^T u_1 - u_2^T u_2}{\tilde{\omega}}  > \frac{1}{\tilde{\omega} (1 + \delta)} \log \left( \frac{(g n k_2)^{|\gamma|_0 - |\gamma^*|_0} [p(\gamma^*)/p(\gamma)]^2}{(1/v-1)^2} \right) \right) dv,
\label{seq:modelppbound_nonoverfitted}
\end{align}
where $\tilde{\omega}=\omega \tau$ and $\delta$ are as in \eqref{seq:modelppbound_overfitted},
$u_1 \sim SG(0,\tilde{\omega})$ has dimension $|\gamma'|_0 - |\gamma^*|_0$, $\gamma'= \gamma \cup \gamma^*$ is the model with design matrix combining all columns in $\gamma$ and those in $\gamma^*$, and $u_2 \sim SG(\mu_\gamma, \tilde{\omega})$ is a $|\gamma'|_0-|\gamma|_0$ dimensional sub-Gaussian vector with $\mu_\gamma= (Z_{\gamma' \setminus \gamma}^T Z_{\gamma' \setminus \gamma})^{-1/2} Z_{\gamma' \setminus \gamma}^T (\widetilde{W}_{\gamma^*} \eta_{\gamma^*}^* - \widetilde{W}_\gamma \eta_\gamma^*) \neq 0$.
Recall that $\delta$ should be thought of as a constant arbitrarily close to 0, that $\tilde{\omega}$ is bounded by constants under our assumptions, and that from \eqref{seq:noncentrality_param} the non-centrality parameter can be written as
\begin{align}
\lambda_\gamma= \mu_\gamma^T\mu_\gamma= (\widetilde{W}_{\gamma^*} \eta_{\gamma^*}^*)^T (I - H_\gamma) \widetilde{W}_{\gamma^*} \eta_{\gamma^*}^*.
\end{align}

The strategy is to note that \eqref{seq:modelppbound_nonoverfitted} is an integral of the form considered in Proposition \ref{prop:intbound_sg_noncentral}.
Specifically in Proposition \ref{prop:intbound_sg_noncentral} take
$\sigma^2= \tilde{\omega}$, $b=(1+\delta)/(1-\delta)$, $a= 1/[\tilde{\omega} (1+\delta)]$,
and $h=(g n k_2)^{|\gamma|_0 - |\gamma^*|_0} [ p(\gamma^*)/p(\gamma) ]^2$, and the sub-Gaussian dimensions to be
$d_1=|\gamma'|_0 - |\gamma^*|_0$ and $d_2=|\gamma'|_0 - |\gamma|_0$.
Proposition \ref{prop:intbound_sg_noncentral}(ii) considers the case where $a/b \leq 1$, whereas Part (i) considers $a/b$ that is sufficiently larger than 1. 
Since $\tilde{\omega}=\omega \tau > 1$ by Assumption (B2), we have that
\begin{align}
 \frac{a}{b} = \frac{(1-\delta)}{\tilde{\omega} (1+\delta)^2} < 1.
\nonumber
\end{align}
Note that if Assumption (B2) were not to hold then we could apply Proposition \ref{prop:intbound_sg_noncentral}(ii), which leads to faster rates.

The other condition to apply Proposition \ref{prop:intbound_sg_noncentral}(ii) is that
\begin{align}
& \log(h) + \frac{\lambda_\gamma}{a \sigma^2 \log \lambda_\gamma} > \frac{q_0bd_1}{a}
\Leftrightarrow
\nonumber \\
& (|\gamma|_0 - |\gamma^*|_0) \log(g n k_2) + 2 \log\left( \frac{p(\gamma^*)}{p(\gamma)} \right)
+ \frac{(1 + \delta) \lambda_\gamma}{\log \lambda_\gamma} > \frac{q_0 (|\gamma'|_0 - |\gamma^*|_0) \omega \tau (1+\delta)^2}{(1-\delta)}
\nonumber
\end{align}
where $q_0= 2(1+2^{1/2})^2$. This condition follows from Assumption (B5) for models of size $|\gamma|_0 \leq |\gamma^*|_0$, and from (B5') for models of size $|\gamma|_0 > |\gamma^*|_0$.
To see this, (B5) implies
\begin{align}
 (|\gamma|_0 - |\gamma^*|_0) \log(g n k_2) + 2 \log\left( \frac{p(\gamma^*)}{p(\gamma)} \right)
+ \frac{(1 + \delta) \lambda_\gamma}{\log \lambda_\gamma} > \frac{q_0 |\gamma|_0 \omega \tau (1+\delta)^2}{(1-\delta)}
\nonumber \\
\Rightarrow (|\gamma|_0 - |\gamma^*|_0) \log(g n k_2) + 2 \log\left( \frac{p(\gamma^*)}{p(\gamma)} \right)
+ \frac{(1 + \delta) \lambda_\gamma}{\log \lambda_\gamma} > \frac{q_0 (|\gamma'|_0 - |\gamma^*|_0) \omega \tau (1+\delta)^2}{(1-\delta)},
\nonumber
\end{align}
where we used that $|\gamma'|_0-|\gamma^*|_0 \leq |\gamma|_0$, and that $q_0$, $\omega$, $\tau$ and $\delta$ are constants.

Hence we may apply Proposition \ref{prop:intbound_sg_noncentral}(ii) to bound \eqref{seq:modelppbound_nonoverfitted}, obtaining that $ E_F (p(\gamma \mid y))$
\begin{align}
 \leq  
 \exp \left\{ -\frac{\lambda_\gamma}{8\sigma^2} \left( 1 - \frac{1}{\log \lambda_\gamma} \right)\right\} +
\frac{2 \max \left\{ \left[ \frac{a e^{\frac{q_0}{a'}}}{a'b d_1/2} \log r \right]^{\frac{d_1}{2}} , \log \left( r^{\frac{a}{b a'}} \right)  \right\}}{r^{\frac{a}{b a'}}}
+ \frac{3/2}{r^{\frac{a}{b (1+k_0)}}},
\label{seq:modelppbound_nonoverfitted2}
\end{align}
where $r=h^{1/2} e^{\lambda_\gamma/[2 \log \lambda_\gamma]]}$,
 $a'= 1 + q_0^{1/2}/[q_0 + a \log \log (2r / (b d_1))]^{1/2}$ and
$k_0= q_0^{1/2}/[(2 a/b) \log r]^{1/2}$.
For the sake of precision, in the particular case where $\gamma \subset \gamma^*$ is a strict subset of the optimal $\gamma^*$ (i.e. $\gamma$ misses some parameters from $\gamma^*$, but does not add any spurious parameters), then the second and third terms in Expression \eqref{seq:modelppbound_nonoverfitted2} are equal to zero (see the proposition's proof). To simplify the upcoming presentation we keep these terms in the remainder of the proof, however, with the understanding that we define $d_1=1$ in these cases.

The rest of this sub-section is devoted to simplifying \eqref{seq:modelppbound_nonoverfitted2}.
To simplify \eqref{seq:modelppbound_nonoverfitted2} it is possible to show that under Assumption (B5)  $\lim_{n \rightarrow \infty} r/bd_1= \infty$, hence $a'$ may be taken arbitrarily close to 1 and $k_0$ arbitrarily close to 0 as $n$ grows, and the third term in \eqref{seq:modelppbound_nonoverfitted2} is asymptotically smaller than the second term. 
To see this, (B5) implies that
\begin{align}
&\lim_{n \rightarrow \infty} (|\gamma|_0 - |\gamma^*|_0) \log(gnk_2) + 2 \log \left( \frac{p(\gamma^*)}{p(\gamma)} \right) + \frac{\lambda_\gamma}{\log \lambda_\gamma} - 2 \log(|\gamma|_0)= \infty \Leftrightarrow
\nonumber \\
&\lim_{n \rightarrow \infty} \log(h) + \frac{\lambda_\gamma}{\log \lambda_\gamma} - 2 \log(|\gamma|_0)= \infty
\Rightarrow
\nonumber \\
&\lim_{n \rightarrow \infty} \frac{1}{2}\log(h) + \frac{\lambda_\gamma}{2\log \lambda_\gamma} - \log(|\gamma'|_0 - |\gamma^*|_0)= \infty
\Leftrightarrow
\nonumber \\
&\lim_{n \rightarrow \infty} \frac{h^{1/2} e^{\lambda_\gamma/[2 \log \lambda_\gamma]]}}{b (|\gamma'|_0 - |\gamma^*|_0)}= \infty 
\Leftrightarrow \lim_{n \rightarrow \infty} \frac{r}{bd_1}= \infty,
\nonumber
\end{align}
since $b=(1+\delta)/(1-\delta)$ is a constant, and $|\gamma|_0 \geq |\gamma'|_0 - |\gamma^*|_0$.

These observations imply that there is a fixed $n_0$ such that for all $n \geq n_0$ we have
\begin{align}
 E_F (p(\gamma \mid y)) \leq  
 \exp \left\{ -\frac{\lambda_\gamma}{8\sigma^2} \left( 1 - \frac{1}{\log \lambda_\gamma} \right)\right\} +
\frac{3.5 \max \left\{ \left[ \frac{2 e^{q_0}}{\tilde{\omega}} \log r \right]^{\frac{|\gamma|_0}
{2}} , \log \left( r^{\frac{1}{\tilde{\omega}}} \right)  \right\}}{r^{\frac{1}{\tilde{\omega}}}}.
\label{seq:modelppbound_nonoverfitted3}
\end{align}
where we used that $d_1= |\gamma'|_0 - |\gamma^*|_0 \in [1, |\gamma|_0]$.

Although not essential to carry out the proof, this expression can be further simplified using Assumption (B5), by showing that both terms are asymptotically smaller than $r^{(1-\delta)/\tilde{\omega}}$ for any fixed $\delta>0$.
Clearly $\log(r^{\frac{1}{\tilde{\omega}}})/r^{\frac{1}{\tilde{\omega}}}$ is asymptotically smaller than $r^{(1-\delta)/\tilde{\omega}}$.
Hence, denoting $z= r^{1/\tilde{\omega}}= h^{1/[2\tilde{\omega}]} e^{\lambda_\gamma/[2 \tilde{\omega} \log \lambda_\gamma]}$, we just need to show that for any $\epsilon >0$
\begin{align}
&\lim_{n \rightarrow \infty} \frac{z^\epsilon}{[ \log z ]^{d_1/2}}= \infty
\Leftrightarrow
\lim_{n \rightarrow \infty} \epsilon \log(z) - \frac{d_1}{2} \log\log z= \infty
\Leftrightarrow
\nonumber \\
&\lim_{n \rightarrow \infty} \frac{\epsilon}{\tilde{\omega}} \left[\log(h) + \frac{\lambda_\gamma}{2 \tilde{\omega} \log \lambda_\gamma} \right]
- (|\gamma'|_0 - |\gamma^*|_0) \log \left( \log(h) + \frac{\lambda_\gamma}{2 \tilde{\omega} \log \lambda_\gamma} \right)= \infty.
\nonumber
\end{align}
Since $\tilde{\omega}$ is bounded by constants, and $|\gamma'|_0 - |\gamma^*|_0 \leq |\gamma|_0$, it suffices that
\begin{align}
 \lim_{n \rightarrow \infty} \left[\log(h) + \frac{\lambda_\gamma}{2 \tilde{\omega} \log \lambda_\gamma} \right]
- \frac{|\gamma|_0}{\epsilon} \log \left( \log(h) + \frac{\lambda_\gamma}{2 \tilde{\omega} \log \lambda_\gamma} \right)= \infty.
\nonumber
\end{align}
This latter condition holds Assumption (B5), since (B5) implies that for every fixed $\kappa = 1/\epsilon>0$ 
\begin{align}
\lim_{n \rightarrow \infty} \log(h) + \frac{(1+\delta) \lambda_\gamma}{\log \lambda_\gamma} - \frac{|\gamma|_0}{\epsilon} \log \left( \log(h) + \frac{\lambda_\gamma}{\log \lambda_\gamma} \right) = \infty
\nonumber
\end{align}

In conclusion, plugging  in $h=(g n k_2)^{|\gamma|_0 - |\gamma^*|_0} [ p(\gamma^*)/p(\gamma) ]^2$ gives that,
for every sufficiently large $n$,
\begin{align}
 E_F (p(\gamma \mid y)) \leq  
e^{ -\frac{\lambda_\gamma (1-\epsilon)}{8\tilde{\omega}}}
+ 3.5 \left( \frac{ p(\gamma)/p(\gamma^*) }{(gnk_2)^{\frac{|\gamma|_0 - |\gamma^*|_0}{2}} e^{\frac{\lambda_\gamma}{2 \tilde{\omega}}}} \right)^{1-\epsilon}
\label{seq:modelppbound_nonoverfitted4}
\end{align}
for any fixed $\epsilon>0$.

\subsection{Single non-overfitted model of size $|\gamma|_0 > |\gamma^*|_0$}

The goal is to bound \eqref{seq:modelppbound_nonoverfitted}. The strategy is to Proposition \ref{prop:intbound_sg_noncentral}(iii).
Specifically, in Proposition \ref{prop:intbound_sg_noncentral}(iii) take 
$\sigma^2= \tilde{\omega}$, $b=(1+\delta)/(1-\delta)$, $a= 1/[\tilde{\omega} (1+\delta)]$,
and $h=(g n k_2)^{|\gamma|_0 - |\gamma^*|_0} [ p(\gamma^*)/p(\gamma) ]^2$, and the sub-Gaussian dimensions to be
$d_1=|\gamma'|_0 - |\gamma^*|_0$ and $d_2=|\gamma'|_0 - |\gamma|_0$.

Proposition \ref{prop:intbound_sg_noncentral}(iii) requires $a/b \leq 1 \Leftrightarrow (1-\delta)/[\omega \tau (1+\delta)^2]$, which holds under Assumption (B2), since $\delta$ is a constant taken arbitrarily close to 0.
Proposition \ref{prop:intbound_sg_noncentral}(iii) also requires that
\begin{align}
& \log h > \frac{q_0 b d_1}{a} \Leftrightarrow
(|\gamma|_0 - |\gamma^*|_0) \log(g n k_2) + 2 \log \left( \frac{p(\gamma^*)}{p(\gamma)} \right) 
 > \frac{q_0 \tilde{\omega} (1+\delta)^2 (|\gamma'|_0-|\gamma^*|_0)}{1-\delta}
\nonumber \\
\Leftrightarrow
&\frac{1}{2}\log(g n k_2) + \frac{1}{|\gamma|_0 - |\gamma^*|_0} \log \left( \frac{p(\gamma^*)}{p(\gamma)} \right) 
 > \frac{(1+2^{1/2})^2 \tilde{\omega} (1+\delta)^2 (|\gamma'|_0-|\gamma^*|_0)}{(1-\delta) (|\gamma|_0 - |\gamma^*|_0)},
\end{align}
which holds under Assumption (B3), since 
\begin{align}
 \frac{q_0 \tilde{\omega} (1+\delta)^2 \bar{q}}{1-\delta}
\geq \frac{q_0 \tilde{\omega} (1+\delta)^2 |\gamma|_0}{1-\delta}
\geq \frac{q_0 \tilde{\omega} (1+\delta)^2 (|\gamma'|_0-|\gamma^*|_0)}{(1-\delta)(|\gamma|_0 - |\gamma^*|_0)}.
\nonumber
\end{align}
where we used that $|\gamma'|_0 - |\gamma^*|_0 \leq |\gamma|_0 \leq \bar{q}$ and that $|\gamma|_0 - |\gamma^*|_0 \geq 1$.


Since the conditions to apply Proposition \ref{prop:intbound_sg_noncentral}(iii) are met, we obtain
\begin{align}
E_F(p(\gamma \mid \gamma)) \leq \frac{2 \max \left\{([e^{q_0}/d_1] \log(h^{\frac{a}{b a'}}))^{d_1/2}  , \log(h^{\frac{a}{2b a'}}) \right\}}{h^{\frac{a}{2b a'}}}
+ \frac{1}{2 h^{\frac{a}{2b(1+k_0)}}}
\nonumber
\end{align}
where $$
a'> 1 + \frac{q_0^{1/2}}{[q_0 + (a/b) \log \log (h^{1/d_1})]^{1/2}},
$$
and $k_0= [q_0b d_1/(a \log h)]^{1/2}$ may both be taken arbitrarily close to 1 under Assumption (B3).
Then, arguing as in \eqref{seq:modelppbound_overfitted2} gives that
\begin{align}
& E_F (p(\gamma \mid y)) \leq 2.5 h^{-\frac{ad_1(1-\epsilon)}{2}}
=2.5 \left( \frac{b}{(g n)^{\frac{1-\epsilon}{\omega \tau}}} \right)^{\frac{|\gamma|_0 - |\gamma^*|_0}{2}} \left( \frac{p(\gamma)}{p(\gamma^*)} \right)^{\frac{1-\epsilon}{\omega \tau}}
 \label{seq:modelppbound_largenonoverfitted}
\end{align}
for all $n \geq n_0$ and some fixed $n_0$, where $b= e^{2(1+2^{1/2})^2}/k_2^{(1-\epsilon)/[\omega \tau]}$ is a constant.

\subsection{Sum across non-overfitted models of size $|\gamma|_0 \leq |\gamma^*|_0$}

The strategy is to split the sum into models with dimension $\leq |\gamma^*|_0$, where recall that $\gamma^*$ is the optimal model, and those of dimension $|\gamma|_0 > |\gamma^*|_0$. That is,
\begin{align}
 \sum_{l=0}^{\bar{q}} \sum_{|\gamma|_0= l, \gamma \not\subset \gamma^*} E_F(p(\gamma \mid y))
= \sum_{l=0}^{|\gamma^*|_0} \sum_{|\gamma|_0= l, \gamma \not\subset \gamma^*} E_F(p(\gamma \mid y))
+ \sum_{l=|\gamma^*|_0+1}^{\bar{q}} \sum_{|\gamma|_0= l, \gamma \not\subset \gamma^*} E_F(p(\gamma \mid y)).
\label{seq:modelpp_sumnonoverfit}
\end{align}

Consider the first term in \eqref{seq:modelpp_sumnonoverfit}. From \eqref{seq:modelppbound_nonoverfitted3}, there is a fixed $n_0$ such that for every $n \geq n_0$ it holds that
\begin{align}
 \sum_{l=0}^{|\gamma^*|_0} \sum_{|\gamma|_0= l, \gamma \not\subset \gamma^*} E_F(p(\gamma \mid y)) \leq
\sum_{l=0}^{|\gamma^*|_0} \sum_{|\gamma|_0= l, \gamma \not\subset \gamma^*} e^{ -\frac{\lambda_\gamma (1-\epsilon)}{8\tilde{\omega}}}
+ 3.5 \left( \frac{ p(\gamma)/p(\gamma^*) }{(gnk_2)^{\frac{|\gamma|_0 - |\gamma^*|_0}{2}} e^{\frac{\lambda_\gamma}{2 \tilde{\omega}}}} \right)^{1-\epsilon}.
\label{seq:pp_smallnonoverfit}
\end{align}

Let $\underline{\lambda}= \min_{|\gamma|_0 \leq |\gamma^*|_0} \lambda_\gamma/\max\{|\gamma^*|_0-|\gamma|_0,1\}$, so that
$
 e^{-\lambda_\gamma} \leq e^{-\underline{\lambda} \max\{|\gamma^*|_0 - |\gamma|_0,1\}}.
$
Then \eqref{seq:pp_smallnonoverfit} is
\begin{align}
&\leq 
\sum_{l=0}^{|\gamma^*|_0} \sum_{|\gamma|_0= l, \gamma \not\subset \gamma^*} e^{ -\frac{\lambda_\gamma (1-\epsilon)}{8\tilde{\omega}}}
+ 3.5 \left( \frac{ p(\gamma)/p(\gamma^*) }{(gnk_2)^{\frac{l - |\gamma^*|_0}{2}} e^{\frac{\lambda_\gamma}{2 \tilde{\omega}}}} \right)^{1-\epsilon}
=
\sum_{|\gamma|_0= |\gamma^*|_0, \gamma \not\subset \gamma^*} e^{ -\frac{\underline{\lambda} (1-\epsilon)}{8\tilde{\omega}}} + 3.5 e^{-\frac{\underline{\lambda}(1-\epsilon)}{2 \tilde{\omega}}} +
\nonumber \\
& \sum_{l=0}^{|\gamma^*|_0-1} \sum_{|\gamma|_0= l, \gamma \not\subset \gamma^*} \left[ e^{ -\frac{\underline{\lambda} (1-\epsilon)}{8\tilde{\omega}}} \right]^{|\gamma^*|_0 - l}
+ 3.5 \left( \left[ \frac{ e^{\frac{\underline{\lambda}}{2 \tilde{\omega}}} }{q^c (gnk_2)^{\frac{1}{2}}} \right]^{1-\epsilon} \right)^{l - |\gamma^*|_0}
\left[ {q \choose |\gamma^*|_0 }/ {q \choose l} \right]^{1-\epsilon}
\nonumber \\
&={q \choose |\gamma^*|_0}
4.5 e^{-\frac{\underline{\lambda}(1-\epsilon)}{2 \tilde{\omega}}} +
\sum_{l=0}^{|\gamma^*|_0-1} 
3.5 \left( \left[ \frac{ e^{\frac{\underline{\lambda}}{2 \tilde{\omega}}} }{q^c (gnk_2)^{\frac{1}{2}}} \right]^{1-\epsilon} \right)^{l - |\gamma^*|_0}
{q \choose |\gamma^*|_0 }^{1-\epsilon}  {q \choose l}^{\epsilon}
\nonumber \\
& \leq 4.5 e^{-\frac{\underline{\lambda}(1-\epsilon)}{2 \tilde{\omega}} + |\gamma^*|_0\log q} +
3.5 e^{|\gamma^*|_0 \log q} \sum_{l=0}^{|\gamma^*|_0-1} 
 \left( \left[ \frac{ e^{\frac{\underline{\lambda}}{2 \tilde{\omega}}} }{q^c (gnk_2)^{\frac{1}{2}}} \right]^{1-\epsilon} \right)^{l - |\gamma^*|_0}
q^{l \epsilon}=
\nonumber
\end{align}
where we used the expression for $p(\gamma)/p(\gamma^*)$ in \eqref{seq:ratio_priormodelprob},
and that there are ${q \choose l} \leq q^l$ models of size $l$.
Rearranging terms in the latter expression gives
\begin{align}
& 4.5 e^{-\frac{\underline{\lambda}(1-\epsilon)}{2 \tilde{\omega}} + |\gamma^*|_0\log q} +
3.5 e^{|\gamma^*|_0 \log q}
\left[ \frac{q^{c(1-\epsilon)} (gnk_2)^{(1-\epsilon)/2}}{ e^{\frac{\underline{\lambda}(1-\epsilon)}{2 \tilde{\omega}}} } \right]^{|\gamma^*|_0}
 \sum_{l=0}^{|\gamma^*|_0-1} 
 \left( \left[ \frac{ e^{\frac{\underline{\lambda}(1-\epsilon)}{2 \tilde{\omega}}} }{q^{c(1-\epsilon)-\epsilon} (gnk_2)^{\frac{1-\epsilon}{2}}} \right] \right)^{l}
\nonumber \\
&=
4.5 e^{-\frac{\underline{\lambda}(1-\epsilon)}{2 \tilde{\omega}} + |\gamma^*|_0\log q} +
3.5 e^{|\gamma^*|_0 \log q}
\left[ \frac{q^{c(1-\epsilon)} (gnk_2)^{(1-\epsilon)/2}}{ e^{\frac{\underline{\lambda}(1-\epsilon)}{2 \tilde{\omega}}} } \right]^{|\gamma^*|_0}
\frac{1 - \left[ \frac{ e^{\frac{\underline{\lambda}(1-\epsilon)}{2 \tilde{\omega}}} }{q^{c(1-\epsilon)-\epsilon} (gnk_2)^{\frac{1-\epsilon}{2}}} \right]^{|\gamma^*|_0 -1} }{1 - \left[ \frac{ e^{\frac{\underline{\lambda}(1-\epsilon)}{2 \tilde{\omega}}} }{q^{c(1-\epsilon)-\epsilon} (gnk_2)^{\frac{1-\epsilon}{2}}} \right]}
,
\label{seq:pp_smallnonoverfit2}
\end{align}
the right-hand side following from the geometric series.

To find a simpler asymptotic expression for \eqref{seq:pp_smallnonoverfit2}, note that
\begin{align}
\lim_{n \rightarrow \infty} \frac{ e^{\frac{\underline{\lambda}(1-\epsilon)}{2 \tilde{\omega}}} }{q^{c(1-\epsilon)-\epsilon} (gnk_2)^{\frac{1-\epsilon}{2}}}
\Leftrightarrow
\lim_{n \rightarrow \infty}
\frac{\underline{\lambda}(1-\epsilon)}{2 \tilde{\omega}} - [c(1-\epsilon)-\epsilon] \log(q) - \frac{1-\epsilon}{2} \log(gnk_2)= \infty,
\nonumber
\end{align}
which holds under Assumption (B6).
Hence for sufficiently large $n$ we have that \eqref{seq:pp_smallnonoverfit2} is
\begin{align}
&\leq 4.5 e^{|\gamma^*|_0 \log q}
\left( e^{-\frac{\underline{\lambda}(1-\epsilon)}{2 \tilde{\omega}}}
+ \left[ \frac{q^{c(1-\epsilon)} (gnk_2)^{(1-\epsilon)/2}}{ e^{\frac{\underline{\lambda}(1-\epsilon)}{2 \tilde{\omega}}} } \right]^{|\gamma^*|_0}
\left[ \frac{ e^{\frac{\underline{\lambda}(1-\epsilon)}{2 \tilde{\omega}}} }{q^{c(1-\epsilon)-\epsilon} (gnk_2)^{\frac{1-\epsilon}{2}}} \right]^{|\gamma^*|_0-1}
 \right)
\nonumber \\
&\leq 4.5 e^{|\gamma^*|_0 \log q}
\left( e^{-\frac{\underline{\lambda}(1-\epsilon)}{2 \tilde{\omega}}}
+ 
q^{\epsilon (|\gamma^*|_0-1)}
\left[ \frac{q^{c(1-\epsilon)} (gnk_2)^{(1-\epsilon)/2}}{ e^{\frac{\underline{\lambda}(1-\epsilon)}{2 \tilde{\omega}}} } \right] \right)
\nonumber \\
&= 4.5 e^{-\frac{\underline{\lambda}(1-\epsilon)}{2 \tilde{\omega}} + |\gamma^*|_0 \log q}
\left( 1
+ 
\exp \left\{ [\epsilon (|\gamma^*|_0-1) + c(1-\epsilon)] \log(q) + \frac{(1-\epsilon)}{2} \log(gnk_2) \right\}
\right)
\nonumber \\
&\leq 4.5 e^{-\frac{\underline{\lambda}(1-\epsilon)}{2 \tilde{\omega}} + |\gamma^*|_0 \log q}
\left( 1
+ 
\exp \left\{ [\epsilon |\gamma^*|_0 + c] \log(q) + \frac{1-\epsilon}{2} \log(gnk_2) \right\}
\right)
\nonumber \\
&
< 9 \exp \left\{  -\frac{\underline{\lambda}(1-\epsilon)}{2 \tilde{\omega}} + [|\gamma^*|_0(1+\epsilon) + c] \log q + \frac{1-\epsilon}{2} \log(gnk_2) \right\},
\nonumber
\end{align}
where for large enough $n$ we may upper-bound $(1-\epsilon)\log(gnk_2)$ by $\log(gn)$, as we wished to prove.

\subsection{Sum across non-overfitted models of size $|\gamma|_0 > |\gamma^*|_0$. Sparsity-based bound}

The strategy is to use the bound for $E_F(p(\gamma \mid y))$ given in \eqref{seq:modelppbound_largenonoverfitted}
and to proceed analogously to \eqref{seq:pp_spurious_bound} and \eqref{seq:pp_spurious_bound2}, where plugging in the expression of prior model probabilities in \eqref{seq:ratio_priormodelprob}, gives that
\begin{align}
& \sum_{l=|\gamma^*|_0+1}^{\bar{q}} \sum_{|\gamma|_0=l, \gamma \not\supset \gamma^*} E_F( p(\gamma \mid y)) \leq
 \sum_{l=|\gamma^*|_0+1}^{\bar{q}} 
\left( \frac{b^{1/2}}{q^{cr} (g n)^{r/2}} \right)^{l - |\gamma^*|_0} 
{l \choose |\gamma^*|_0 }^{r} {q - |\gamma^*|_0 \choose l - |\gamma^*|_0}^{-r}
\sum_{|\gamma|_0=l, \gamma \not\supset \gamma^*} 1
\nonumber \\
&= \sum_{l=|\gamma^*|_0+1}^{\bar{q}}  \left( \frac{b^{1/2}}{q^{cr} (g n)^{r/2}} \right)^{l - |\gamma^*|_0}  {l \choose |\gamma^*|_0 }^{r} {q - |\gamma^*|_0 \choose l - |\gamma^*|_0}^{-r} q^l
\label{seq:sum_largenonoverfit_sparsebound}
\end{align}
for all $n \geq n_0$ and some fixed $n_0$, where $r=(1-\epsilon)/[\omega \tau] < 1$ (from Assumption (B2)),
$b= e^{2(1+2^{1/2})^2}/k_2^{(1-\epsilon)/[\omega \tau]}$ is a constant
and the right-hand side follows from noting that there are
${q \choose l} - {q - |\gamma^*|_0 \choose l - |\gamma^*|_0} \leq q^l$ non-overfitted models of size $l$.

Using that ${l \choose |\gamma^*|_0}^r \leq {l \choose |\gamma^*|_0}$ for $r<1$ and that ${x \choose z} \geq (x/z)^z$ for all $(x,z)$, which implies that
${q - |\gamma^*|_0 \choose l - |\gamma^*|_0} \geq ([q - |\gamma^*|_0]/[l-|\gamma^*|_0])^{l - |\gamma^*|_0}$, we obtain that \eqref{seq:sum_largenonoverfit_sparsebound} is
\begin{align}
&\leq q^{|\gamma^*|_0}
\sum_{l=|\gamma^*|_0+1}^{\bar{q}}  \left( \frac{b^{1/2}}{q^{cr-1} (g n)^{r/2}} \right)^{l - |\gamma^*|_0}  {l \choose |\gamma^*|_0 }
\left( \frac{l-|\gamma^*|_0}{q - |\gamma^*|_0} \right)^{r(l-|\gamma^*|_0)}
\nonumber \\
&\leq q^{|\gamma^*|_0}
\sum_{l=|\gamma^*|_0+1}^{\bar{q}}  \left( \frac{b^{1/2} (\bar{q}-|\gamma^*|_0)^r}{q^{cr-1} (q - |\gamma^*|_0)^r (g n)^{r/2}} \right)^{l - |\gamma^*|_0}  {l \choose |\gamma^*|_0 }
\label{seq:sum_largenonoverfit_sparsebound2}
\end{align}

Finally, using Lemma \ref{slem:binomial_ogf} gives that the right-hand side of \eqref{seq:sum_largenonoverfit_sparsebound2} is asymptotically equal to
\begin{align}
\frac{q^{|\gamma^*|_0} (|\gamma^*|_0+1) b^{1/2} (\bar{q}-|\gamma^*|_0)^r}{q^{cr-1} (q - |\gamma^*|_0)^r (g n)^{r/2}}
\leq 
\frac{(|\gamma^*|_0+1) b^{1/2} \bar{q}^r}{q^{cr-1-|\gamma^*|_0} (q - |\gamma^*|_0)^r (g n)^{r/2}},
\nonumber
\end{align}
where note that under Assumption (B7) the right-hand side converges to 0 as $n \rightarrow \infty$.

\section{Alternative to Theorem \ref{thm:pp}}
\label{ssec:alt_thm_pp}

Theorem \ref{thm:pp_alternative} state a result that provides an alternative to Theorem \ref{thm:pp}(iii) where one obtains faster model selection rates, under Assumptions (B5') and (B7') that overall are milder than (B5) and (B7) used in Theorem \ref{thm:pp}(iii).

More precisely, (B5') is a slightly stronger version of (B5). Similar to (B5) it requires that the non-centrality parameter $\lambda_\gamma$ is large enough relative to the model size. The difference is that (B5') requires the condition to hold for all models, whereas (B5) required it only for models of size less than the optimal model ($|\gamma|_0 \leq |\gamma^*|_0$). This requirement is however not overly stringent, since for $|\gamma|_0 > |\gamma^*|_0$ the term
$
 t= (|\gamma|_0 - |\gamma^*|_0) \log(g n k_2) + 2 \log \left( p(\gamma^*)/p(\gamma) \right)
$
grows with $n$.

Assumption (B7') introduces a non-centrality parameter $\bar{\lambda}$ that lower-bounds the decrease in the explained sum of squares for each truly active parameter. Under betamin and restricted eigenvalue conditions, $\bar{\lambda}$ is proportional to $n$ times the smallest square entry in the optimal coefficients $\eta_{\gamma^*}^*$ (see \cite{rossell:2022}, Sections 2.2 and 5.4).

\begin{enumerate}[leftmargin=*,label=(B\arabic*')]
\setcounter{enumi}{4}
\item Condition (B5) holds for any non over-fitted model $\gamma \not\subset \gamma^*$ of size $|\gamma|_0 \leq \bar{q}$,
\end{enumerate}

\begin{enumerate}[leftmargin=*,label=(B\arabic*')]
\setcounter{enumi}{6}
\item Let $A_{jl}$ be the set of models of size $|\gamma|_0=l$ that select $j \leq |\gamma^*|_0-1$ out of the $|\gamma^*|_0$ truly active parameters, and $l-j$ inactive parameters, and
$$\bar{\lambda}= \min_{l=\{|\gamma^*|_0+1,\ldots,\bar{q}\}} \min_{\gamma \in A_{jl}} \lambda_\gamma/(|\gamma^*|_0-j).$$
Assume that, for some fixed $\epsilon >0$,
$$
\lim_{n \rightarrow \infty} \frac{\bar{\lambda}(1-\epsilon)}{ 8\omega\tau} - (\bar{q}+1) \log(q) - |\gamma^*|_0 \log |\gamma^*|_0 = \infty.
$$
\end{enumerate}

We next state the theorem and discuss its implications.

\begin{thm}
Assume Conditions (A1)-(A4), (B1), (B2), (B5'), and (B7'). 
Let $S_2= \{\gamma: |\gamma|_0 > |\gamma^*|_0, \gamma \not\subset \gamma^*\}$
and $\bar{\lambda}$ be the signal strength parameter defined in (B7'). 
Then there is a fixed $n_0$ such that
\begin{align}
E_F(P(S_2 \mid y)) \leq
 4.5 \exp \left\{  -\frac{\bar{\lambda}(1-\epsilon)}{ 8\omega\tau} + (\bar{q}-|\gamma^*|_0+1) \log(q - |\gamma^*|_0) + (|\gamma^*|_0-1) \log |\gamma^*|_0 \right\}
\nonumber
\end{align}
for all $n \geq n_0$ and a constant $\epsilon >0$ that may be taken arbitrarily close to 0.
\label{thm:pp_alternative}
\end{thm}

Theorem \ref{thm:pp_alternative} says that large non-overfitted models are discarded at an exponential rate that is essentially upper-bounded by $\bar{\lambda}/[8 \omega \tau] + \bar{q} \log q$, where $\bar{\lambda}$ can be thought of as proportional to $n$ under betamin and restricted eigenvalue conditions.
These models receive vanishing posterior probability as long as $\bar{q} \log q$ grows at a slower rate (since the term $(|\gamma^*|_0-1) \log |\gamma^*|_0$ is even smaller).

\subsection{Proof}

Consider the second term in \eqref{seq:modelpp_sumnonoverfit}, corresponding to models of dimension $|\gamma|_0 > |\gamma^*|_0$.
The strategy is to sum posterior model probabilities according to the model dimension $|\gamma|_0$ and the number of truly active parameters that the model is missing out of the $|\gamma^*|_0$ parameters in the optimal $\gamma^*$.

As defined in (B7'), let $A_{jl}$ be the set of models of size $|\gamma|_0=l$ that select $j \leq |\gamma^*|_0-1$ out of the $|\gamma^*|_0$ truly active parameters, and $l-j$ inactive parameters.
Since $S_2=\cup_{j=0}^{|\gamma^*|_0-1} A_j$ is the whole set of non-overfitted models of size $l$, we obtain that
\begin{align}
&E_F(p(S_2 \mid y)) = \sum_{l=|\gamma^*|_0+1}^{\bar{q}} \sum_{j=0}^{|\gamma^*|_0-1} E_F(p(A_{lj} \mid y)) \leq
\nonumber \\
& \sum_{l=|\gamma^*|_0+1}^{\bar{q}} \sum_{j=0}^{|\gamma^*|_0-1}
{|\gamma^*|_0 \choose j} {q - |\gamma^*|_0 \choose l-j}
\left[ e^{ -\frac{(|\gamma^*|_0-j)\bar{\lambda}(1-\epsilon)}{ 8\tilde{\omega}}}
+ 3.5 \left( \frac{ p(\gamma)/p(\gamma^*) }{(gnk_2)^{\frac{l - |\gamma^*|_0}{2}} e^{(|\gamma^*|_0-j)\frac{\bar{\lambda}}{2 \tilde{\omega}}}} \right)^{1-\epsilon} \right]
\label{seq:pp_largenonoverfit}
\end{align}
for all $n \geq n_0$ and some fixed $n_0$,
where we used \eqref{seq:modelppbound_nonoverfitted3},
that there are ${|\gamma^*|_0 \choose j}{q - |\gamma^*|_0 \choose l -j }$ models in the set $A_{jl}$,
and that by Assumption (B7') all $\gamma \in A_{jl}$ satisfy that $\bar{\lambda} \leq \lambda_\gamma/(|\gamma^*|_0-j)$.

Given that $e^{-\bar{\lambda}/8} > e^{\bar{\lambda}/2}$ and that, since $l > |\gamma^*|_0$, we have that $(gnk_2)^{-(l-|\gamma^*|_0)} p(\gamma)/p(\gamma^*) < 1$, 
\eqref{seq:pp_largenonoverfit} is upper-bounded by
\begin{align}
& \sum_{j=0}^{|\gamma^*|_0-1} {|\gamma^*|_0 \choose j} \sum_{l=|\gamma^*|_0+1}^{\bar{q}} 
 {q - |\gamma^*|_0 \choose l-j}
4.5 e^{ -\frac{(|\gamma^*|_0-j)\bar{\lambda}(1-\epsilon)}{ 8\tilde{\omega}}}
\nonumber \\
&\leq 
4.5 e^{ -\frac{|\gamma^*|_0\bar{\lambda}(1-\epsilon)}{ 8\tilde{\omega}}}
 \sum_{j=0}^{|\gamma^*|_0-1} \left( \frac{|\gamma^*|_0 e^{ \frac{\bar{\lambda}(1-\epsilon)}{ 8\tilde{\omega}}}}{q - |\gamma^*|_0} \right)^j  
\sum_{l=|\gamma^*|_0+1}^{\bar{q}} (q - |\gamma^*|_0)^{l}
\label{seq:pp_largenonoverfit2}
\end{align}
where to obtain the right-hand side we used that ${q - |\gamma^*|_0 \choose l - j} \leq (q-|\gamma^*|_0)^{l-j}$ and ${|\gamma^*|_0 \choose j} \leq |\gamma^*|_0^j$ and rearranged terms.

We now use the geometric series to obtain the inner summation in \eqref{seq:pp_largenonoverfit2}, which gives that \eqref{seq:pp_largenonoverfit2} is
\begin{align}
&=4.5 e^{ -\frac{|\gamma^*|_0\bar{\lambda}(1-\epsilon)}{ 8\tilde{\omega}}}
 \sum_{j=0}^{|\gamma^*|_0-1} \left( \frac{|\gamma^*|_0 e^{ \frac{\bar{\lambda}(1-\epsilon)}{ 8\tilde{\omega}}}}{q - |\gamma^*|_0} \right)^j  
\frac{(q-|\gamma^*|_0)^{\bar{q}+1} - (q-|\gamma^*|_0)^{|\gamma^*|_0+1}}{1 - (q-|\gamma^*|_0)}
\nonumber \\
&=4.5 e^{ -\frac{|\gamma^*|_0\bar{\lambda}(1-\epsilon)}{ 8\tilde{\omega}}}
(q-|\gamma^*|_0)^{\bar{q}}
 \sum_{j=0}^{|\gamma^*|_0-1} \left( \frac{|\gamma^*|_0 e^{ \frac{\bar{\lambda}(1-\epsilon)}{ 8\tilde{\omega}}}}{q - |\gamma^*|_0} \right)^j  
\label{seq:pp_largenonoverfit3}
\end{align}
where in the right-hand side we used that $(z^{\bar{q}+1} - z^{|\gamma^*|_0-1})/(1-z) \leq z^{\bar{q}}$ for all $z>0$.
Using again the geometric series gives that \eqref{seq:pp_largenonoverfit3} is
\begin{align}
&= 4.5 e^{ -\frac{|\gamma^*|_0\bar{\lambda}(1-\epsilon)}{ 8\tilde{\omega}}}
(q-|\gamma^*|_0)^{\bar{q}}
\left( \frac{1 - \left[ \frac{|\gamma^*|_0 e^{ \frac{\bar{\lambda}(1-\epsilon)}{ 8\tilde{\omega}}}}{q - |\gamma^*|_0} \right]^{|\gamma^*|_0}}{1 - \frac{|\gamma^*|_0 e^{ \frac{\bar{\lambda}(1-\epsilon)}{ 8\tilde{\omega}}}}{q - |\gamma^*|_0}} \right)
\nonumber \\
\asymp
&4.5 e^{ -\frac{|\gamma^*|_0\bar{\lambda}(1-\epsilon)}{ 8\tilde{\omega}}}
(q-|\gamma^*|_0)^{\bar{q}}
\left[ \frac{|\gamma^*|_0 e^{ \frac{\bar{\lambda}(1-\epsilon)}{ 8\tilde{\omega}}}}{q - |\gamma^*|_0} \right]^{|\gamma^*|_0 -1}
=  4.5 e^{ -\frac{\bar{\lambda}(1-\epsilon)}{ 8\tilde{\omega}}} (q - |\gamma^*|_0)^{\bar{q}-|\gamma^*|_0+1}  |\gamma^*|_0^{|\gamma^*|_0-1}
\nonumber \\
&=4.5 \exp \left\{  -\frac{\bar{\lambda}(1-\epsilon)}{ 8\tilde{\omega}} + (\bar{q}-|\gamma^*|_0+1) \log(q - |\gamma^*|_0) + (|\gamma^*|_0-1) \log |\gamma^*|_0 \right\},
\nonumber
\end{align}
as we wished to prove.

\section{Supplementary results}
\label{ssec:suppl_results}

\begin{figure}
\begin{center}
\begin{tabular}{cc}
\includegraphics[width=0.48\textwidth, height=0.42\textwidth]{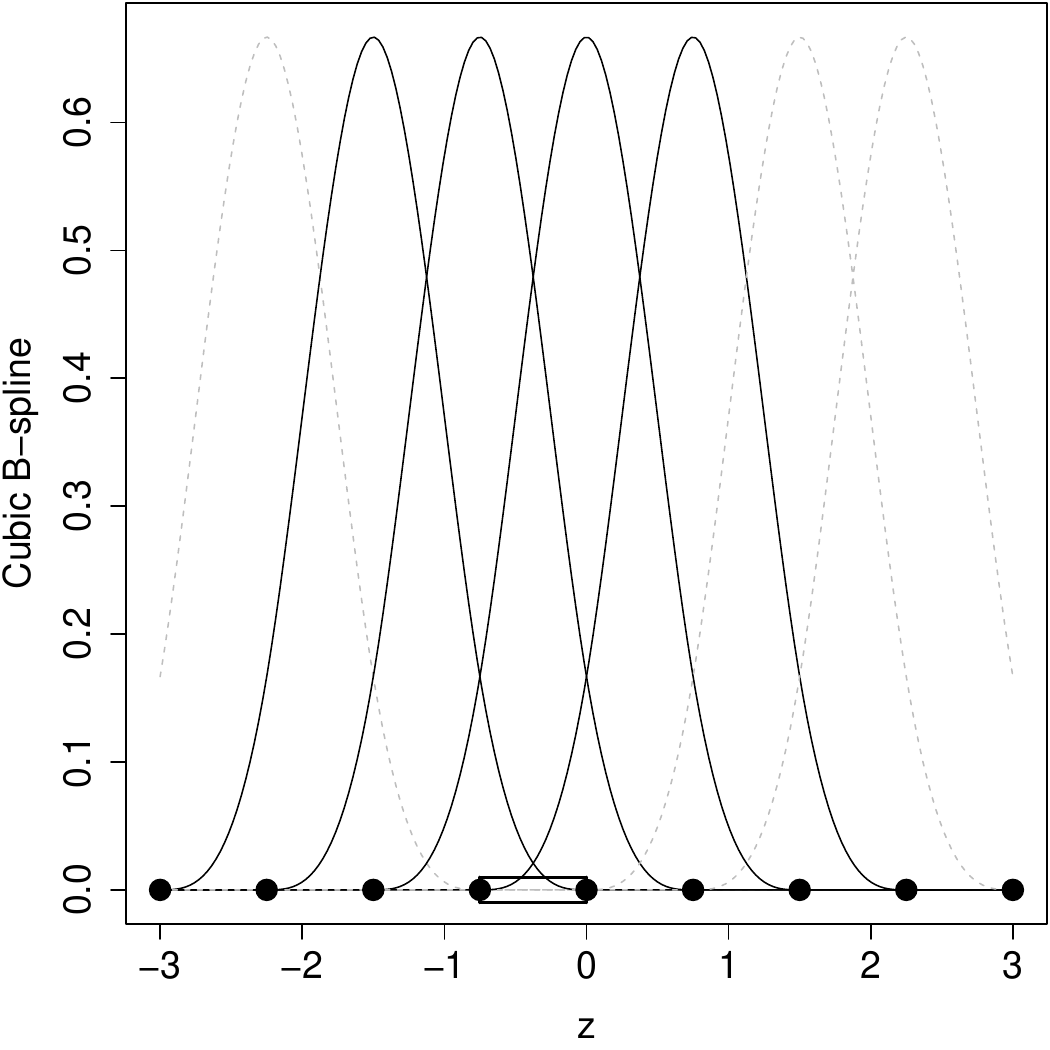} &
\includegraphics[width=0.48\textwidth, height=0.42\textwidth]{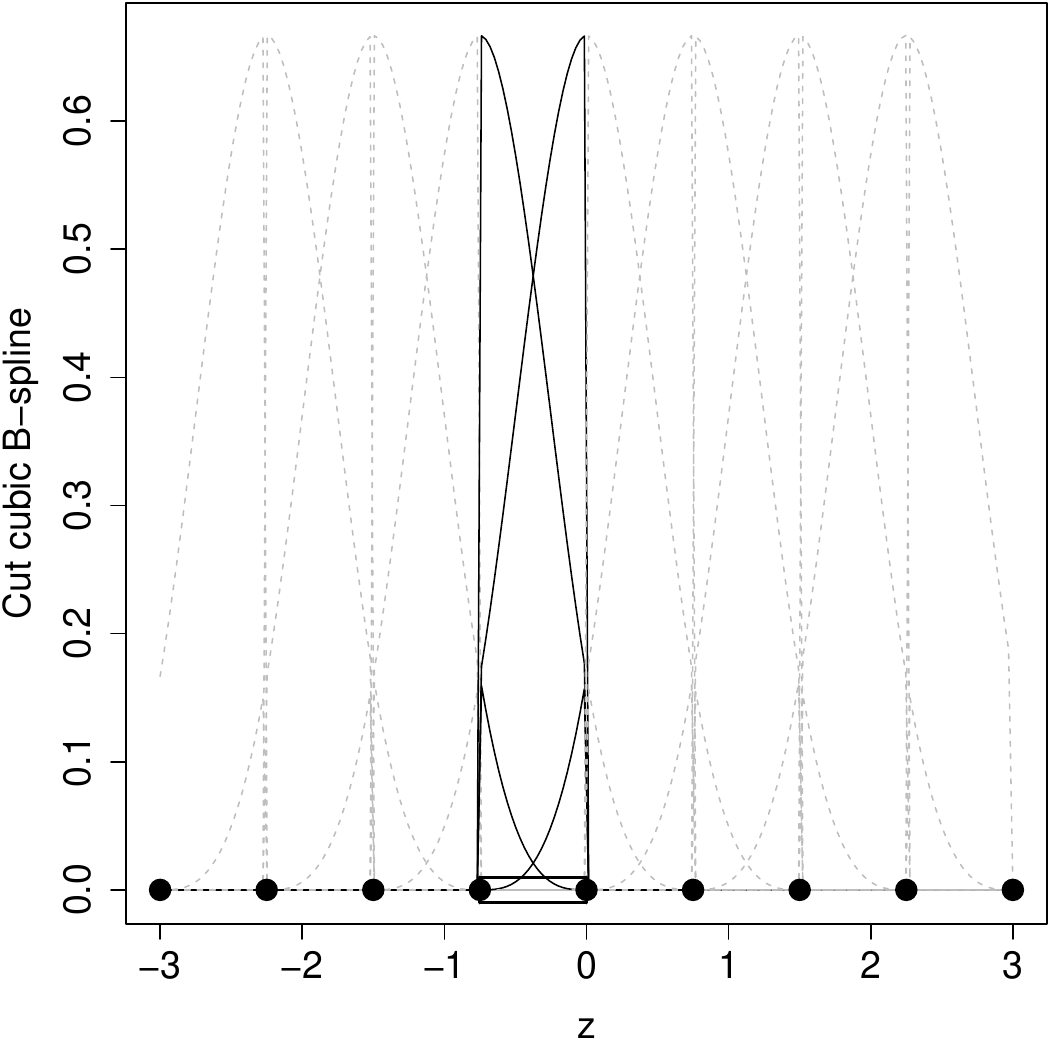} \\
\end{tabular}
\end{center}
\caption{Cubic B-splines (left) and cut cubic B-splines (right) for the simulated illustration in Figure \ref{fig:splinefit}.
}
\label{fig:splinefit_suppl}
\end{figure}

\subsection{Extension of the example in Figure \ref{fig:splinefit}}

\begin{figure}
\begin{center}
\begin{tabular}{cc}
\multicolumn{2}{c}{12 knots for baseline, 8 knots for local tests} \\
$n=1000$ & $n=2000$ \\
\includegraphics[width=0.48\textwidth, height=0.38\textwidth]{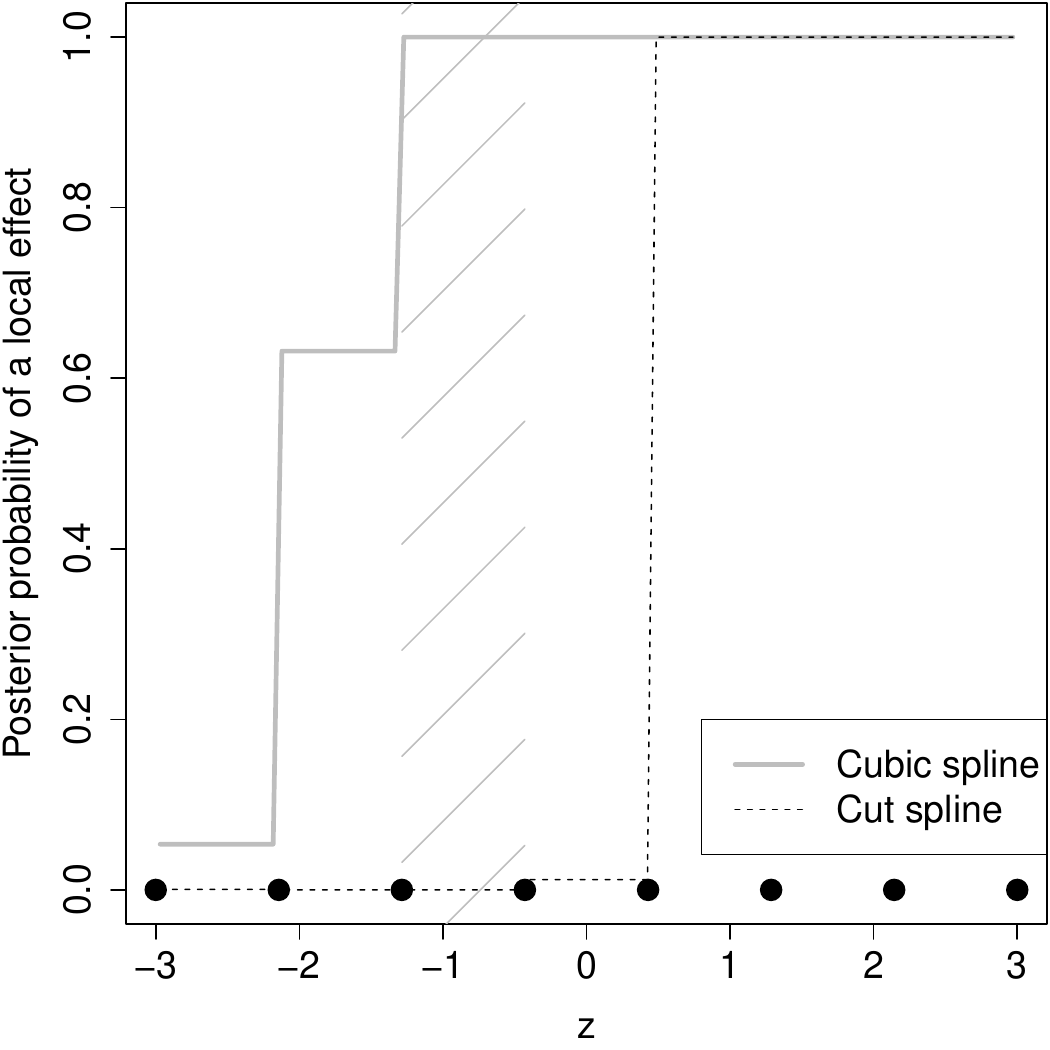} &
\includegraphics[width=0.48\textwidth, height=0.38\textwidth]{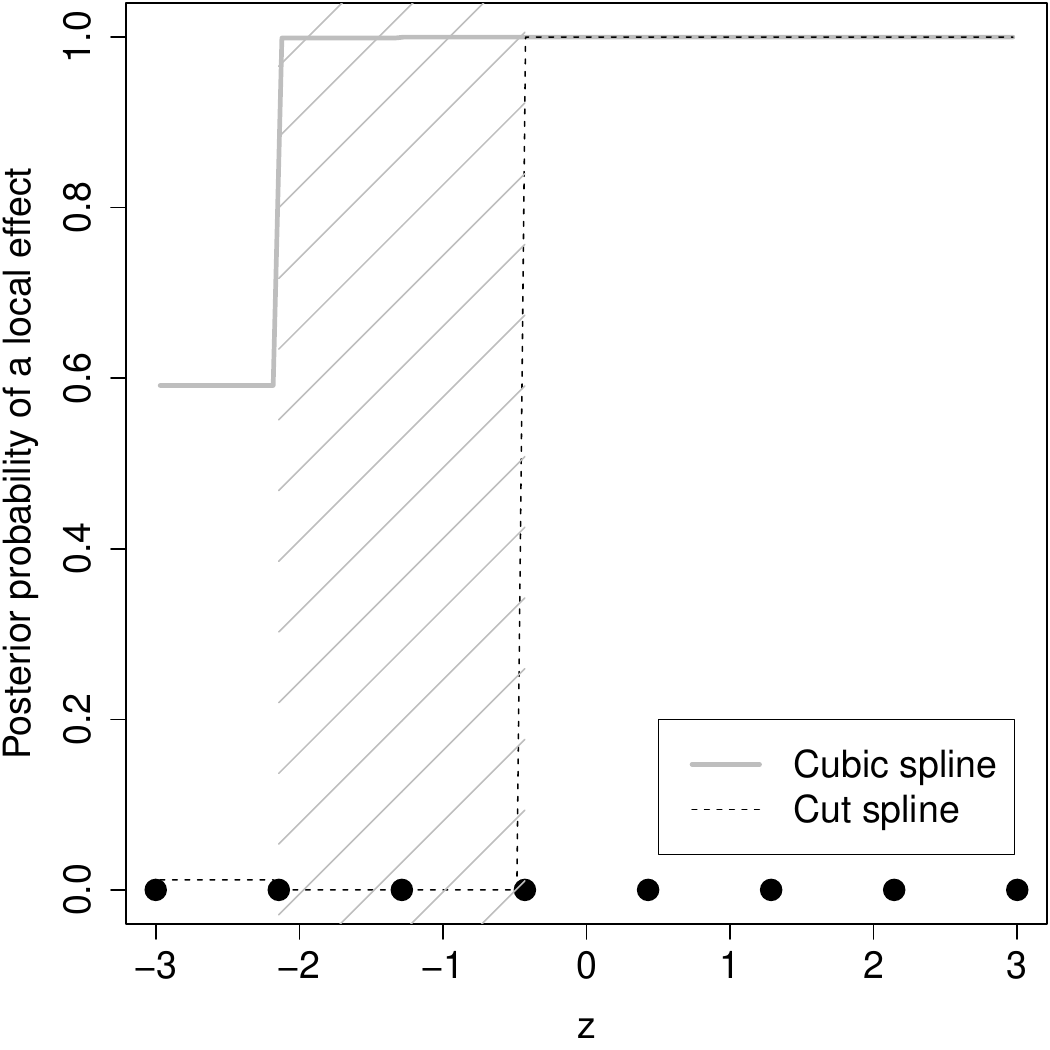} \\
\multicolumn{2}{c}{20 knots for baseline, 15 knots for local tests} \\
$n=1000$ & $n=2000$ \\
\includegraphics[width=0.48\textwidth, height=0.38\textwidth]{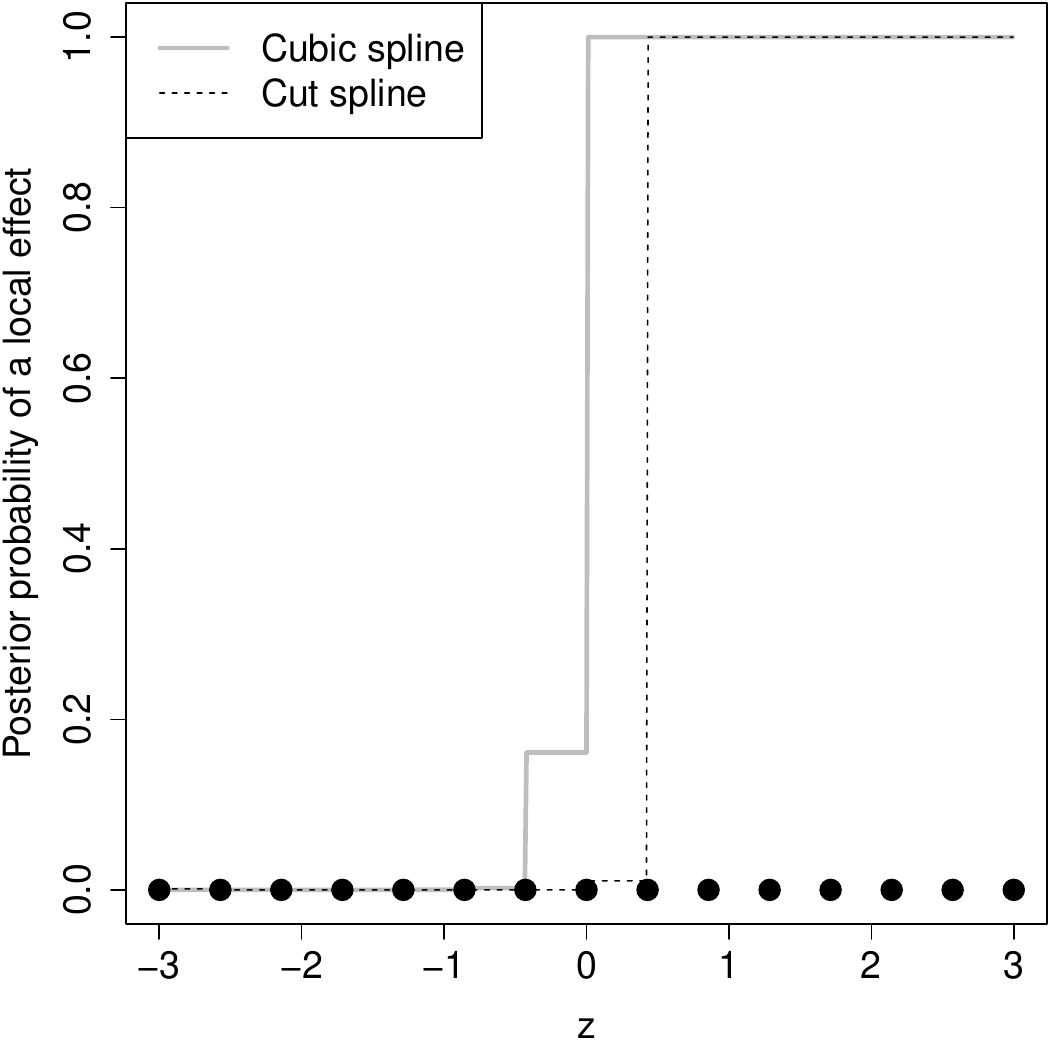} &
\includegraphics[width=0.48\textwidth, height=0.38\textwidth]{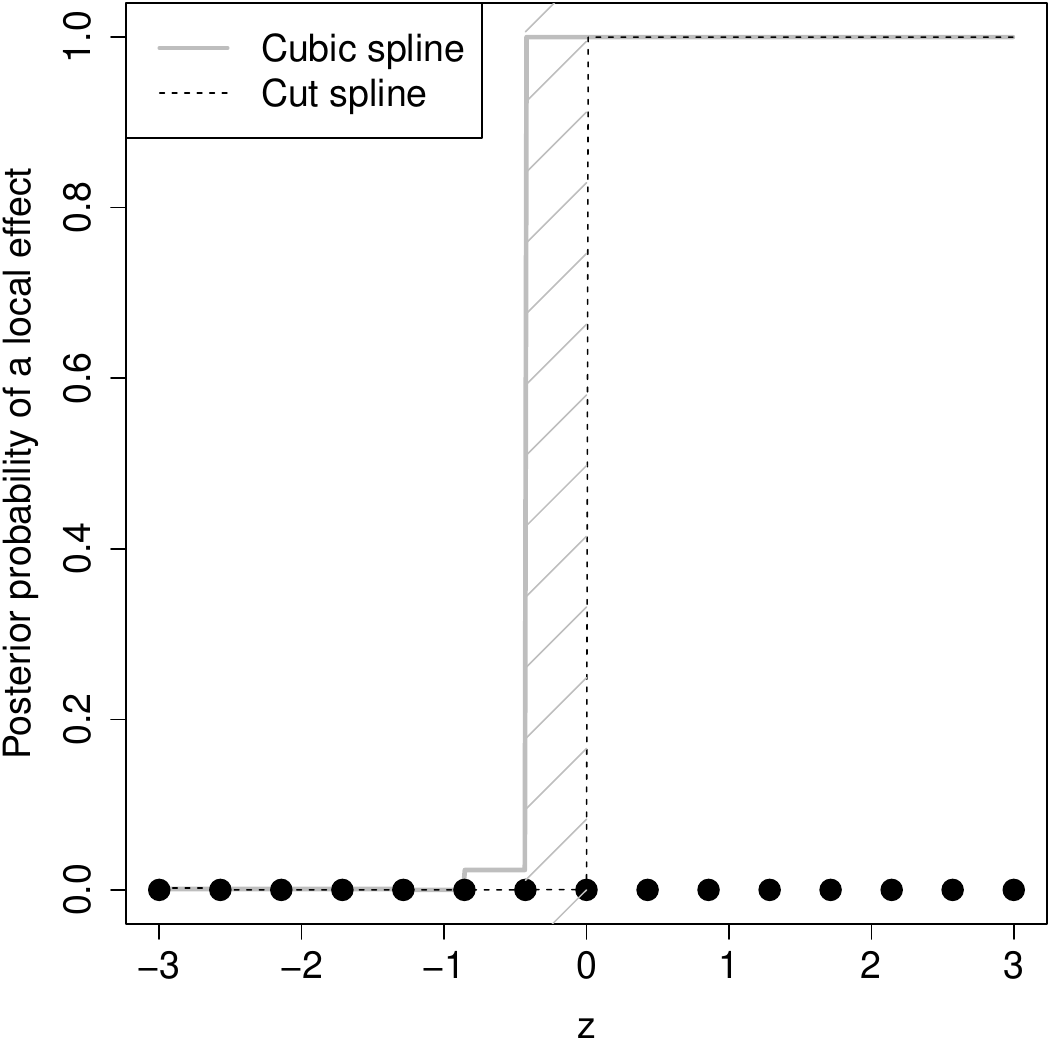} \\
\multicolumn{2}{c}{30 knots for baseline, 30 knots for local tests} \\
$n=1000$ & $n=2000$ \\
\includegraphics[width=0.48\textwidth, height=0.38\textwidth]{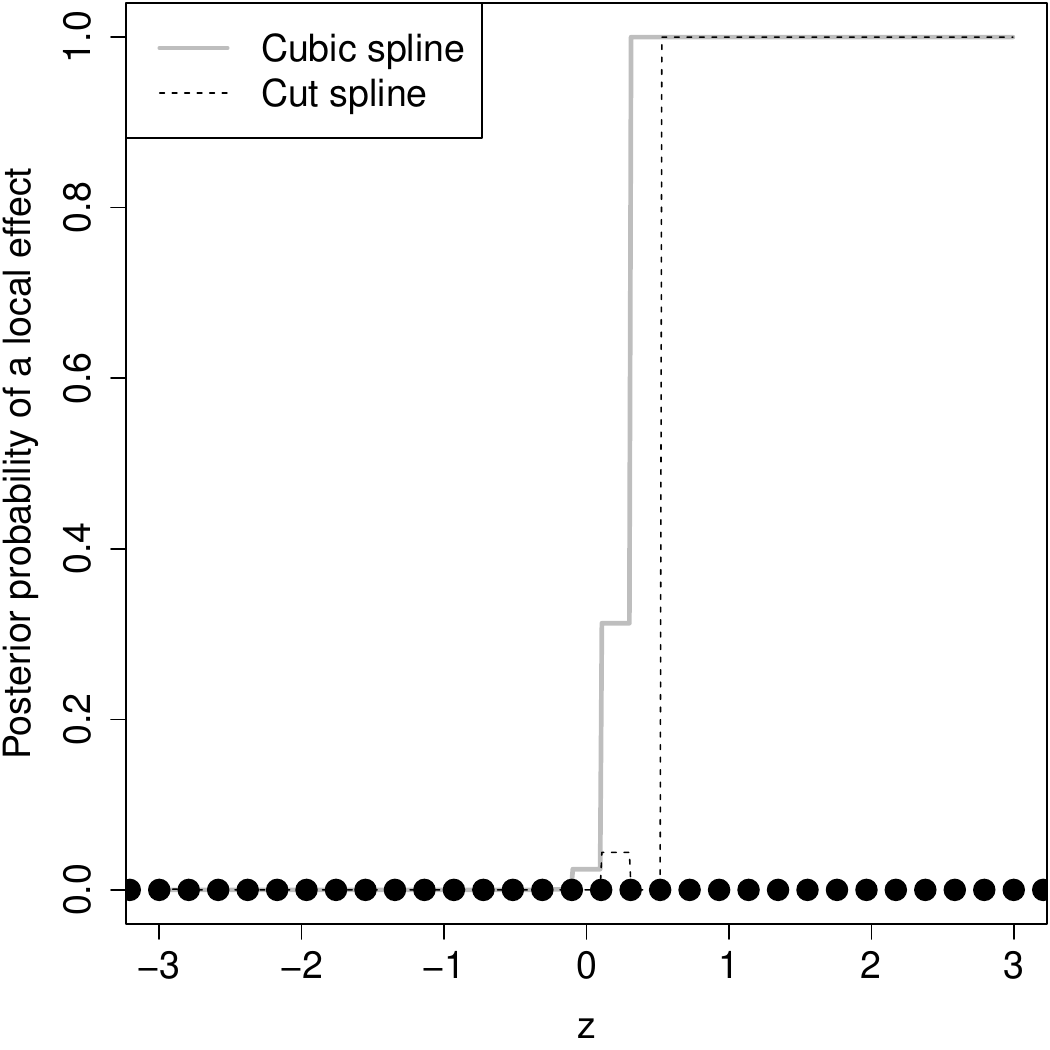} &
\includegraphics[width=0.48\textwidth, height=0.38\textwidth]{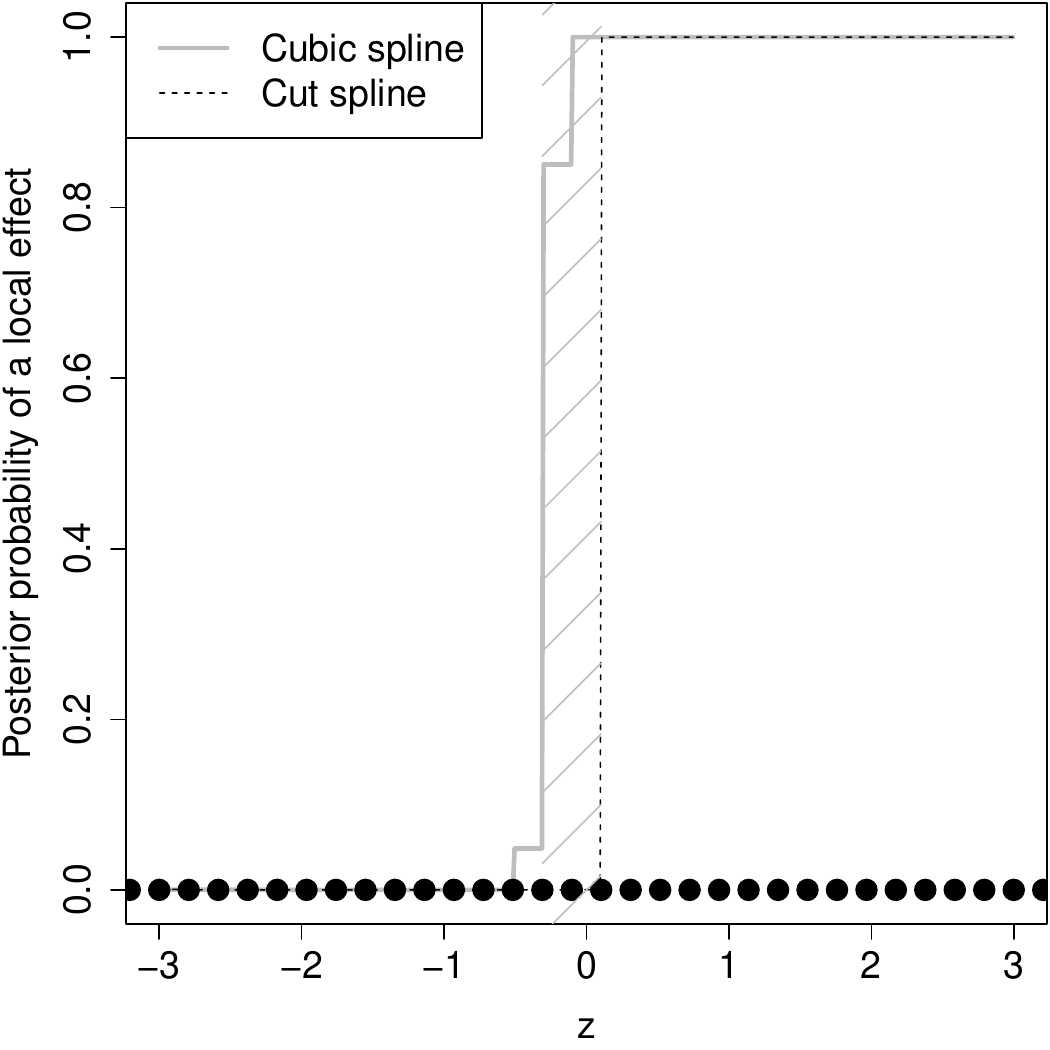} \\
\end{tabular}
\end{center}
\caption{Simulated illustration. Posterior probability of local group differences for cubic / cut cubic B-splines varying the knots and $n$. Shaded regions suffer from false positives}
\label{fig:splinefit_varyingknots}
\end{figure}

We extend the top panels in Figure \ref{fig:splinefit}, where one considers different number of knots and sample sizes.
Figure \ref{fig:splinefit_varyingknots} shows posterior probabilities for the presence of local group differences for the same examples as in Figure \ref{fig:splinefit}, but here we vary the number of knots and sample sizes.
The shaded areas indicate regions where there is false positive inflation, i.e. the posterior probability of a local covariate effect in that region is large, but there's truly no effect.
Overall the same phenomenon persists, for larger $n$ the standard B-splines leads to falsely rejecting the null in a neighborhood of $z<0$, whereas the cut B-splines avoid the false positive issue. 
The bottom left panel illustrates how, if the number of knots is too large, the power to detect local differences decreases, in this case for $z \in [0,0.5]$.

\subsection{Simulation with independent errors}
\label{ssec:simulation_iid}

We outline supplementary results for the simulation study presented in Section \ref{sec:simulation_iid} of the main manuscript.
Figure \ref{fig:simiid} shows the posterior probability (average across 100 simulations) for the presence of a covariate effect as a function of $z$.
Covariate 1 is truly active for $z>0$ and inactive for $z \leq 0$, covariates 2-10 are truly inactive at any $z \in [-3,3]$.

Table \ref{tab:simiid_extra} provide further results regarding our simulation with independent errors, related to using alternative methods to those described in the main paper.
It shows the average proportion of rejected local null hypotheses, separately for covariate 1 (which is truly active for $z>0$) and the other covariates.
Specifically, it considers a Benjamini-Hochberg P-value adjustment from a least-squares fit, a fused LASSO fit based on our degree 0 cut orthogonal basis, standard (uncut) cubic splines (which, as opposed to our cut orthogonal splines, runs into false positive issues) and a generalized additive model in mgcv with twice the default knots (24 knots).

Benjamini-Hochberg adjusted P-values were too conservative, exhibiting very low power for $n=100$, whereas the fused LASSO was overly liberal, with type I errors that were above 0.1 even when $n=1000$.
The uncut cubic splines and GAM resulted in an inflated type I error for covariate 1, for $z \in (0,1]$.

Table \ref{tab:simiid_mse} displays the root mean squared error for the various considered methods.

\begin{table}
\begin{center}
\begin{tabular}{ccccccccc} 
\multicolumn{9}{c}{$n=100$} \\ 
& \multicolumn{4}{c}{Covariate 1} & \multicolumn{4}{c}{Covariates 2-10} \\ 
 Region          & BH  &  Fused       & Cubic       & GAM          & BH    &  Fused        & Cubic & GAM \\ 
                 &     &  LASSO       & uncut       & 24 knots   &       &  LASSO        & uncut & 24 knots \\ 
$z \in $ (-3,-2] & 0   &  0.69$^{**}$ &     0        &  0.04        & 0     &  0.35 $^{**}$ & 0     &  0.04 \\ 
$z \in $ (-2,-1] & 0   &  0.89$^{**}$ &     0        &  0.03        & 0     &  0.52 $^{**}$ & 0     &  0.04 \\ 
$z \in $  (-1,0] & 0   &  0.82$^{**}$ &     1$^{**}$ &  0.28$^{**}$  & 0     &  0.56$^{**}$  &  0     &  0.05 \\ 
$z \in $   (0,1] & 0.01&  0.91$^{**}$ &     1        &  0.79        & 0.004 &  0.50 $^{**}$ & 0     &   0.05\\ 
$z \in $   (1,2] & 0   &  0.89$^{**}$ &     1        &  0.99        & 0     &  0.21 $^{**}$ & 0     &   0.05\\ 
$z \in $   (2,3] & 0   &  0.95$^{**}$ &     1        &  0.94        & 0     &  0.52 $^{**}$ & 0     &  0.04 \\ 
\multicolumn{9}{c}{$n=1000$} \\ 
& \multicolumn{4}{c}{Covariate 1} & \multicolumn{4}{c}{Covariates 2-10} \\ 
 Region         & BH   & Fused        & Cubic  & GAM         & BH    &  Fused & Cubic    & GAM  \\ 
                &     &  LASSO        & uncut  & 24 knots  &       &  LASSO & uncut    & 24 knots \\ 
$z \in $(-3,-2] & 0    &  0.22$^{**}$ & 0       &  0.0008     & 0.004 &  0.07 & 0          & 0.006  \\ 
$z \in $(-2,-1] & 0    &  0.13$^{**}$ & 0.01    &  0.0002     & 0.004 &  0    & 0          & 0.006  \\ 
$z \in $ (-1,0] & 0    &  0.12$^{**}$ & 1$^{**}$ &  0.08$^{**}$& 0    &   0.06 &  0          & 0.01  \\ 
$z \in $  (0,1] & 1    &  1          & 1       &  0.95       & 0.006 &  0    & 0          & 0.01 \\ 
$z \in $  (1,2] & 0.515&  1          & 1       &  1          & 0.001 &  0.05 & 0          & 0.01   \\ 
$z \in $  (2,3] & 1    &  1          & 1       &  1          & 0.002 &  0    & 0          & 0.01  \\ 
\end{tabular}
\end{center}
\caption{Independent errors simulation. Proportion of rejected null hypothesis for several methods.
For covariate 1 and $z>0$,  this proportion is the statistical power; otherwise, it is the type I error. ** indicates a type I error $> 0.05$.
The methods are applying Benjamini-Hochberg P-value adjustment
and fused LASSO to our degree 0 cut splines,
using uncut cubic splines in our Bayesian framework, and a GAM with 24 knots.
For fused LASSO the tuning parameters are set to minimize the extended Bayesian information criterion, further only estimates above 0.1 in absolute value are reported as significant}
\label{tab:simiid_extra}
\end{table}

\begin{figure}
\begin{center}
\begin{tabular}{cc}
Average $P(\beta_j(z) \neq 0 \mid y)$ & Proportion of rejected null hypotheses \\
\includegraphics[width=0.5\textwidth]{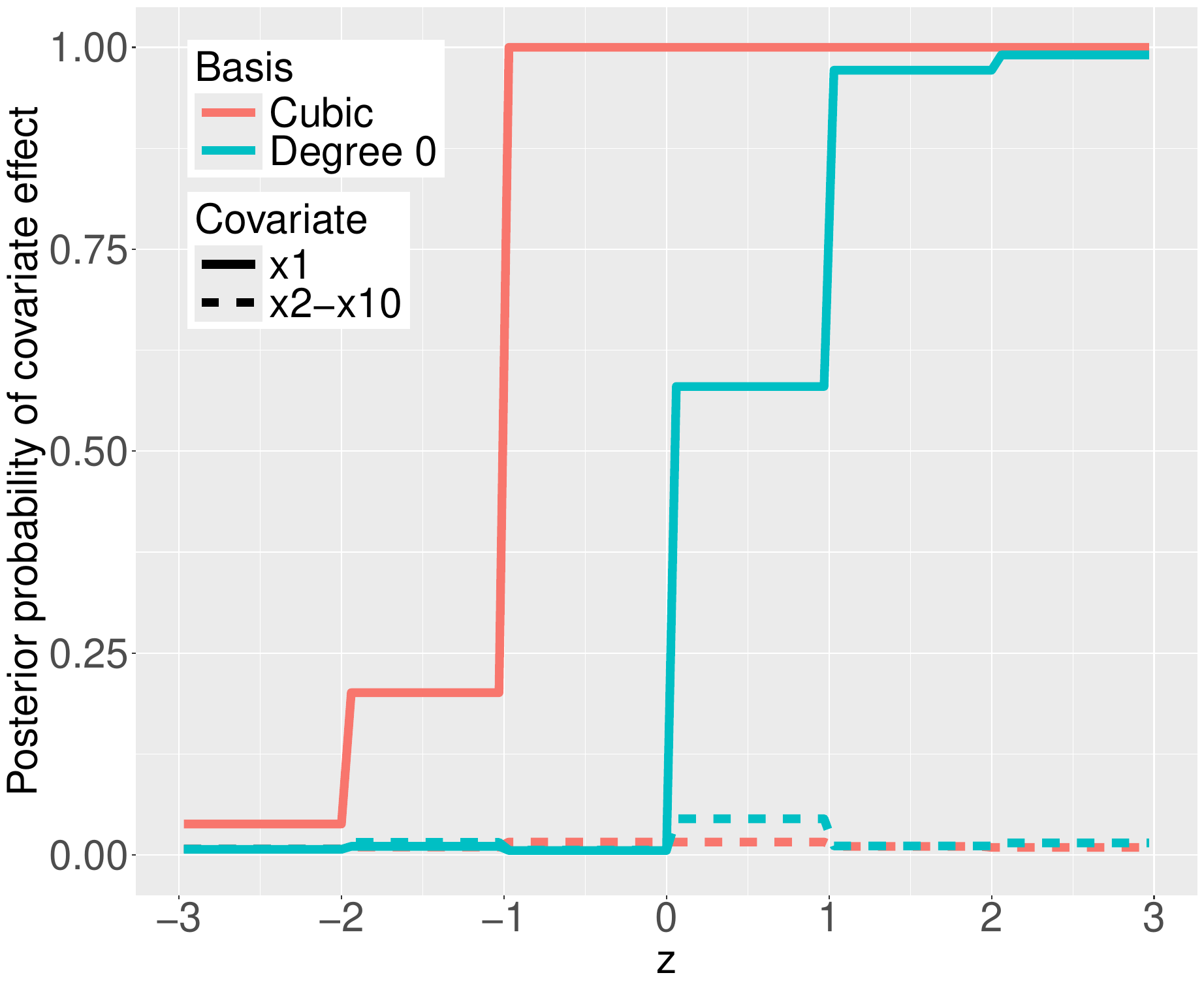} &
\includegraphics[width=0.5\textwidth]{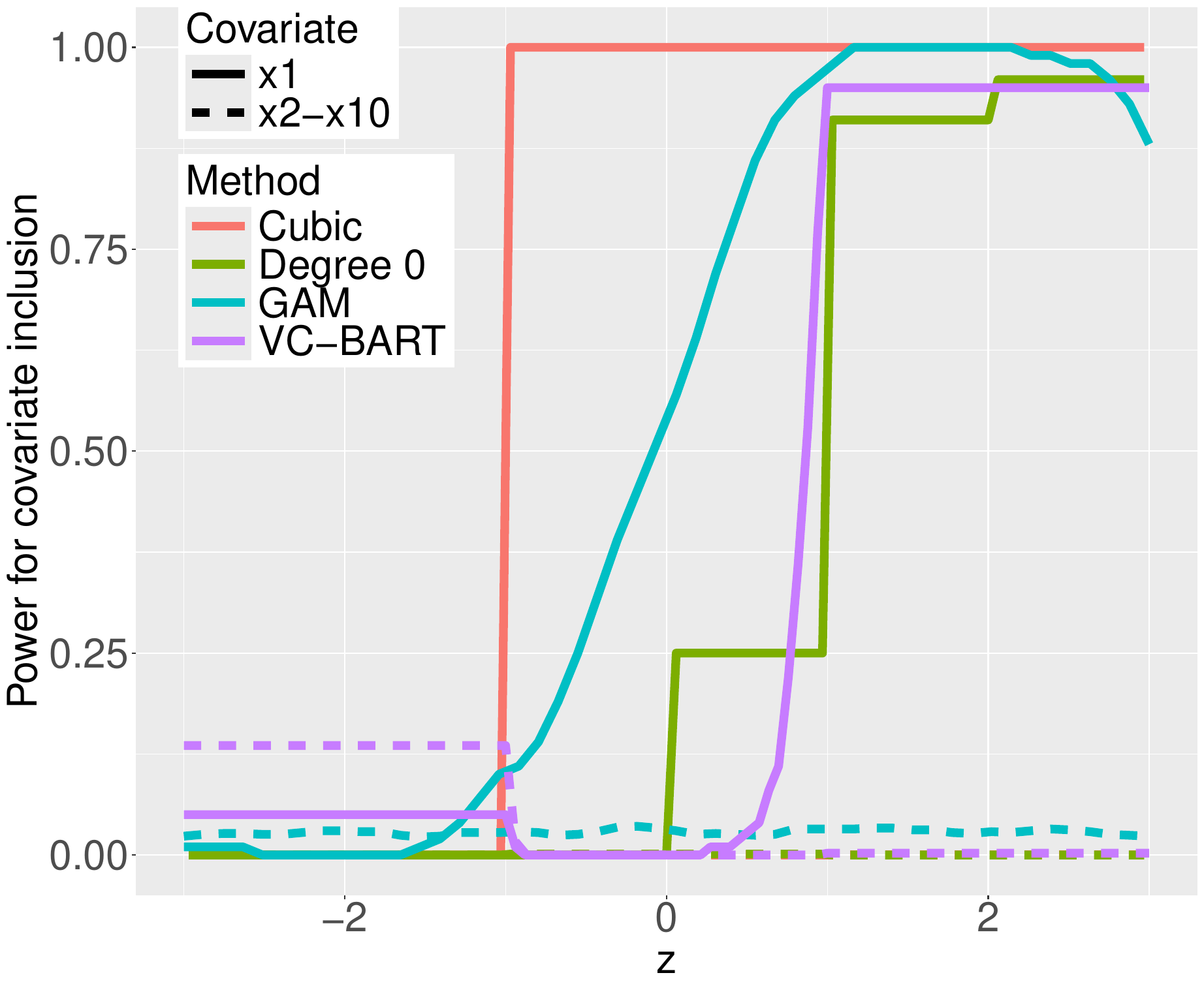} \\
\includegraphics[width=0.5\textwidth]{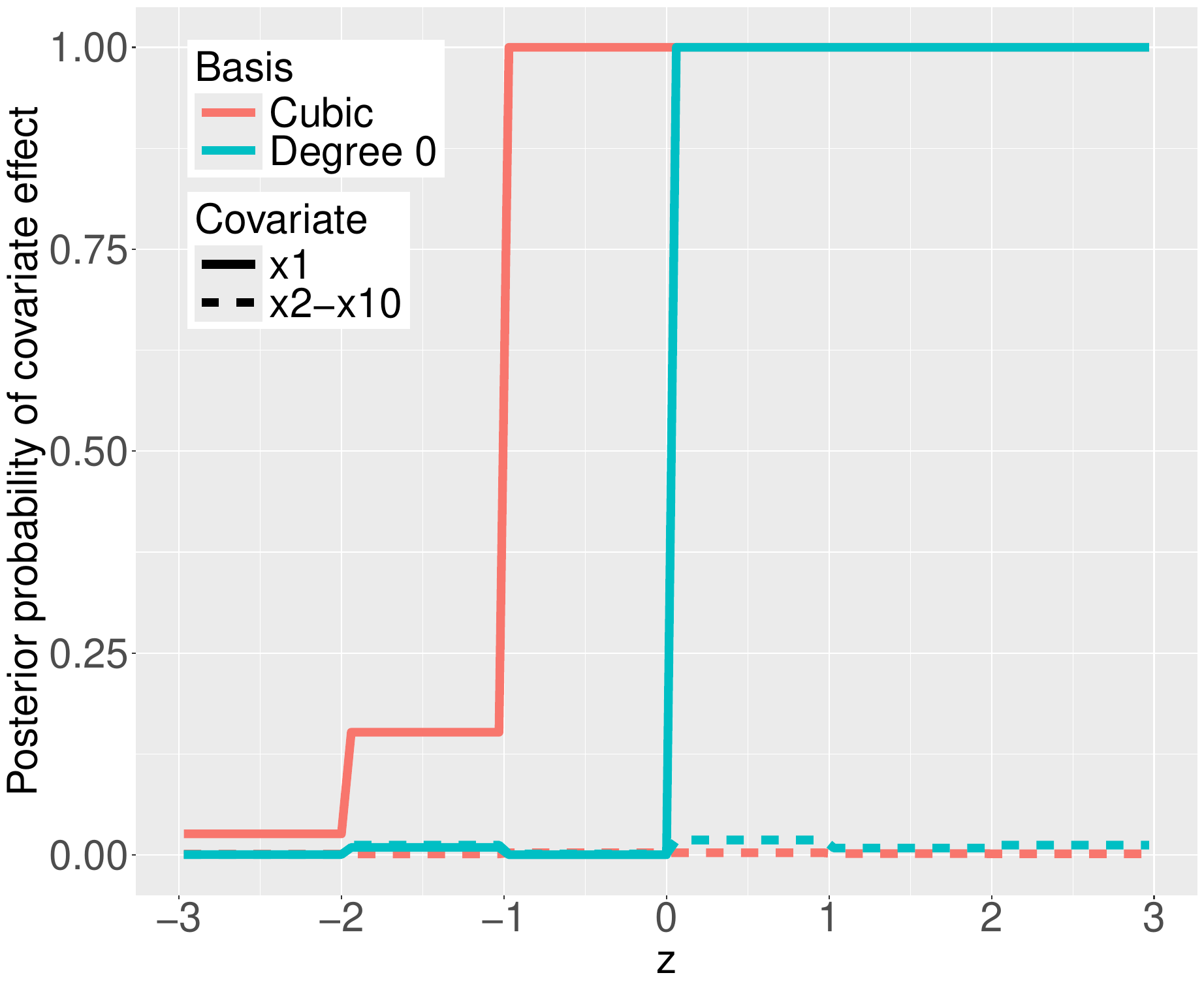} &
\includegraphics[width=0.5\textwidth]{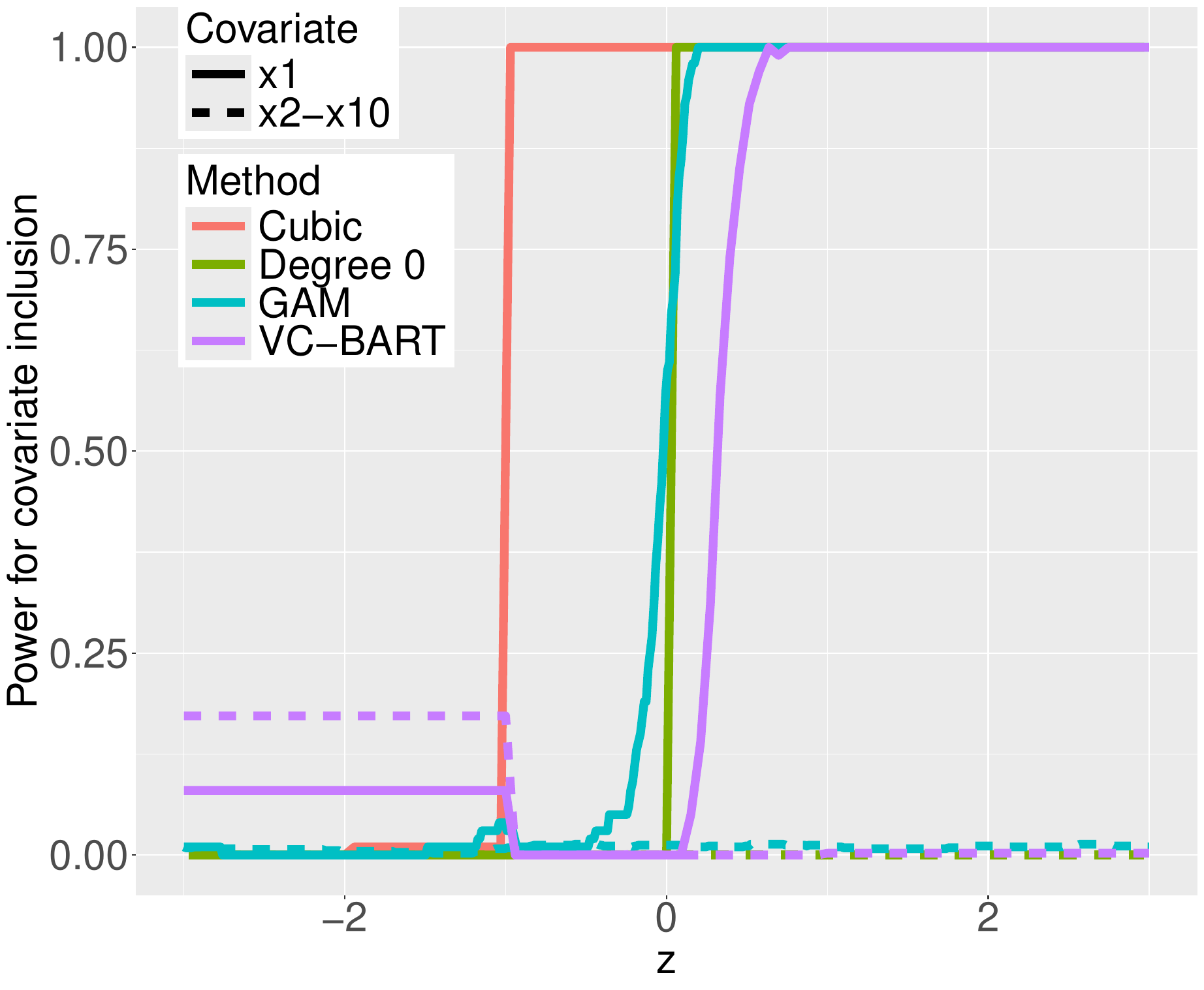}
\end{tabular}
\end{center}
\caption{Independent errors simulation. Posterior probability of local null test and proportion of rejected null hypotheses. Top: $n=100$. Bottom: $n=1,000$}
\label{fig:simiid}
\end{figure}

\begin{table}
\begin{center}
\begin{tabular}{lcccccc}  
         & Gaussian &  ICAR+ & Cubic    & VC-BART & GAM        & GAM \\ 
         & prior    &  prior & B-splines&         & 12 knots & 24 knots\\   
n=100    & 0.173    &  0.166 & 0.168    & 0.297   & 0.159      & 0.173  \\ 
n=1000   & 0.063    &  0.063 & 0.078    & 0.285   & 0.055      & 0.058 \\ 
\end{tabular}
\end{center}
\caption{Independent errors simulation. Root Mean Squared Error $E_F\left[ \hat{E}(y_i) - E_F(y_i) \right]^2$ for 0-degree cut splines, cubic B-splines, VC-BART  and GAM (R package mgcv) with $k=12$ and $k=24$ knots }
\label{tab:simiid_mse}
\end{table}

\begin{table}
\begin{center}
\begin{tabular}{ccc} 
                                                 & $n=100$  & $n=1,000$ \\ 
Gaussian shrinkage prior. Cut degree 0 (1 core)  &   8.5 sec &  5.5 sec  \\ 
Gaussian shrinkage prior. Cut degree 0 (3 cores) &   5.0 sec &  4.0 sec  \\ 
ICAR+ prior. Cut degree 0 (1 core)               &   53.3 sec &  20.3 sec  \\ 
ICAR+ prior. Cut degree 0 (3 cores)              &   31.4 sec &  12.3 sec  \\ 
VCBART          &  69.2 sec  & 7.05 min \\
GAM ($k=12$)             &  1.4 sec  & 4.2 sec\\
GAM ($k=24$)             &  1.3 sec  & 4.2 sec\\
\end{tabular}
\end{center}
\caption{Independent errors simulation. Run times for one simulation on an Ubuntu Linux desktop, 13th Gen Intel core i7 processor, 32Gb RAM}
\label{tab:cputime_simiid}
\end{table}

\subsection{Simulation with independent errors, misaligned knots}
\label{ssec:simulation_iid_misalign}

We extend the simulation results in Section \ref{ssec:simulation_iid} by considering a variation where all the knots in our methodology are misaligned with respect to the true changepoint where $\beta_1(z)$ transitions from zero to non-zero.
Specifically, all settings are equal to those of Section \ref{ssec:simulation_iid}, except that now we set a data-generating $\beta_1(z) = 0$ for $z \leq 0.15$ and $\beta_1(z) \neq 0$ for $z > 0.15$. Recall that $\beta_1(z)$ is the coefficient associated to a binary covariate, hence it defines two group means.
The top panel in Figure \ref{fig:simiid_unaligned} shows this data-generating means for both groups, along with the knots used in our multi-resolution analysis (using our default of 3 resolutions, 
The middle and bottom panels show the average posterior probability $P(\beta_j(z) \neq 0 \mid y)$ and the proportion of rejected local null hypotheses (using a cutoff of $P(\beta_j(z) \neq 0 \mid y) > 0.95$) as a function of $z$.

The results are very similar to those in Section \ref{ssec:simulation_iid}, with some nuisances.
First, note that our specified knots define regions that start at $z=0$ and include the true changepoint $z=0.15$, i.e. regions where $\beta_j(z) \neq 0$ for some $z$ in the region.
Therefore, our approach finds evidence for the existence differences are found within such a region.
That is, our posterior probabilities should be interpreted as there being evidence for differences somewhere in the region.

A second observation is that the average posterior probability and power for $z \in [0,1]$ are slightly lower than in Section \ref{ssec:simulation_iid}, where the knots for two resolutions were aligned with the true changepoint. Intuitively, since truly $\beta_j(z)=0$ for some $z$'s in the region, the statistical power decreases relative to a perfectly aligned region where $\beta_j(z) \neq 0$ for all $z$ in that region.

\begin{figure}
\begin{center}
\begin{tabular}{c}
\includegraphics[width=0.5\textwidth,height=0.4\textwidth]{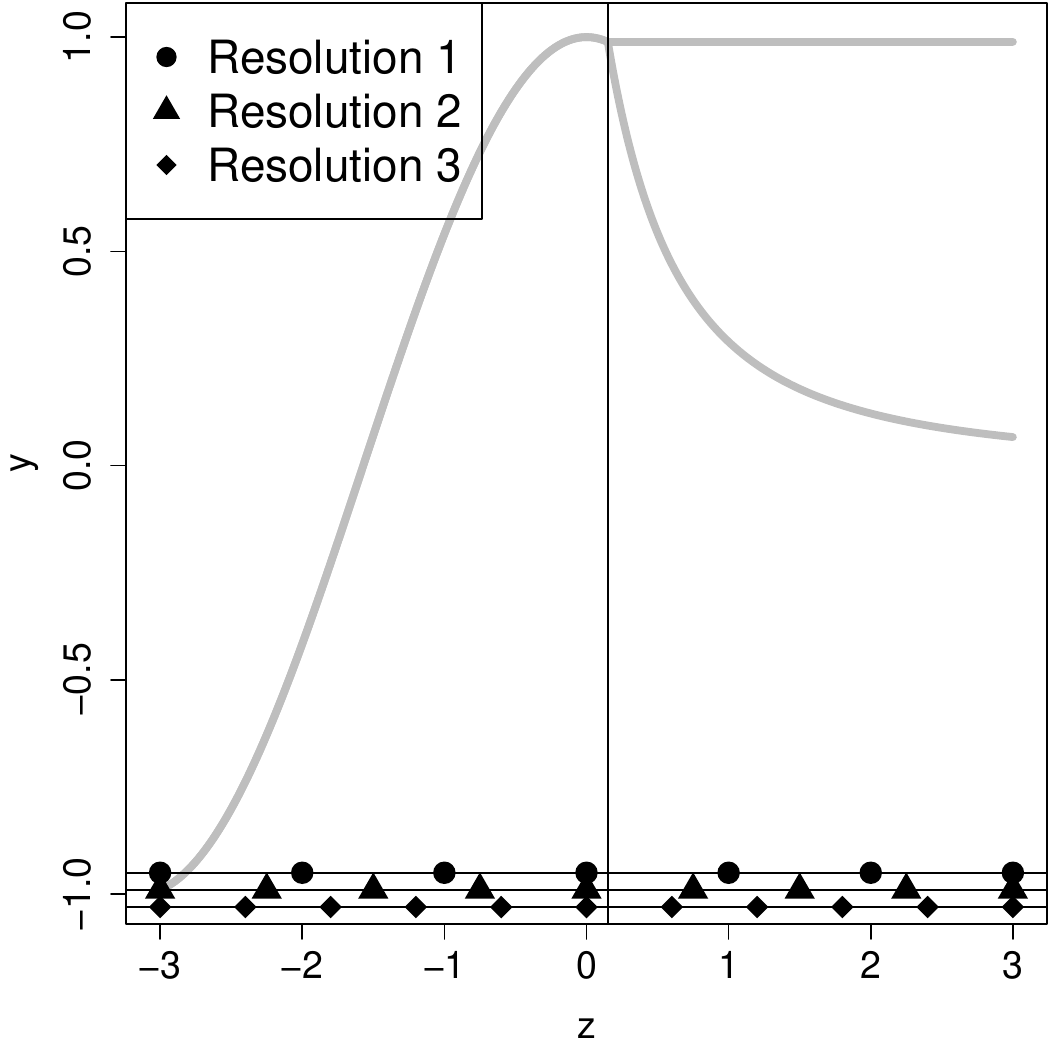} \\
\includegraphics[width=0.5\textwidth,height=0.4\textwidth]{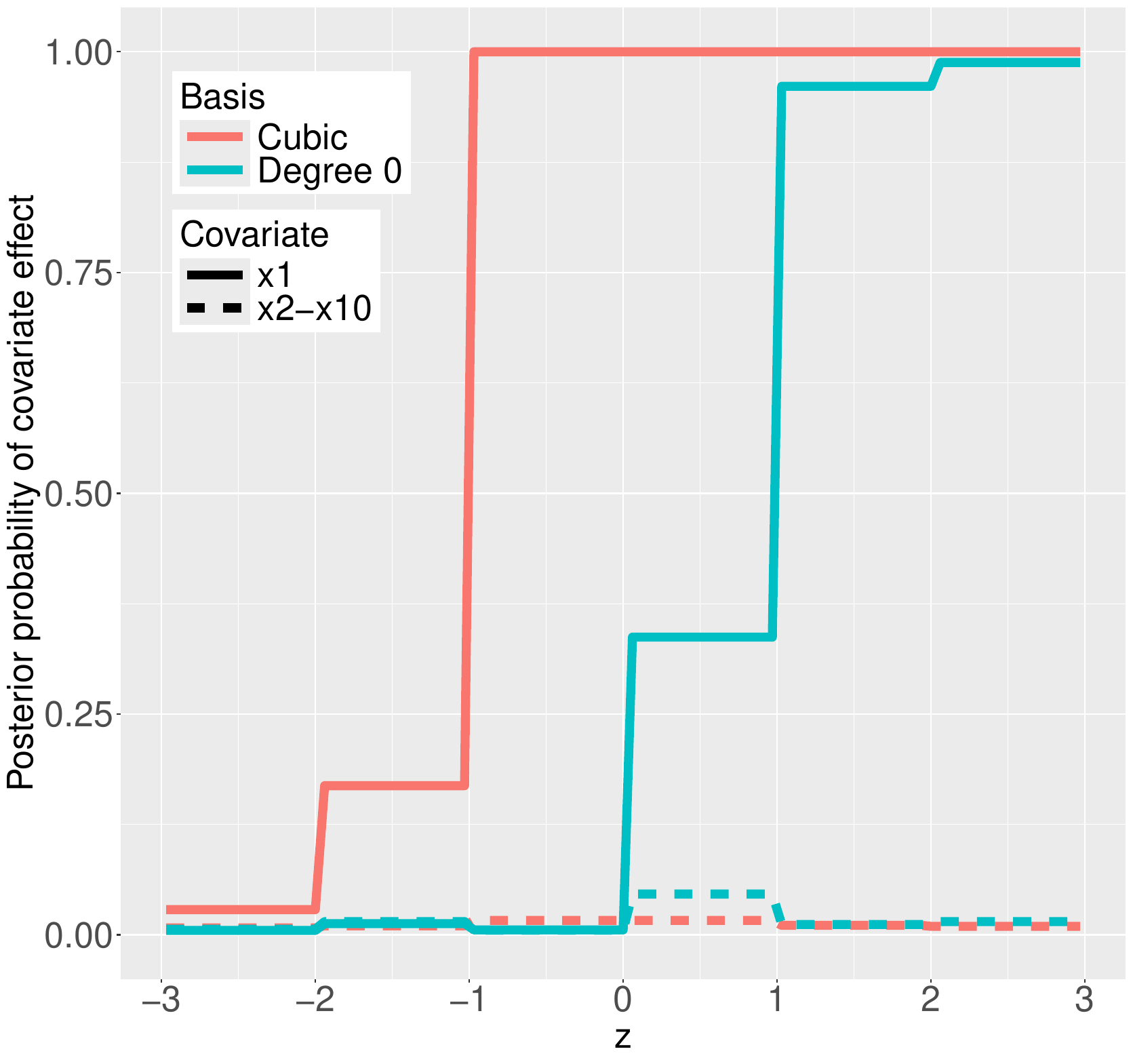} \\
\includegraphics[width=0.5\textwidth,height=0.4\textwidth]{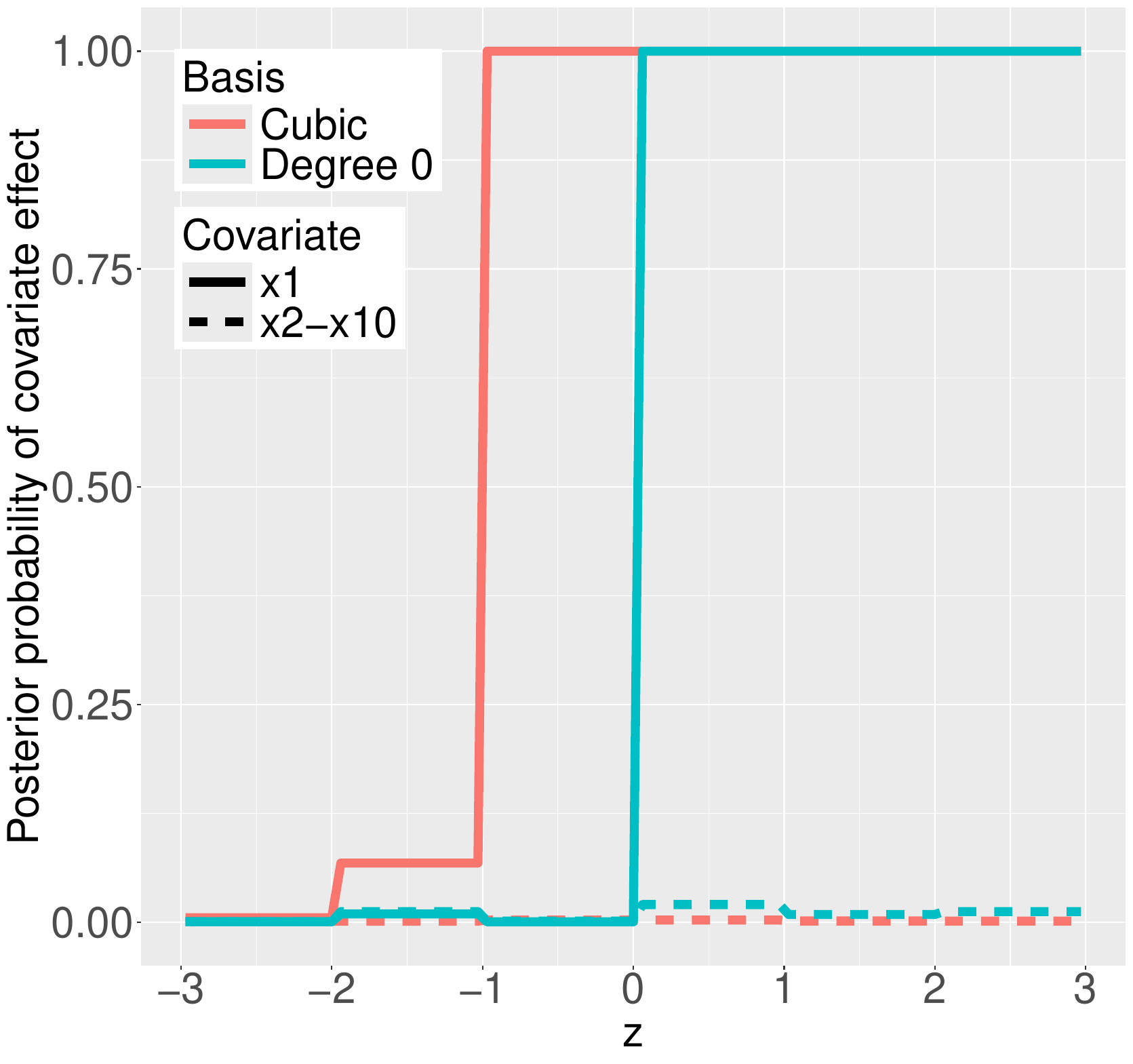} \\
\end{tabular}
\end{center}
\caption{ Independent errors simulation with unaligned knots. 
Data-generating truth (top), and average
posterior probability for a local effect for $n=100$ (middle) and $n=1,000$ (bottom) }
\label{fig:simiid_unaligned}
\end{figure}

\subsection{Simulation with independent errors, non equi-spaced knots}
\label{ssec:simulation_iid_nonequispaced}

We illustrate our methodology in a setting where some of the resolutions place non equi-spaced knots for $\beta_1(z)$, the effect of a binary covariate.
The data-generating truth and the knots for the four considered resolution levels are depicted in Figure \ref{fig:simiid_nonequispaced}, top left panel.
The simulation settings are as in Section \ref{ssec:simulation_iid_misalign}, except that here we consider a single covariate for simplicity, and an increased range of sample sizes $n = \{100,1000,5000\}$.
In particular, note that the knots are misaligned with the true change point at $z=0.15$.
Resolutions 1 and 3 are equi-spaced and set a total of 7 and 11 knots, respectively.
Resolutions 2 and 4 place the same knots as resolution 1 and 3 (respectively) for $z<0$, and twice as many knots for $z>0$.

The right panel in Figure \ref{fig:simiid_nonequispaced} shows the average posterior probability for a local covariate effect for each sample size.
The results are analogous to those in Figure \ref{fig:simiid_unaligned} where only equi-spaced knots were considered.

The bottom panel in Figure \ref{fig:simiid_nonequispaced} show the average posterior probability assigned to the 4 resolutions.
As expected, for the smaller sample size $n=100$ the more parsimonious resolution 1 is favored. For $n=1000$, resolutions 2 and 3 receive more posterior probability. Intuitively, this occurs because the larger numbers of knots in these resolutions allow approximating $\beta_1(z)$ better, and the sample size is large enough to estimate the corresponding parameters accurately enough.
Similarly, for $n=5000$ it is resolutions 2 and 4 which receive the most posterior probability. This is appealing, since these resolutions place more knots at $z>0$ where the data-generating $\beta_1(z)$ varies, and less knots for $z<0$ where $\beta_1(z)=0$ is constant.

\begin{figure}
\begin{center}
\begin{tabular}{cccc}
\multicolumn{2}{c}{Data-generating truth}      & \multicolumn{2}{c}{$P(\beta_1(z) \neq 0 \mid y)$} \\
\multicolumn{2}{c}{and knots for 4 resolutions}& \multicolumn{2}{c}{}\\
\multicolumn{2}{c}{\includegraphics[width=0.5\textwidth]{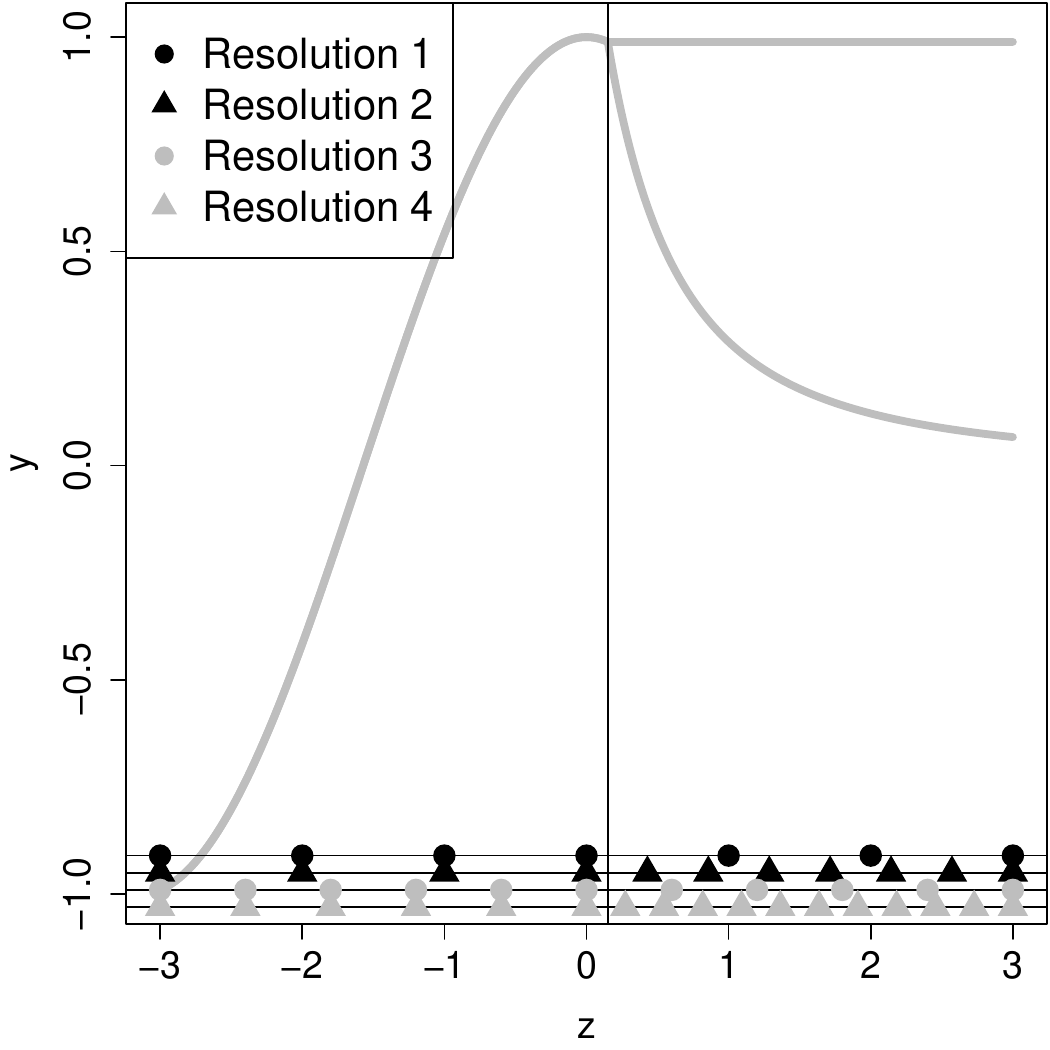}} &
\multicolumn{2}{c}{\includegraphics[width=0.5\textwidth]{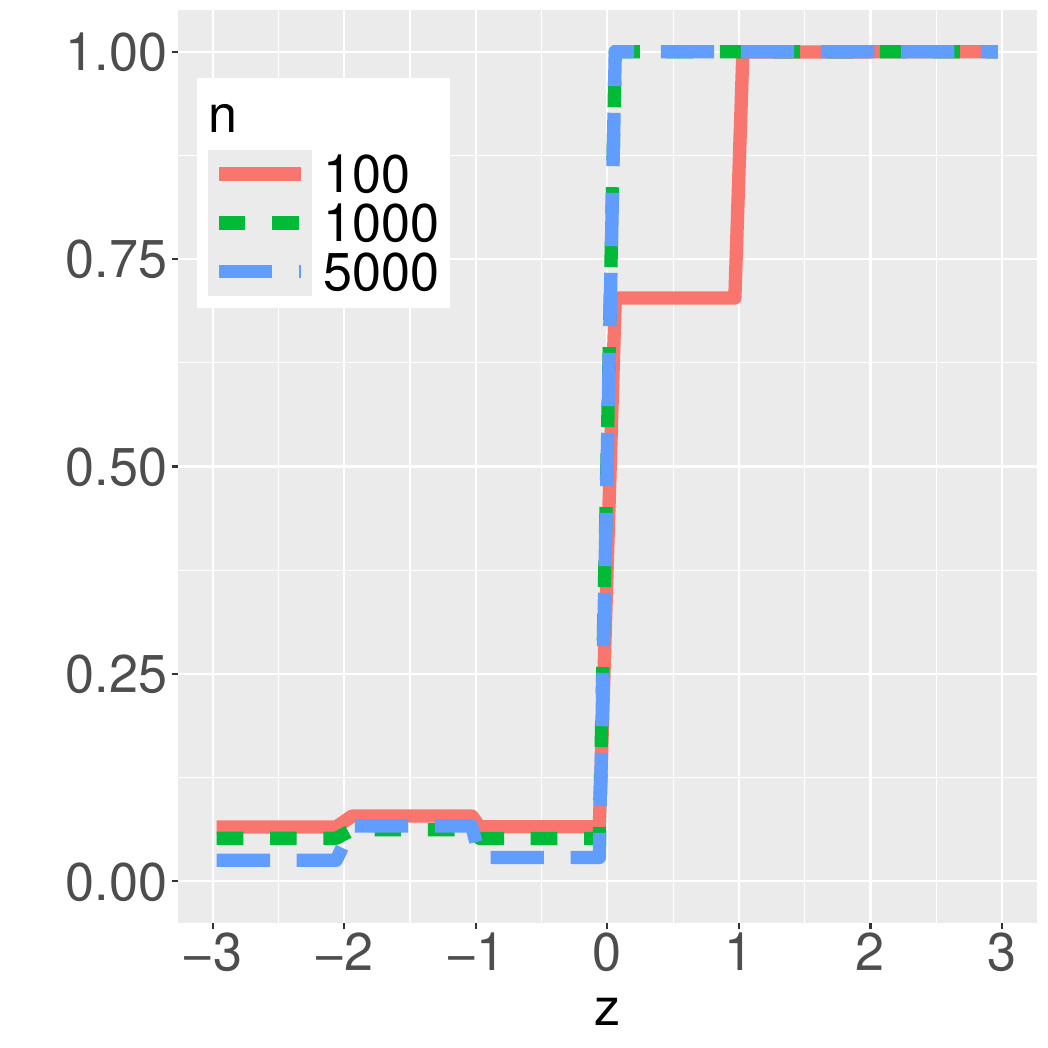}} \\
& & & \\ 
                             &  $n=100$ &  $n=1000$ &  $n=5000$ \\ 
   Resolution 1  &  0.86    &  0.21     &  0.00     \\ 
   Resolution 2  &  0.01    &  0.30     &  0.55     \\ 
   Resolution 3  &  0.13    &  0.49     &  0.03     \\ 
   Resolution 4  &  0.00    &  0.00     &  0.42     \\  
\end{tabular}
\end{center}
\caption{ Independent errors simulation with non-equispaced knots.
Top left: data-generating truth and four resolutions. Resolutions 1 and 3 have the same number of knots for $z<0$ and $z>0$.
Resolutions 2 and 4 have twice as many knots for $z>0$.
Top right: average $P(\beta_1(z) \neq 0 \mid y)$ for $n=100$, $n=1,000$ and $n=5,000$.
Bottom: average posterior probability assigned to each resolution
 }
\label{fig:simiid_nonequispaced}
\end{figure}



\subsection{Simulation with independent errors, bivariate $z$}
\label{ssec:simulation_iid_bivar}

\begin{figure}
\begin{center}
\includegraphics[width=0.6\textwidth]{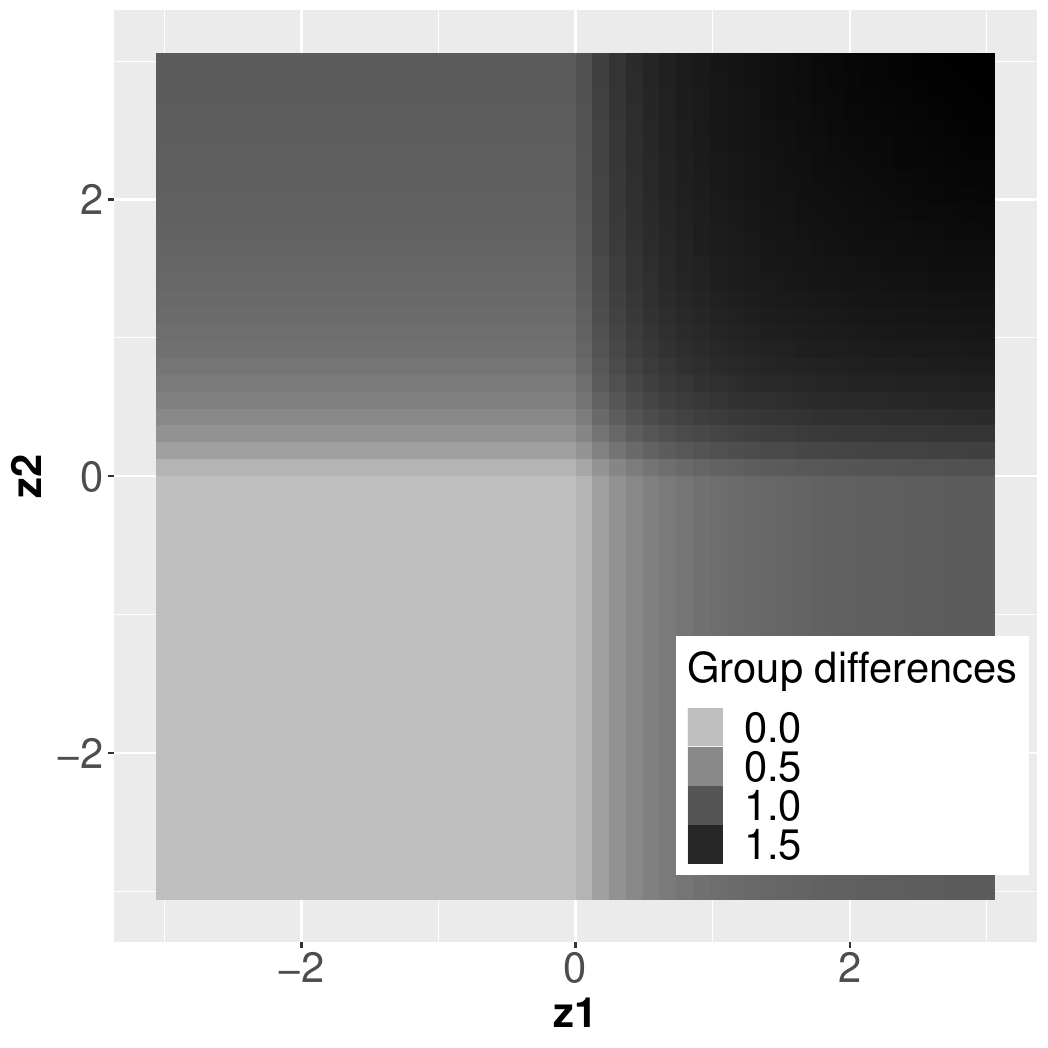}
\end{center}
\caption{True group mean differences $E_F(y_i \mid x_{i1}=1, z_i) - E_F(y_i \mid x_{i1}=0, z_i)$ in bivariate $z$ simulation}
\label{fig:bspline_fit_cubic_bivar}
\end{figure}

\begin{figure}
\begin{center}
Covariate 1 \\
\includegraphics[width=1\textwidth, height=0.65\textwidth]{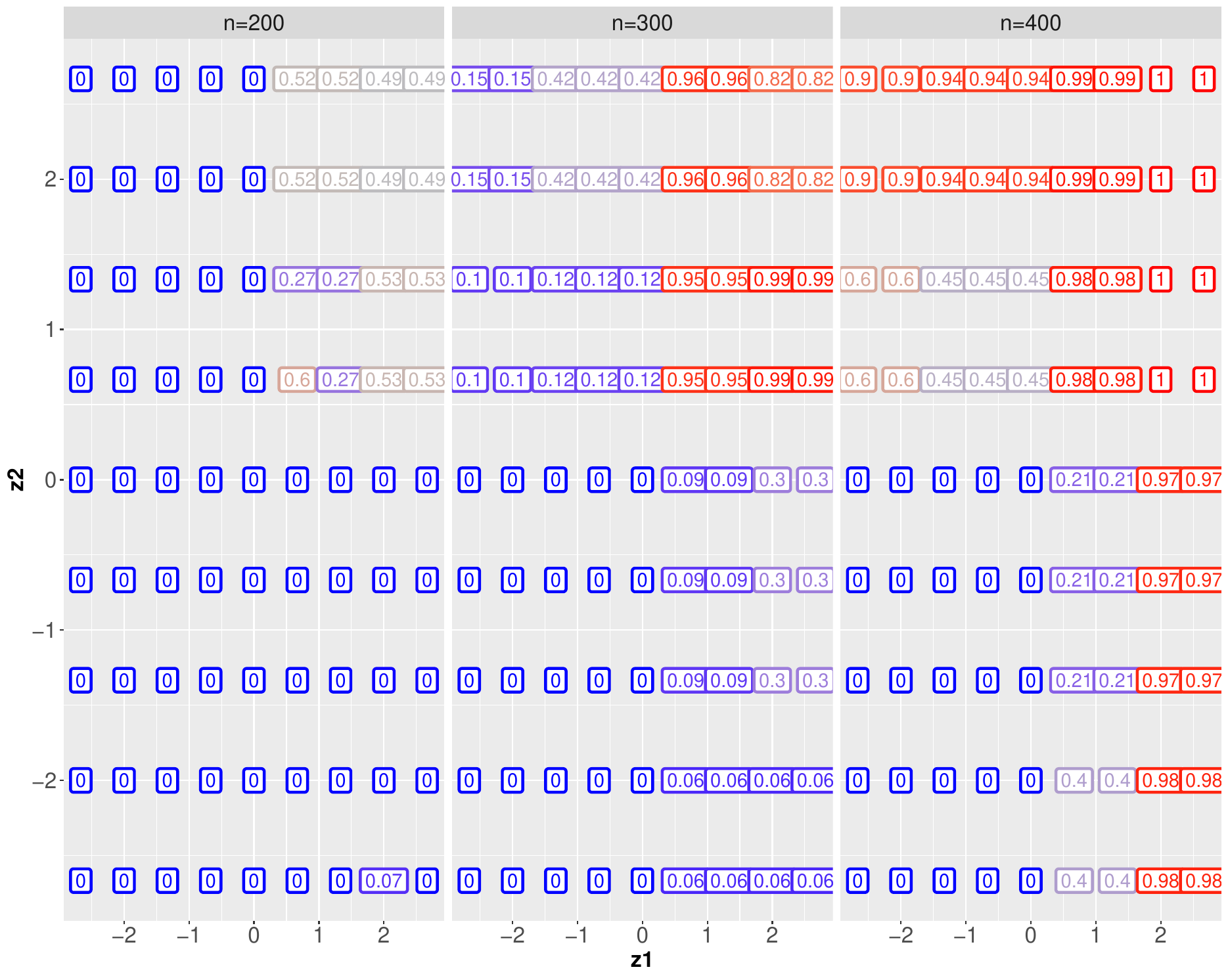} \\
Covariates 2-10 \\
\includegraphics[width=1\textwidth, height=0.65\textwidth]{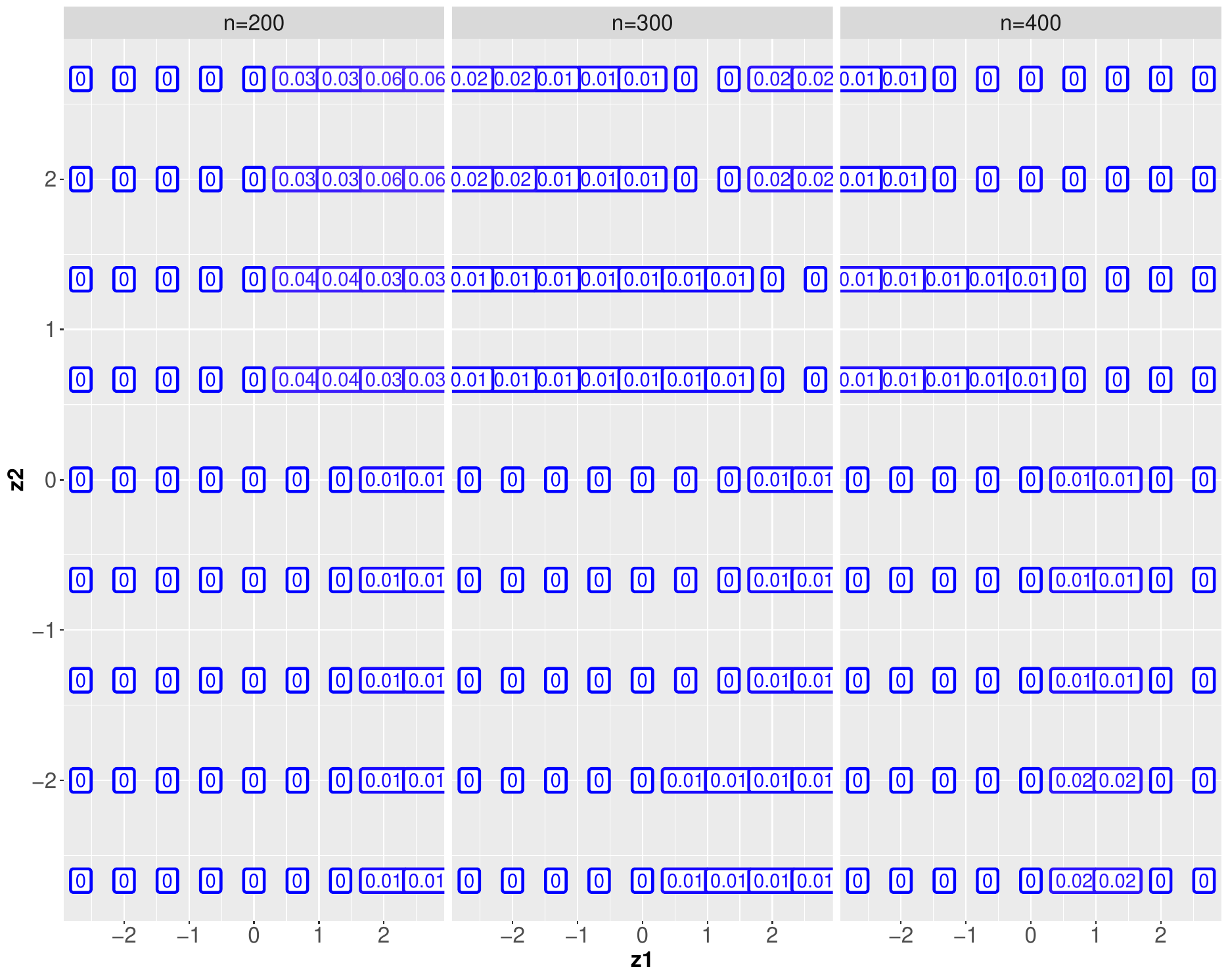}
\end{center}
\caption{Simulation example with bivariate $z$ for $p=10$, $n=200, 300, 400$.
Posterior probabilities for a local effect of covariate 1 (truly active if $z_1>0$ or $z_2>0$) and inactive covariates 2-10}
\label{fig:simiid_bivar}
\end{figure}

Figure \ref{fig:simiid_bivar} shows a simulation study with bivariate $z$ that extends that presented in Section \ref{sec:simulation_iid}. The results are averaged across 50 independent simulations.
 We used a multi-resolution analysis that placed 5, 10 and 15 knots in each of the two dimensions. The average run time was 20 seconds using 1 core and 10 seconds using 3 cores on an Ubuntu Linux desktop, 13th Gen Intel core i7 processor, with 32Gb RAM. 
As in Section \ref{sec:simulation_iid} we consider $p=10$ covariates where $x_{i1}$ is binary and the remaining covariates $x_{i2},\ldots,x_{ip}$ are Gaussian, correlated with $x_{i1}$, and truly have no effect on the outcome. 
The true mean is set such that covariate $x_{i1}$ has an effect whenever $z_{i1}>0$ or $z_{i2}>0$.
Specifically, the true mean for group $x_{i1}=1$ is
\begin{align}
E_F(y_i \mid x_{i1}=0, z_i)= 
\sum_{j=1}^2 \mbox{I}(z_{ij} \leq 0) \cos(z_{ij}) + \mbox{I}(z_{ij} \leq 0)
\nonumber
\end{align}
In contrast, for group $x_{i1}=0$ it is
\begin{align}
E_F(y_i \mid x_{i1}=0, z_i)= 
\sum_{j=1}^2 \mbox{I}(z_{ij} \leq 0) \cos(z_{ij}) + \mbox{I}(z_{ij} \leq 0) \frac{1}{(z_{ij}^2+1)^2}.
\nonumber
\end{align}
That is, the setting is akin to that in the top panels of Figure \ref{fig:splinefit}, with bivariate $z$. 
Figure \ref{fig:bspline_fit_cubic_bivar} shows the group differences $E_F(y_i \mid x_{i1}=1, z_i) - E_F(y_i \mid x_{i1}=0, z_i)$.
We consider sample sizes $n \in \{200, 300, 400\}$, which suffice to illustrate the transition from a low- to a high-power setting.

The top panel in Figure \ref{fig:simiid_bivar} show that for $n=200$ the average marginal posterior probabilities for $x_{i1}$ are nearly zero, except for the region $(z_{i1}>0, z_{i2}>0)$, where the group differences are larger.
For $n=300$ and $n=400$ said posterior probabilities grow closer to 1 whenever there are truly group differences, but remain close to 0 for $(z_{i1}<0, z_{i2}<0)$ where there are no group differences.
The bottom panel shows how for covariates 2-10 the marginal probabilities remain close to zero for all $(z_{i1},z_{i2})$, indicating that the type I error probability is satisfactorily controlled.

\subsection{Functional data simulation}
\label{ssec:simulation_fda}

\begin{table}
\begin{center}
\begin{tabular}{ccccc} 
\multicolumn{5}{c}{$M=50$ individuals}        \\ 
& \multicolumn{2}{c}{Cut degree 0} & \multicolumn{2}{c}{VC-BART} \\
 Region         & Covariate 1 & Covariates 2-10 &  Covariate 1  &  Covariates 2-10  \\
$z \in $(-3,-2] & 0.06         & 0.080          &  0.60         &  0.58 \\ 
$z \in $(-2,-1] & 0.06         & 0.078          &  0.61         &  0.57 \\ 
$z \in $(-1,0]  & 0.056        & 0.072          &  0.61         &  0.58 \\ 
$z \in $(0,1]   & 0.73         & 0.084          &  0.82         &  0.56 \\ 
$z \in $(1,2]   & 0.873        & 0.070          &  0.97         &  0.56 \\ 
$z \in $(2,3]   & 1            & 0.057          &  0.98         &  0.54 \\ 
\multicolumn{5}{c}{$M=100$ individuals} \\ 
& \multicolumn{2}{c}{Cut degree 0} & \multicolumn{2}{c}{VC-BART} \\
 Region          & Covariate 1   & Covariates 2-10 &  Covariate 1 &  Covariates 2-10 \\
$z \in $ (-3,-2] & 0.05          & 0.036           &  0.64        &   0.58           \\
$z \in $ (-2,-1] & 0.045         & 0.038           &  0.72        &   0.52           \\ 
$z \in $ (-1,0]  & 0.0376        & 0.039           &  0.59        &   0.48           \\ 
$z \in $ (0,1]   & 0.99          & 0.036           &  0.89        &   0.52           \\ 
$z \in $ (1,2]   & 0.995         & 0.031           &  1           &   0.54           \\ 
$z \in $ (2,3]   & 1             & 0.027           &  1           &   0.51           \\ 
\end{tabular}
\end{center}
\caption{Functional data simulation. Proportion of rejected null hypothesis for our 0-degree cut orthogonal basis,  and VC-BART.  For covariate 1 and $z>0$ this is the statistical power, otherwise it is the type I error}
\label{tab:simfda_10covar}
\end{table}

\begin{figure}
\begin{center}
\begin{tabular}{cc}
\includegraphics[width=0.48\textwidth]{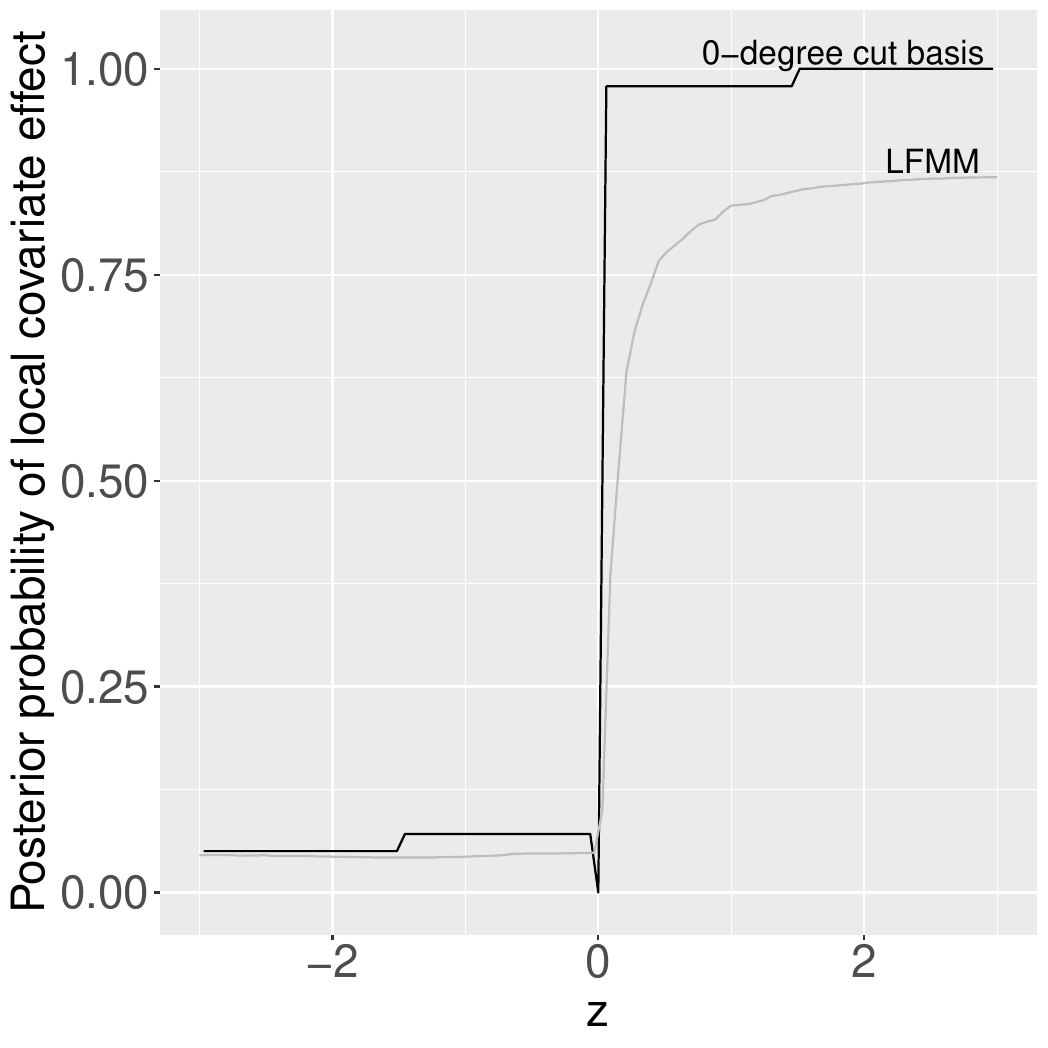} &
\includegraphics[width=0.48\textwidth]{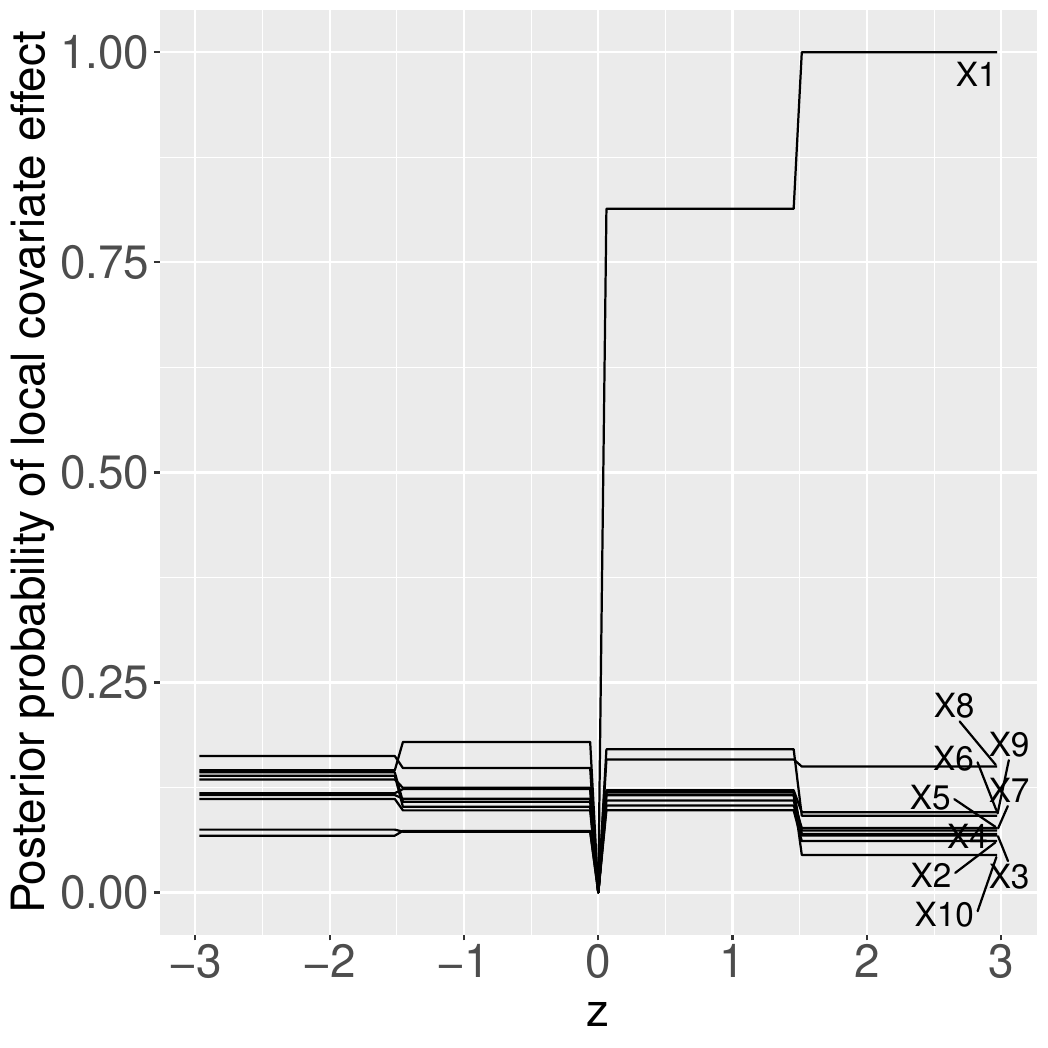}
\end{tabular}
\end{center}
\caption{Functional data simulation. Posterior probability of a local covariate effect when using a single covariate (left) and 10 covariates (right)}
\label{fig:pp_simfda}
\end{figure}

Figure \ref{fig:pp_simfda} displays a comparison of the two Bayesian methods, our framework and tensor model (referred to by the acronym LFMM by its authors) proposed by \cite{paulon:2023}, in terms of the posterior probability assigned to the existence of a local covariate effect as a function of $z$ (averaged across 100 simulations).
The left panel is based on fitting the model where one only consider the truly active covariate 1, whereas the right panel also includes truly spurious covariates 2-10. Recall that covariate 1 truly has an effect only for $z>0$.

\subsection{Salary data}

\begin{figure}
\begin{center}
\includegraphics[width=0.7\textwidth]{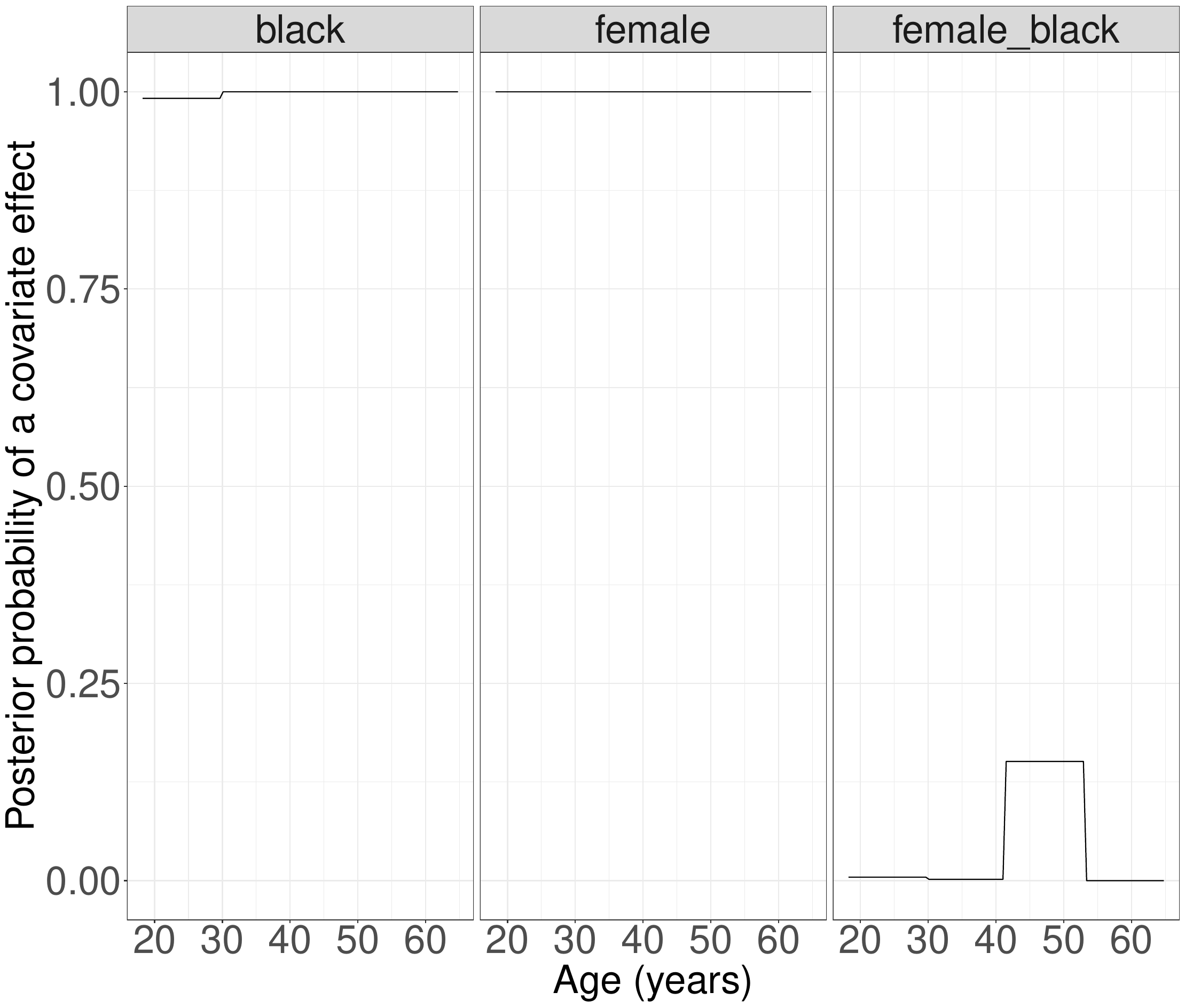}
\end{center}
\caption{Salary Data. Posterior probability of a salary gap associated with black race, sex and a black:sex interaction versus age, adjusted by occupation and worker class, and also including local effects for college, government, Hispanic race, and self-employment}
\label{fig:salary_female_race}
\end{figure}

\begin{figure}
\begin{center}
\begin{tabular}{cc}
\multicolumn{2}{c}{20 knots for baseline, 20 knots for local tests} \\
\includegraphics[width=0.5\textwidth]{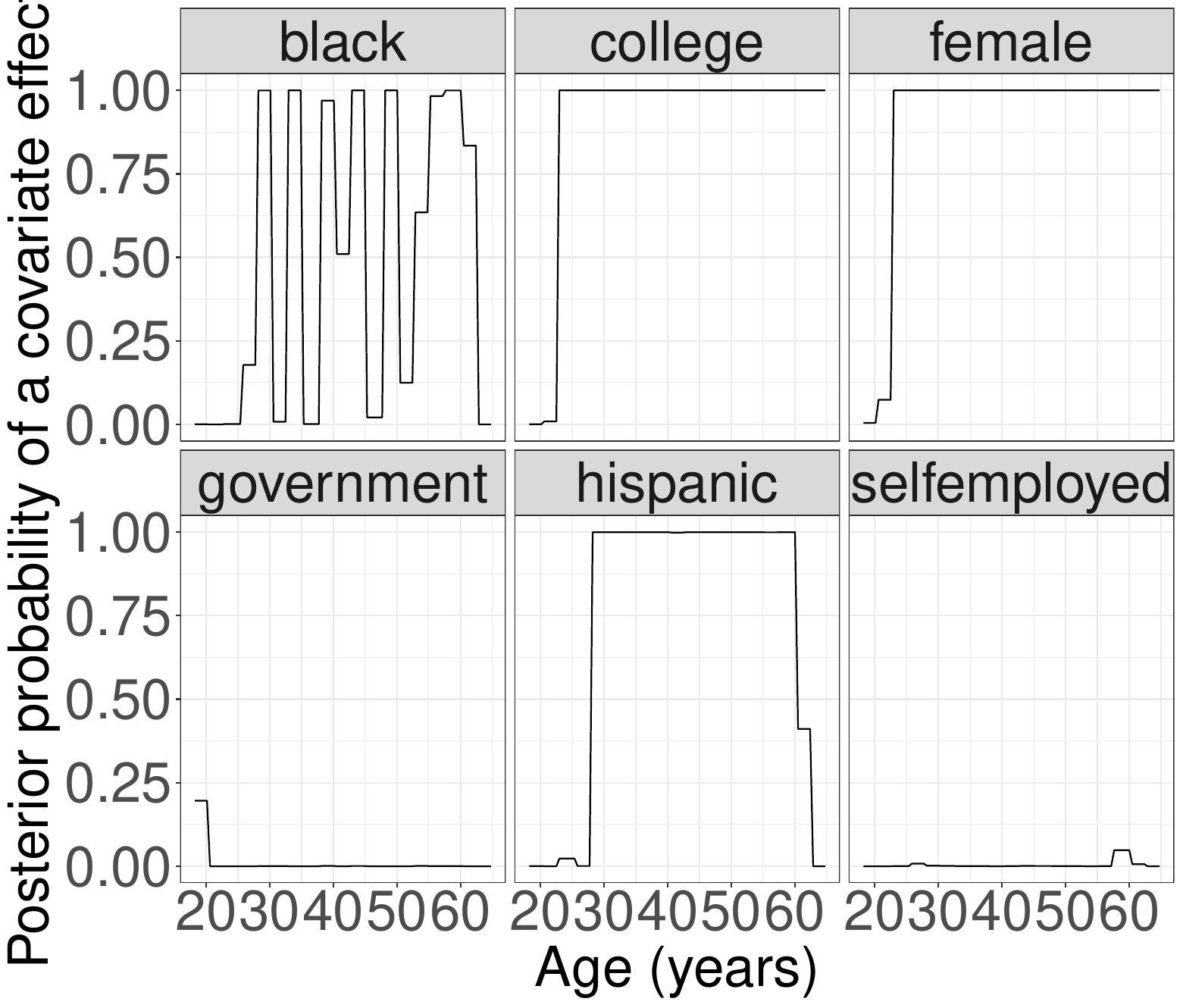} &
\includegraphics[width=0.5\textwidth]{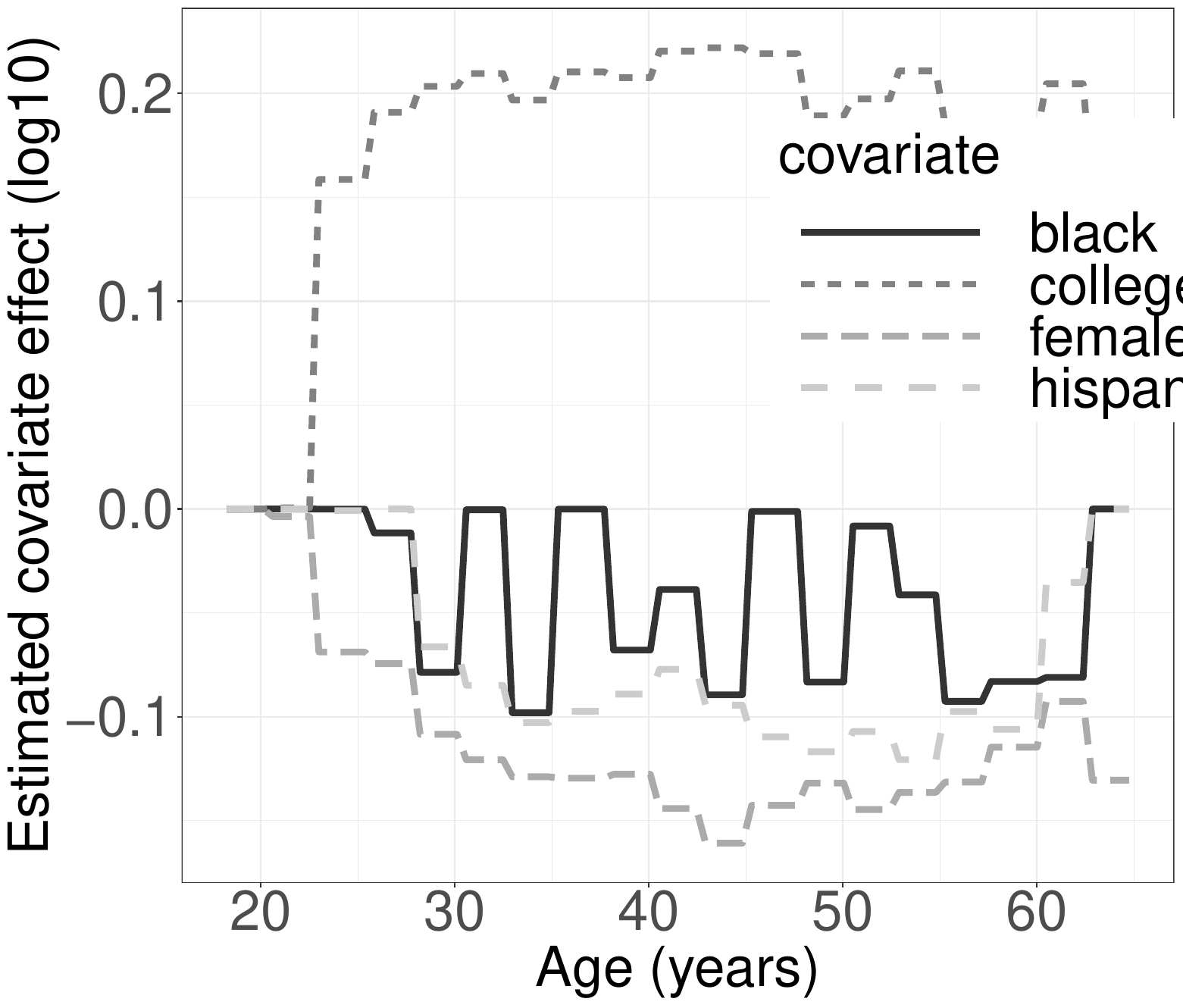} \\
\multicolumn{2}{c}{30 knots for baseline, 30 knots for local tests} \\
\includegraphics[width=0.5\textwidth]{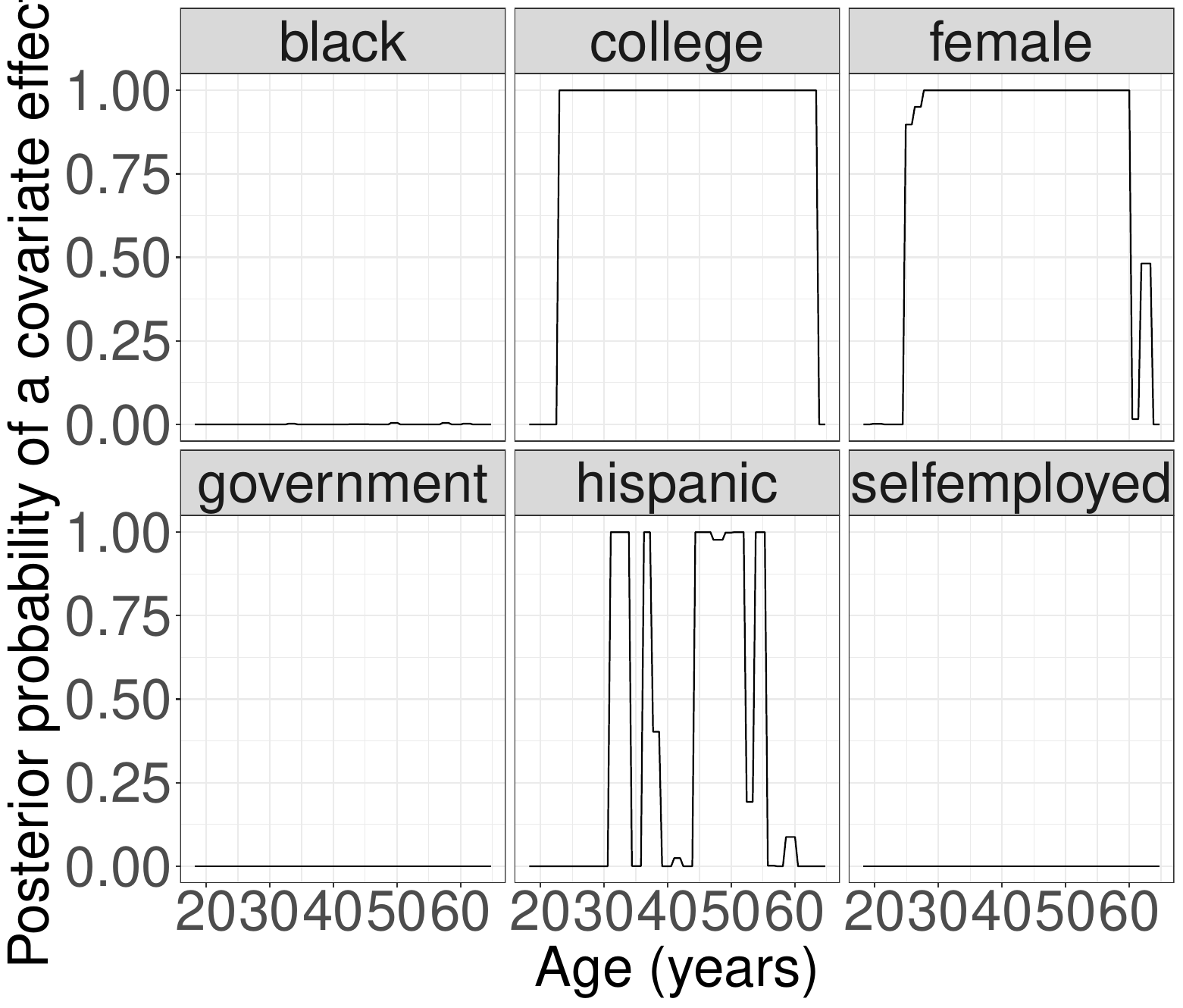} &
\includegraphics[width=0.5\textwidth]{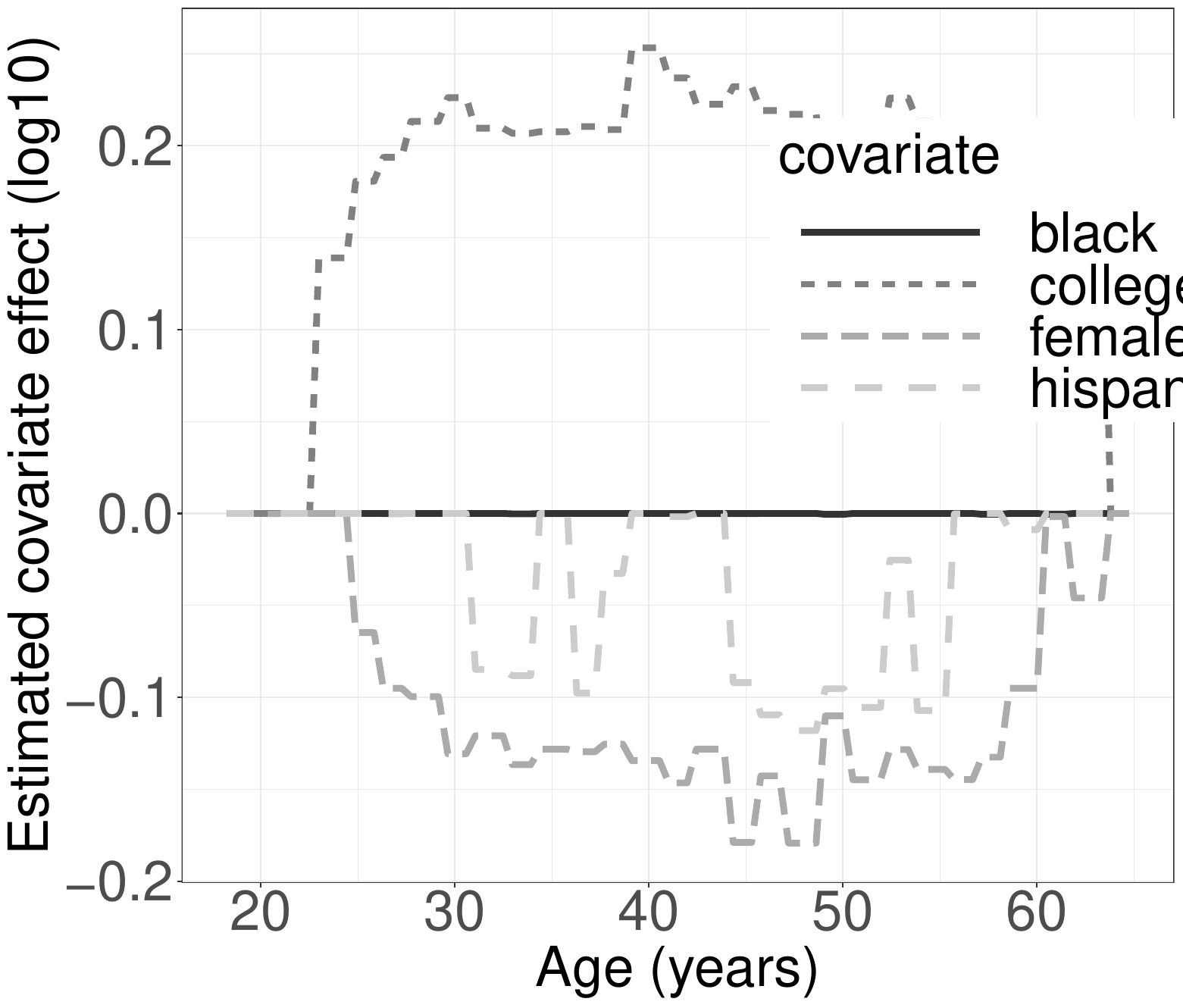}
\end{tabular}
\end{center}
\caption{Salary data results with fixed 20 (top) and 30 (bottom) fixed knots. 
Left: posterior probability of a salary gap associated with race, gender and college education versus age, adjusted by occupation and worker class. Right: corresponding estimated effect (log-10 scale)}
\label{fig:salary_suppl}
\end{figure}

The black race and female sex were found to be strongly associated to salary gap in our main analysis. To further explore this important source of potential salary discrimination, we repeated the analysis, now adding an interaction between black race and sex.
Figure \ref{fig:salary_female_race} shows that we obtained a small probability for the existence of said interaction, at all ages.

We also did additional analyses to assess the sensitivity of the results if one were to fix the number of knots, instead of using a multi-resolution analysis where one averages over resolutions.
Figure \ref{fig:salary_suppl} shows results for the salary data when fixing the number of knots to 20 (top panels) and 30 (bottom panels). The results suggest that, similar to what was observed in Figure \ref{fig:splinefit_suppl}, statistical power decreases when one adds more knots. This is apparent for covariates black and hispanic, which had the smaller estimated effects in our original analysis in Figure \ref{fig:salary} (where Bayesian model averaging was used to average over the uncertainty in the number of knots).
The results also suggest that there is no false positive inflation, covariates government and self-employed continue to show no local covariate effects at any age. Finally, covariates college and female are similar to our original analysis, in that they continue to be found to have a local effect on salary at almost all ages.

\subsection{Application to multi-electrode electrocorticography data}

\begin{figure}
	\begin{center}
		\begin{tabular}{cc}
		 	\includegraphics[width=0.45\textwidth]{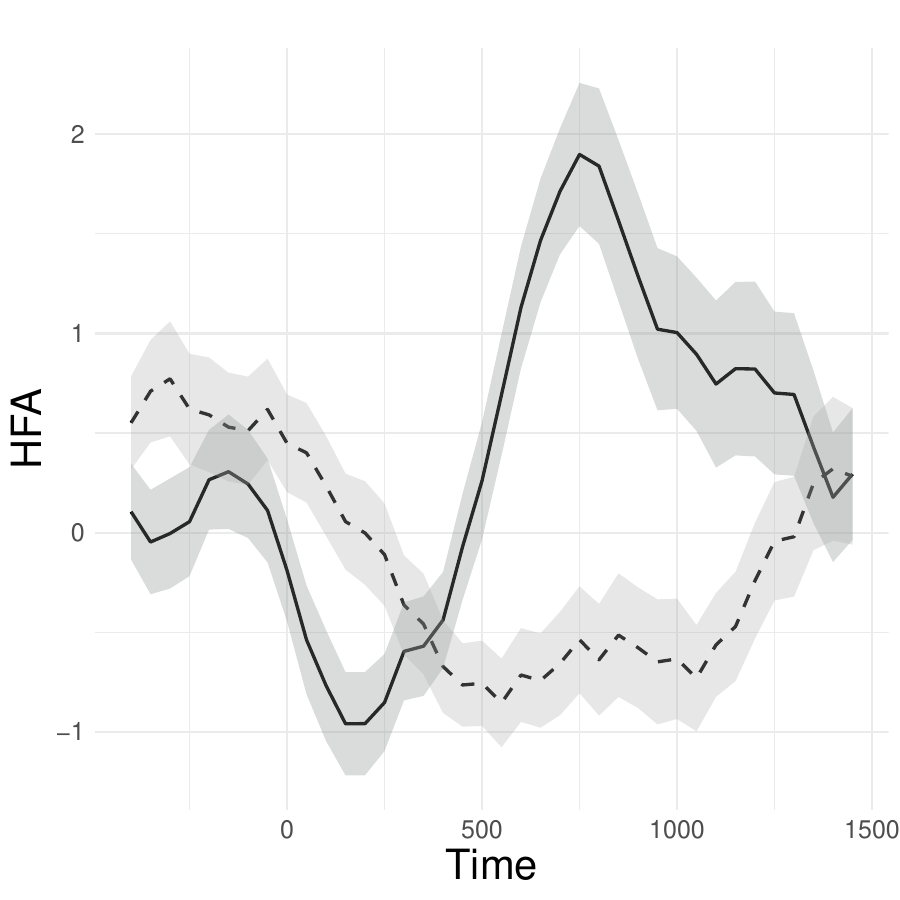} &
			\includegraphics[width=0.5\textwidth]{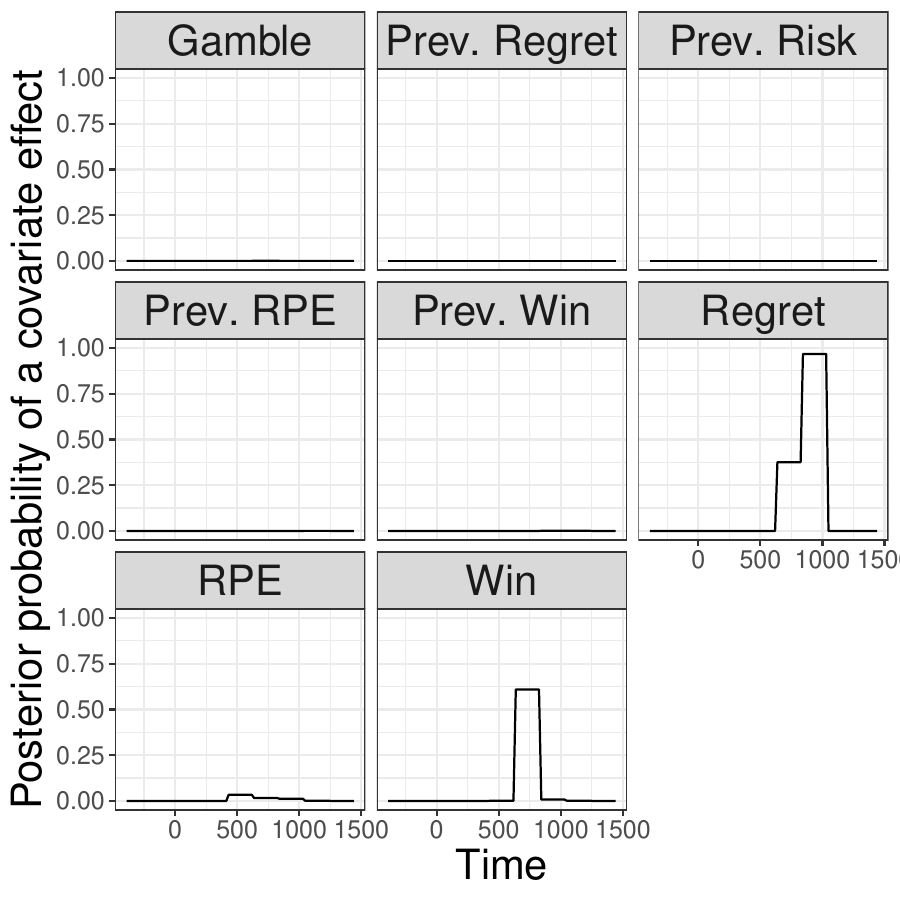} 
		\end{tabular}
	\end{center}
	\caption{ECoG data. Left: mean high-frequency activity signal averaged across trials where the gamble was won (solid line) or lost  (dashed line), and point-wise 80\% confidence intervals. 
	Right: posterior probability of local covariate effects on brain activity}
	\label{fig:ECoG}
\end{figure}

\begin{figure}[htbp]
\begin{center}
\begin{tabular}{cc}
\includegraphics[width=0.45\textwidth]{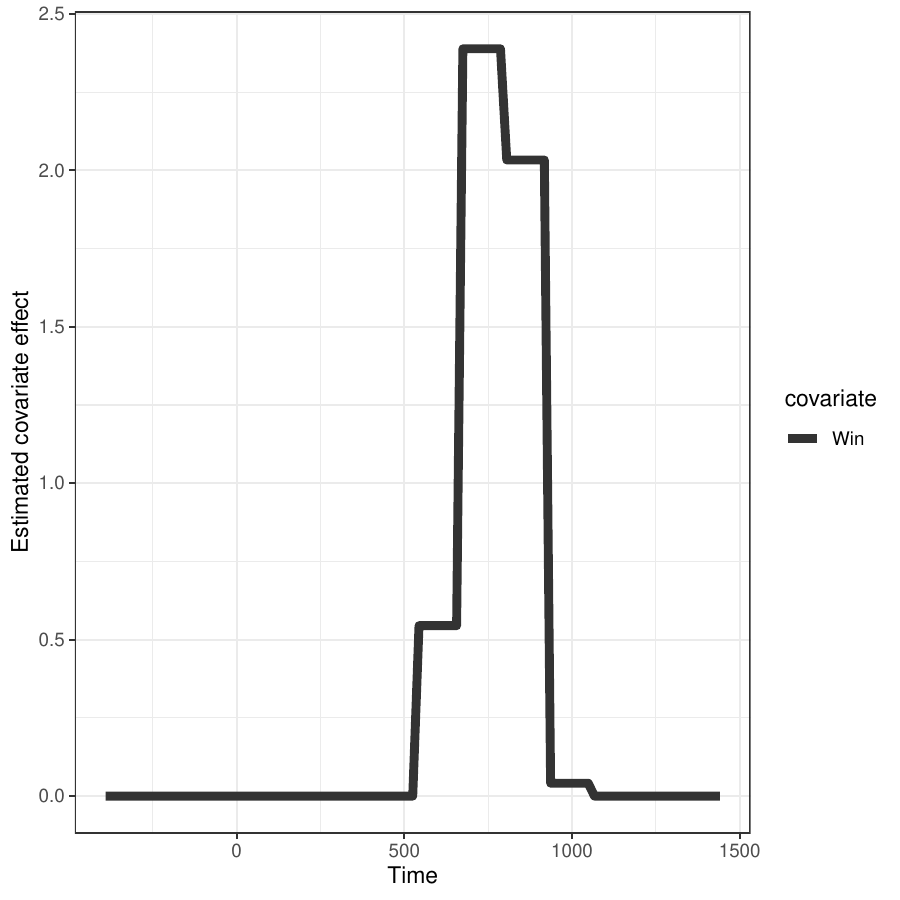} &
\includegraphics[width=0.45\textwidth]{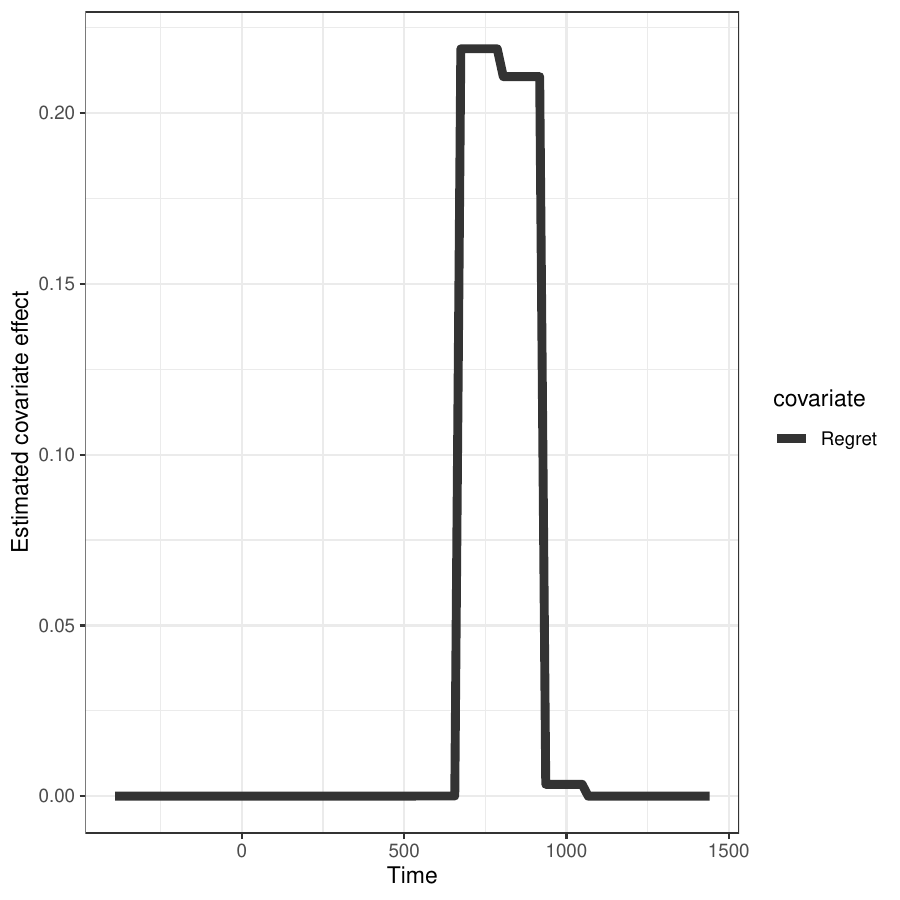} \\
Estimated effect of a Win/Loss over time &
Estimated effect of Regret over time \\
\includegraphics[width=0.45\textwidth]{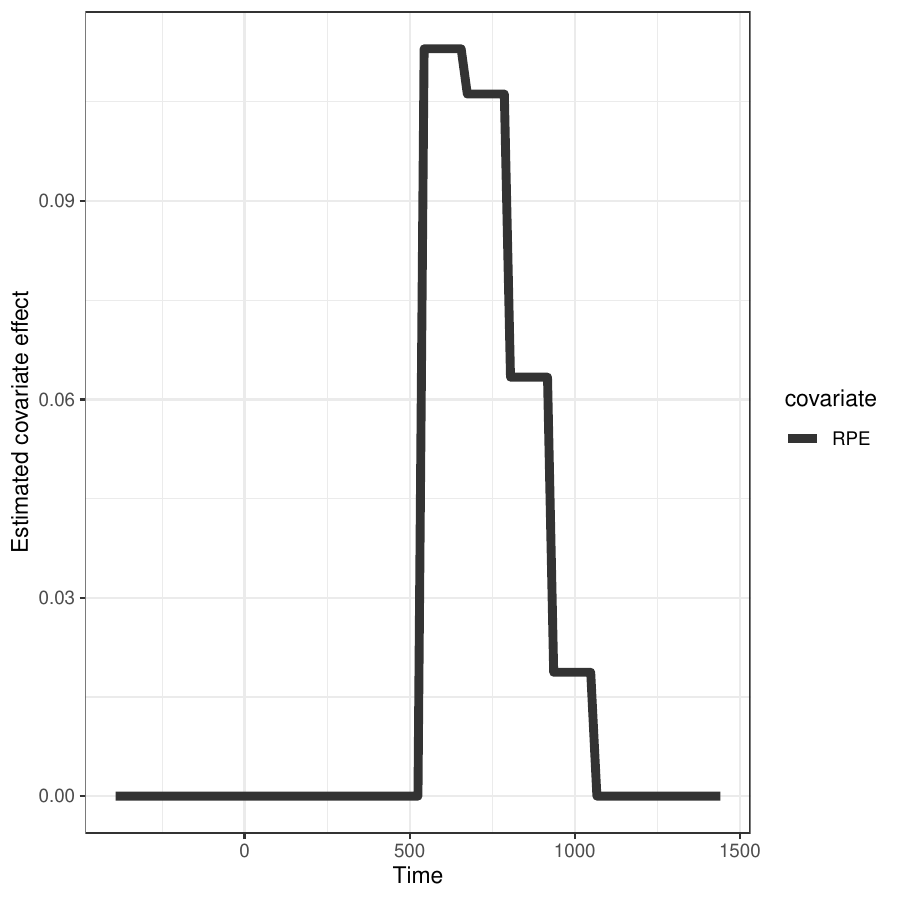} &
\includegraphics[width=0.45\textwidth]{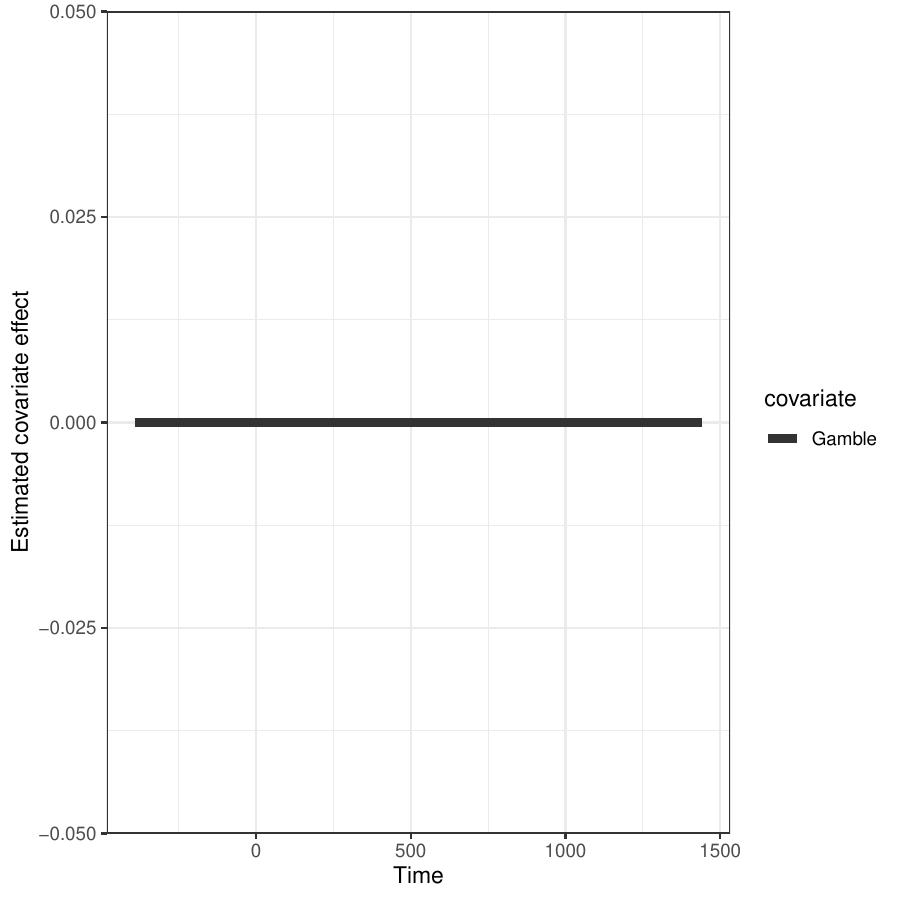} \\
Estimated effect of the reward prediction error (RPE) over time &
Estimated effect of Gamble over time \\
\end{tabular}
\end{center}
\caption{Estimated time-varying effects of the three outcome-related covariates (Win, Regret, RPE) and the choice-related covariate Gamble in single-variate analyses for the application of Section \ref{ssec:saez_data}. Only Win and Regret show high posterior probabilities of having an effect on brain activity in the time intervals}
\label{fig:EcoG_effect_singles}
\end{figure}

\begin{figure}
	\begin{center}
		\includegraphics[width=0.7\textwidth]{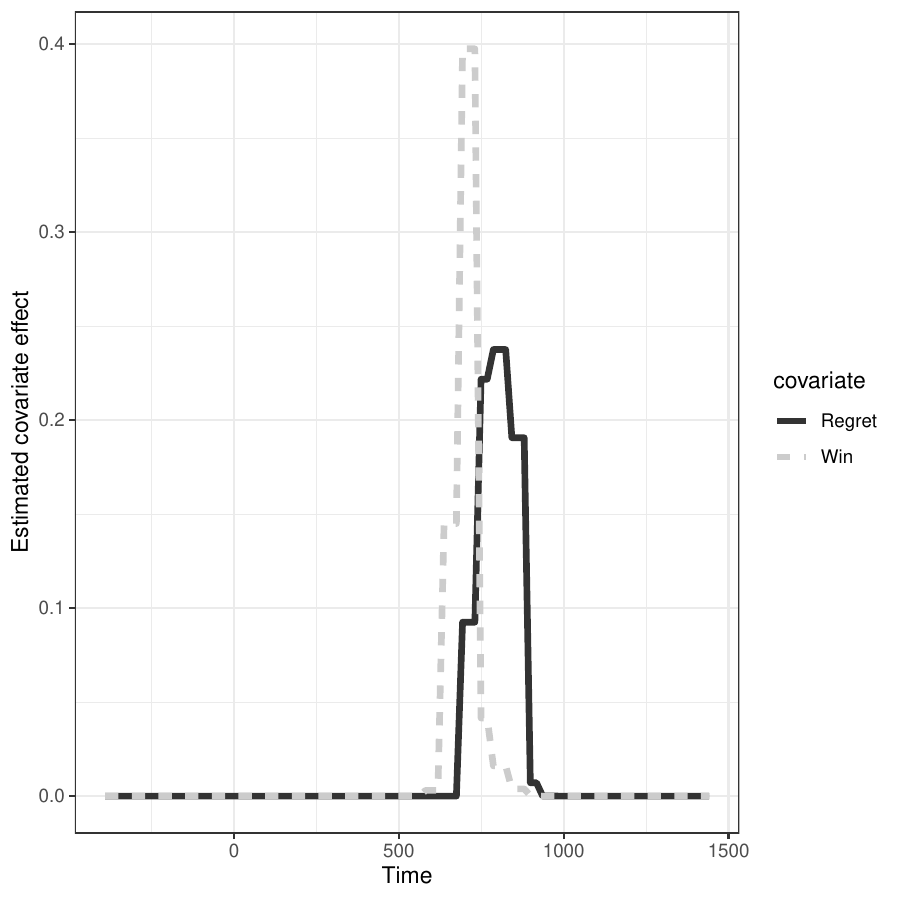}
	\end{center}
	\caption{Estimated time-varying effects of two outcome-related covariates (Win, Regret) in the multi-variable  analysis for the application of Section \ref{ssec:saez_data}.  The only relevant variable is Regret, with highest marginal probability of an effect $0.99$}
	\label{fig:EcoG_effect_multi}
\end{figure}

\bibliographystyle{Chicago}

\bibliography{references}
\end{document}